\documentclass[12pt]{article}
\input{preamble_JASA}
\usepackage{amsmath}
\usepackage{graphicx}
\usepackage{enumerate}
\usepackage{url} 
\usepackage{caption}
\usepackage{subcaption}
\newcommand{\blind}{1}

\begin{document}

\def\spacingset#1{\renewcommand{\baselinestretch}%
{#1}\small\normalsize} \spacingset{1}


\if1\blind
{
  \title{\bf Deep Learning of Multivariate Extremes \\via a Geometric Representation}
  \author{Callum J. R. Murphy-Barltrop\thanks{The authors
    would like to thank Phil Jonathan for providing the data, Ryan Campbell for access to code, and Lambert de Monte for helpful discussions. Reetam Majumder was supported by grants from the United States Geological Survey’s National Climate Adaptation Science Center (G24AC00197), and the National Science Foundation (DMS2152887). This work has made use of the resources provided by the Edinburgh Compute and Data Facility (ECDF) (http://ecdf.ed.ac.uk/)}\hspace{.2cm}\\
    Technische Universität Dresden \&\\
   Center for Scalable Data Analytics and Artificial Intelligence (ScaDS.AI)\\
    and \\
    Reetam Majumder \\
    Department of Mathematical Sciences, University of Arkansas\\
    and \\
    Jordan Richards \\
    School of Mathematics, University of Edinburgh}
  \maketitle
} \fi

\if0\blind
{
  \bigskip
  \bigskip
  \bigskip
  \begin{center}
    {\LARGE\bf Deep Learning of Multivariate Extremes \\via a Geometric Representation}
\end{center}
  \medskip
} \fi

\bigskip
\begin{abstract}
The study of geometric extremes, where extremal dependence properties are inferred from the deterministic limiting shapes of scaled sample clouds, provides an exciting approach to modelling the extremes of multivariate data. These shapes, termed limit sets, link several popular extremal dependence modelling frameworks.
Although the geometric approach is becoming an increasingly popular modelling tool for multivariate extremes, current inference techniques are limited to a low dimensional setting ($d\leq 4$), and generally require rigid modelling assumptions. In this work, we propose a range of novel theoretical results to aid with the implementation of the geometric extremes framework and introduce the first approach to modelling limit sets using deep learning. By leveraging neural networks, we construct asymptotically-justified and flexible semi-parametric models for extremal dependence of high-dimensional data. We showcase the efficacy of our approach by modelling the complex extremal dependencies between meteorological and oceanographic variables in the North Sea. 
\end{abstract}

\noindent%
{\it Keywords:} extremal dependence; geometric extremes; neural networks; deep learning
\vfill

\newpage
\spacingset{1.9} 
\section{Introduction}\label{sec:intro}

Multivariate extreme value theory seeks to understand relationships between the extremes of multiple variables. 
A wide variety of modelling approaches for multivariate extremes (or, equivalently, extremal dependence) have been proposed; classical approaches use the framework of regular variation \citep[see, e.g.,][]{Tawn1990,Rootzen2006,Einmahl2009}, but these models are restrictive in the forms of extremal dependence that they can capture \citep{Huser2024}. In particular, they can only provide accurate inference for random vectors exhibiting asymptotic dependence, i.e., when all variables are simultaneously extreme \citep{Ledford1997,Heffernan2004}. Assuming this form of extremal dependence for data is unrealistic in many applications, and numerous works have advocated against the use of regular variation models for environmental studies \citep[e.g.,][]{Opitz2016,Dawkins2018,Huser2024}. 

The first approach to modelling asymptotically independent data was provided by \citet{Ledford1996,Ledford1997}; see, also, \textit{hidden regular variation} \citep{Resnick2002}. 
Letting $\boldsymbol{X}_E:=(X_{E,1},\dots,X_{E,d})^T$ denote a random vector with standard exponential margins, \citet{Ledford1996} assumed that, as $u \to \infty$,
\begin{equation} \label{eqn:led_tawn}
    \Pr\left(\min_{i\in \mathcal{D}} \{X_{E,i}\} > u\right) = L(e^u)\exp\{-u/\eta\},
\end{equation}
where $\mathcal{D} := \{1,\hdots,d\}$, $L(\cdot)$ is a slowly varying function, i.e., $\lim_{u \to \infty}L(cu)/L(u) = 1$ for any constant $c > 0$, and $\eta \in (0,1]$ is termed the coefficient of tail dependence. Under asymptotic dependence, we have $\eta = 1$ and $\lim_{u \to \infty} L(u) > 0$, with other forms of extremal dependence arising when these conditions are not satisfied. In practice, model \eqref{eqn:led_tawn} is only applicable when all components of $\boldsymbol{X}_E$ are jointly large \citep[see, e.g., ][]{Wadsworth2017}. To overcome this limitation, \citet{Wadsworth2013} introduced the \textit{angular dependence function} (ADF), which generalises $\eta$. Consider any angle $\boldsymbol{w} := (w_1,\dots,w_d)^T \in \dsphere_+$ where $\mathcal{S}_+^{d-1} := \{ \boldsymbol{x} \in \RR_+^{d}: ||\boldsymbol{x}|| = 1\}$ is the strictly positive part of the unit $(d-1)$-sphere and $\|\cdot\|$ is the Euclidean norm; \citet{Wadsworth2013} assume that
\begin{equation} \label{eqn:wads_tawn}
    \Pr\left(\min_{i \in \mathcal{D}}\{X_{E,i}/w_{i}\}>u\right) = L(e^u;\boldsymbol{w})e^{-\lambda(\boldsymbol{w})u}, \; \; \lambda(\boldsymbol{w}) \geq \max(\boldsymbol{w}),
\end{equation}
as $u \to \infty$, where $L(\cdot \; ;\boldsymbol{w})$ is a slowly varying function and $\lambda(\cdot)$ denotes the ADF; the latter provides information about the joint tail of $\boldsymbol{X}_{E}$, and we have $\eta = \{\sqrt{d}\lambda(d^{-1/2},\hdots,d^{-1/2})\}^{-1}$. Model \eqref{eqn:wads_tawn} can capture both extremal dependence regimes, with asymptotic dependence implying the lower bound, $\lambda(\boldsymbol{w}) = \max(\boldsymbol{w})$, for all $\boldsymbol{w} \in \dsphere_+$. Loosely speaking, the angle $\boldsymbol{w}$ is the direction in $\RR_+^d$ for which the joint tail region in \eqref{eqn:wads_tawn} is defined. This model has been successfully applied in environmental applications by, e.g., \cite{Murphy-Barltrop2023}, \cite{Murphy-Barltrop2024b}, and \cite{Murphy-Barltrop2024c}.

Several of the aforementioned models introduced are defined for random vectors exhibiting standard margins with finite lower bounds, e.g., Pareto, Fr\'echet, or exponential, which limits the study of extremal dependence. In particular, taking random vectors with double tailed margins and applying such approaches (following marginal transformation), these models reduce the study of extremal dependence to data observed only in the positive orthant, $\RR^d_+$, of $\RR^d$; see Section \ref{sec:theory} for further discussion. This can be restrictive in practical applications where different regions of low joint probability mass may be of interest. Consequently, many recent works have introduced modelling approaches for data on standard Laplace margins. For example, \citet{Keef2013} extended the model of \citet{Heffernan2004} to Laplace margins, with the resulting framework providing greater flexibility and interpretability. Moreover, \citet{Mackay2023}, \citet{Simpson2024}, \citet{Papastathopoulos2024} and \citet{Murphy-Barltrop2024} demonstrate that models on Laplace margins permit evaluation of the joint tail behaviour of random vectors in all $2^d$ orthants of $\RR^d$. To demonstrate this, Figure \ref{fig:equidensity_contours} illustrates extreme isodensity contours for a bivariate Gaussian copula with correlation parameter $\rho = -0.5$ on both standard exponential and standard Laplace margins. One can observe that the Laplace representations offers a more detailed perspective of the entire dependence structure. Hereafter, we use $\boldsymbol{X}:=(X_1,\dots,X_d)^T$ to denote a $d$-dimensional random vector with standard Laplace margins, with distribution function $F_{\boldsymbol{X}}(\cdot)$ and continuous density function $f_{\boldsymbol{X}}(\cdot)$. 

\begin{figure}
    \centering
    \includegraphics[width=.7\textwidth]{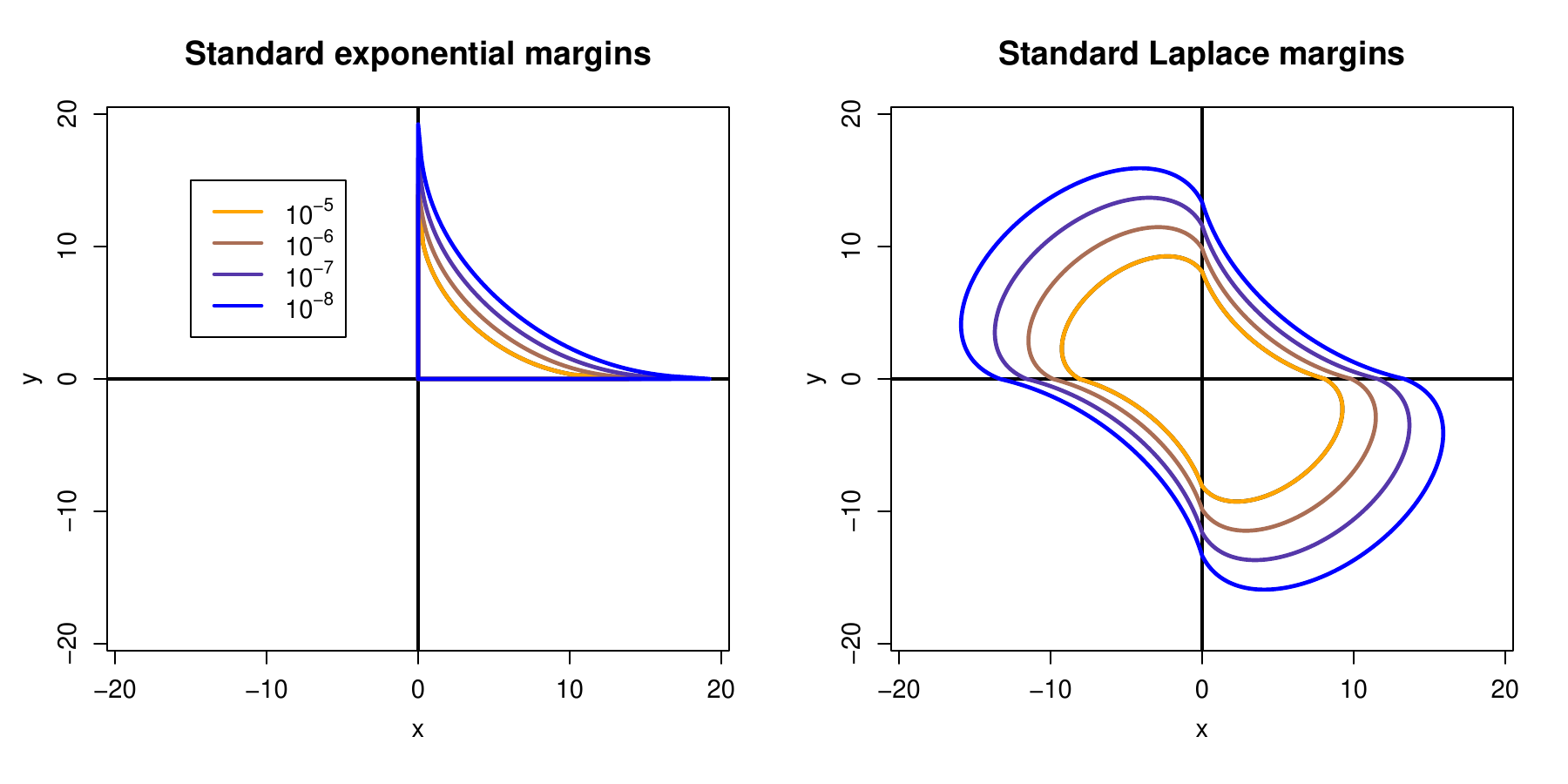}
    \caption{Isodensity contours for a bivariate random vector $\boldsymbol{X}$ with a Gaussian copula (with negative correlation) and with standard exponential (left) or standard Laplace (right) margins. The density levels for each colour are given in the legend in the left panel.}
    \label{fig:equidensity_contours}
\end{figure}

Recent theoretical developments for multivariate extremes have focused on \textit{geometric extremes}, whereby extremal dependence properties of $\boldsymbol{X}$ can be inferred directly from the deterministic limiting shapes of scaled sample clouds. Let $C_n := {\{ \boldsymbol{X}_i/r_n\}^n_{i=1}}$ denote $n$ independent copies of $\boldsymbol{X}$, scaled by a suitably chosen positive sequence $(r_n)_{n \in \NN}$ satisfying $r_n \to \infty$ as $n \to \infty$. Under mild conditions, $C_n$ converges in probability, with respect to the Hausdorff metric, onto the compact, star-shaped \textit{limit} set $\mathcal{G} := \{\boldsymbol{x}: g(\boldsymbol{x}) \leq 1 \} \subset [-1,1]^d$, where $g:\RR^d \mapsto \RR_+$ is the \textit{gauge function} of $\mathcal{G}$ \citep{Fisher1969,Davis1988,Kinoshita1991}. We formally define star-shaped sets in Section \ref{sec:the:overview} below. A sufficient condition for convergence onto $\mathcal{G}$ is that 
\begin{equation} \label{eqn:gauge_def}
    -\log f_{\boldsymbol{X}}(t\boldsymbol{x}) \sim tg(\boldsymbol{x}), \; \; \boldsymbol{x} \in \RR^d, \; \; t \to \infty,
\end{equation}
or, equivalently, $g(\boldsymbol{x}) = \lim_{t \to \infty}[-\log f_{\boldsymbol{X}}(t\boldsymbol{x})]/t$ \citep{Balkema2010}. Following \citet{Nolde2014}, we also define the \textit{unit-level}, or boundary, set $\partial \mathcal{G} := \{\boldsymbol{x}: g(\boldsymbol{x}) = 1 \} \subset \mathcal{G}$. For standard Laplace margins, it suffices to set $r_n = \log(n/2)$ to achieve the required convergence \citep{Papastathopoulos2024}. To demonstrate this convergence, Figure \ref{fig:true_gauges} illustrates the limit set $\mathcal{G}$ and unit-level set $\partial \mathcal{G}$ for three copulas, alongside $n = 10,000$ samples $\{\boldsymbol{x}_i/\log(n/2) \}_{i=1}^n$ from each copula; formal definitions of these copulas are given in Section \ref{sec:sim_study}. One can observe that the (finite) observed sample clouds lie approximately within the theoretical limit set. Hereafter, we implicitly assume that any $\boldsymbol{X}$ satisfies the conditions for convergence onto $\mathcal{G}$. 

\begin{figure}
    \centering
    \includegraphics[width=\textwidth]{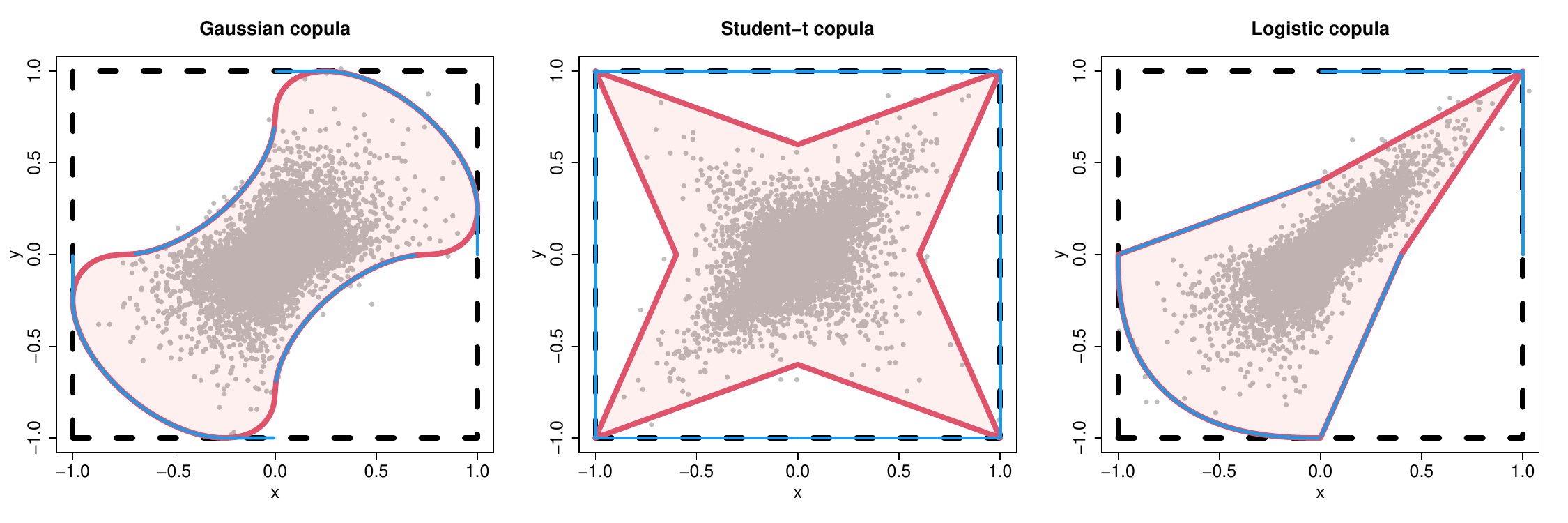}
    \caption{Scaled sample clouds of size $n=10,000$ from Gaussian (left), Student-t (middle), and logistic (right) copulas on standard Laplace margins. For each panel, the shaded regions and solid red lines give the limit set $\mathcal{G}$ and the unit-level set $\partial \mathcal{G}$, respectively, while the solid blue lines denote the set $\{\boldsymbol{w}/\Lambda(\boldsymbol{w}) : \boldsymbol{w} \in \mathcal{S}^{d-1}\setminus \mathcal{A} \}$; see Section \ref{sec:extended_ADF} for further details. The black dashed lines denote the unit square.}
    \label{fig:true_gauges}
\end{figure}


\citet{Nolde2014} and \citet{Nolde2022} show that the shape of $\partial \mathcal{G}$, or equivalently $\mathcal{G}$, is directly related to extremal dependence of $\boldsymbol{X}$. Specifically, $\partial \mathcal{G}$ links several representations for multivariate extremes: from $\partial \mathcal{G}$, we can immediately determine parameters of the models described in equations \eqref{eqn:led_tawn} and \eqref{eqn:wads_tawn}, as well as those proposed by \citet{Heffernan2004} and \citet{Simpson2020}. Taking, e.g., the ADF, we have that $\lambda(\boldsymbol{w}) = ||\boldsymbol{w}||_{\infty} \times \mathfrak{r}_{\boldsymbol{w}}^{-1}$ for all $\boldsymbol{w} \in \dsphere_+$, where $\mathfrak{r}_{\boldsymbol{w}} := \max \{ \mathfrak{r} \in [0,1] : \mathfrak{r}\mathcal{R}_{\boldsymbol{w}}\cap \partial \mathcal{G} \neq \emptyset \}$ is a $\boldsymbol{w}$-dependent coefficient used to scale the set $\mathcal{R}_{\boldsymbol{w}} := \bigotimes_{i = 1,\hdots,d} [w_i/||\boldsymbol{w}||_{\infty},\infty]$ to intersect with $\partial \mathcal{G}$. Such parameters provide an interpretable quantification of extremal dependence. Furthermore, from a practical perspective, estimates of $\partial \mathcal{G}$  can be used to estimate extreme statistics, for example, joint tail probabilities \citep{Wadsworth2024}, return curves \citep{Murphy-Barltrop2024c}, and return level sets \citep{Papastathopoulos2024}. The unit-level set $\partial \mathcal{G}$ therefore offers a high degree of practical utility for inference of multivariate extremes, and thus its accurate estimation is of particular importance. 

Owing to this perspective, recent works have introduced techniques for the estimation of $\partial \mathcal{G}$, on both standard exponential and standard Laplace margins. For the former margin, \citet{Simpson2022} and \citet{Majumder2024} proposed semi-parametric techniques using generalised additive models (GAMs) and B\'ezier splines, respectively, to approximate $\partial \mathcal{G}$, while \citet{Wadsworth2024} used a parametric copula-based model. For Laplace margins, \citet{Papastathopoulos2024} proposed a latent Gaussian model to approximate the shape of $\partial \mathcal{G}$, while \citet{Murphy-Barltrop2024} used GAMs to estimate $\partial \mathcal{G}$ using the model introduced by \citet{Mackay2023}. Within these approaches, a wide range of statistical techniques have been considered in both Bayesian and frequentist settings, and it is possible to obtain accurate estimates of $\partial \mathcal{G}$ for a large variety of dependence structures. 
However, current techniques for modelling and estimating $\partial \mathcal{G}$ have several shortcomings: i) all existing approaches are limited to a low-dimensional setting ($d\leq 4$), ii) the restriction of several approaches to standard exponential margins offers a limited perspective for evaluating joint tail properties (see Figure \ref{fig:equidensity_contours}), iii) one must always specify parametric or semi-parametric forms for quantities related to $\partial \mathcal{G}$, or select a large number of tuning parameters, and iv) the existing techniques for Laplace margins do not guarantee the theoretical properties of limit sets. Therefore, the existing estimation techniques for $\partial \mathcal{G}$ offer limited practical utility, motivating novel developments. 

Recent literature combining extreme value theory and deep learning has seen the construction of flexible, computationally-scalable models for univariate extremes \citep[see, e.g.,][]{pasche2022neural,richards2022regression, cisneros2023deep}, generative models for multivariate and spatial extremes \citep{boulaguiem2022modeling, lafon2023vae, zhang2023flexible, majumder_aoas_24}, and classifiers for extremal dependence \citep{Ahmed2022,Wixson2024}. While \cite{hasan2022modeling} used neural networks to build flexible models for asymptotically-dependent multivariate extremes, deep learning is yet to be exploited in the construction of models that can capture non-asymptotically dependent data structures. Here, we propose the first deep learning-based approach for modelling $\partial \mathcal{G}$, referred to hereafter as the DeepGauge framework. The use of deep learning methods gives the DeepGauge framework a high degree of flexibility and, as we demonstrate in Sections \ref{sec:sim_study} and \ref{sec:app}, allows us to capture a wide variety of extremal dependence structures. Furthermore, it can be applied in higher dimensional settings ($d > 4$) than existing techniques for estimating $\partial \mathcal{G}$, and also requires fewer modelling assumptions.

The paper is organised as follows. Section \ref{sec:theory} outlines the theory underpinning the DeepGauge framework, alongside novel theoretical results pertaining to the unit-level set $\partial \mathcal{G}$ and its estimation for standard Laplace margins. Section~\ref{sec:methods} outlines our methodology for estimating $\partial \mathcal{G}$ and related quantities, with Section~\ref{sec:method:neuralnets} detailing our neural network-based representation for geometric extremes and Section~\ref{sec:inference:gof} introducing diagnostics for validating model fits in high dimensional settings. Section~\ref{sec:sim_study} provides a simulation study showcasing the efficacy of our framework for inferring the extremal dependence of random vectors. Section~\ref{sec:app} provides an application to the NORA10 hindcast dataset of meteorological and oceanographic (metocean) variables in the North Sea that exhibit complex dependence structures. We conclude in Section~\ref{sec:discussion} with a discussion and outlook on future work.

\section{Theoretical developments in Geometric Extremes} \label{sec:theory}


\subsection{Overview of the angular-radial decomposition} 
\label{sec:the:overview}

To estimate the unit-level set $\partial \mathcal{G}$ and its corresponding gauge function $g(\cdot)$, we first decompose $\boldsymbol{X}$ into angular and radial components, and then model the radii  conditional on a fixed angle. While one could select one of many radial-angular systems for this decomposition, we follow \cite{Murphy-Barltrop2024} and define angular and radial components via the Euclidean norm. For any $\boldsymbol{X} \in \RR^d\setminus  \boldsymbol{0}_d$, with $\boldsymbol{0}_d := (0,\hdots,0)^T$, define $(R,\boldsymbol{W})$ by $\boldsymbol{X}\mapsto(R,\boldsymbol{W}):=(\|\boldsymbol{X}\|,\boldsymbol{X}/\|\boldsymbol{X}\|)$ for $R >0$ and $\boldsymbol{W}\in\mathcal{S}^{d-1}$, where $\mathcal{S}^{d-1} := \{ \boldsymbol{x} \in \RR^{d}: ||\boldsymbol{x}|| = 1\}$ denotes the unit $(d-1)$-sphere. It follows that $\boldsymbol{X} = R\boldsymbol{W}$, implying that $\boldsymbol{X}$ is completely determined by the behaviour of $(R,\boldsymbol{W})$. It is trivial to show that the mapping $t:\RR^d\setminus \boldsymbol{0}_d \mapsto \RR_+ \times \mathcal{S}^{d-1}$, where $t(\boldsymbol{x}) := (|| \boldsymbol{x} ||, \boldsymbol{x}/|| \boldsymbol{x} ||)$, is bijective; thus, no information is lost through considering $(R,\boldsymbol{W})$. Loosely speaking, $R$ is the magnitude of an event, while $\boldsymbol{W}$ defines its direction, e.g., in which orthant of $\RR^d$ the event occurs. 

We now recall some theoretical properties of gauge functions and limit sets. The star-shapedness of the limit set $\mathcal{G}$ implies that, for any $\boldsymbol{x} \in \mathcal{G}$ and $t > 0$, we have $t\boldsymbol{x} \in \mathcal{G}$. Moreover, if $\boldsymbol{0}_d \in \mathcal{G}$ and $g(\boldsymbol{0}_d) <1$, we have that the line segment $\{\boldsymbol{0}_d + t\boldsymbol{x}: t \in [0,1] \} \subset \mathcal{G}$ for any $\boldsymbol{x} \in \mathcal{G}$. Furthermore, one can show that the componentwise maxima and minima of $\mathcal{G}$ equal $\boldsymbol{1}_d$ and $-\boldsymbol{1}_d$ respectively, implying that $\mathcal{G}$ (and, thus, the unit-level set $\partial \mathcal{G}$) must touch all boundaries of the unit hypercube $[-1,1]^d$ at least once. Finally, we note that the gauge function $g(\cdot)$ of $\mathcal{G}$ is 1-homogeneous, i.e., $g(t\boldsymbol{x}) = tg(\boldsymbol{x})$ for any $\boldsymbol{x} \in \RR^d, t \in \RR_+$.

Using the radial-angular decomposition of $\boldsymbol{X}$ and the star-shaped property of $\mathcal{G}$, we can reformulate the unit-level set as $\partial \mathcal{G} = \left\{ r\boldsymbol{w}: r>0, \boldsymbol{w} \in \mathcal{S}^{d-1}, g(r\boldsymbol{w}) = 1 \right\}.$ Homogeneity of $g(\cdot)$ implies that, for any $\boldsymbol{w} \in \mathcal{S}^{d-1}$, the radial value of the corresponding point on the unit-level set must be $1/g(\boldsymbol{w})$; hence, $\partial \mathcal{G} = \left\{  \boldsymbol{w}/g(\boldsymbol{w}) : \boldsymbol{w} \in \mathcal{S}^{d-1} \right\}.$ This reformulation has the powerful implication that, to determine $\partial \mathcal{G}$, we need only to evaluate $g(\cdot)$ on $\dsphere$. Illustrations of the radial-angular representations for $\mathcal{G}$ and $\partial \mathcal{G}$ are given in Figure~\ref{fig:limit_set_breakdown}. 

\begin{figure}[t!]
    \centering
    \includegraphics[width=.7\textwidth]{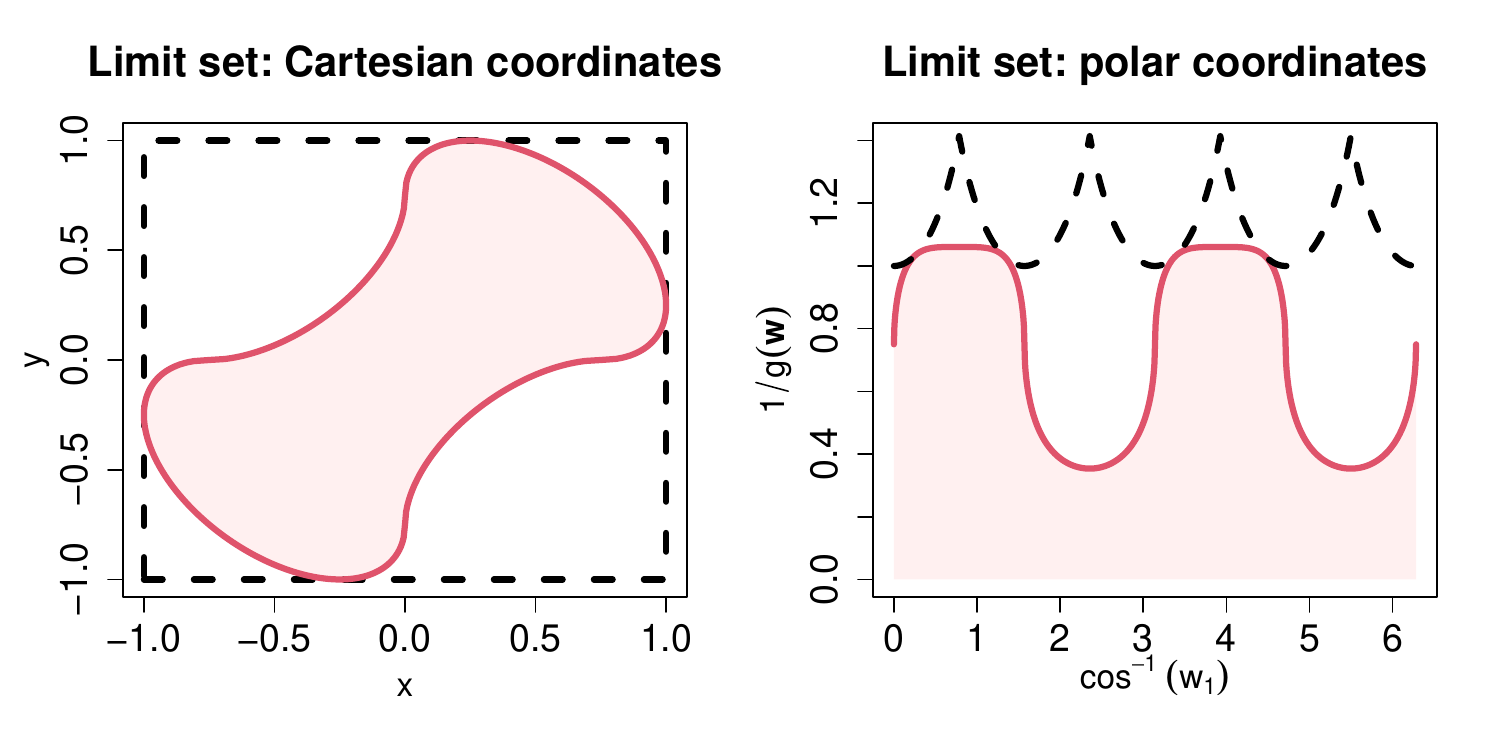}
    \caption{Gaussian limit set $\mathcal{G}$ (shaded regions) and unit-level set $\partial \mathcal{G}$ (solid red lines) in Cartesian (left) and polar (right) coordinates. The black dashed lines in both plots denote the unit square. For the right panel, we plot the radii of the unit-level set (i.e., $1/g(\boldsymbol{w})$) against the standard polar angle, $\cos^{-1}(w_1) = \sin^{-1}(w_2) \in [0,2\pi)$ for each $(w_1,w_2) \in \mathcal{S}^{1}$.}
    \label{fig:limit_set_breakdown}
\end{figure}

\subsection{Geometric extremes on Laplace margins}
As noted in Section \ref{sec:intro}, we consider data on standard Laplace margins as this permits a more detailed description of joint tail behaviour. However, the theoretical results provided by \citet{Nolde2022} linking the unit-level set $\partial \mathcal{G}$ to a variety of modelling frameworks are given only for random vectors on standard exponential margins. The next proposition illustrates that some of the same results also hold for data on Laplace margins.   

\begin{proposition} \label{prop:gauge_exp_laplace}
    Consider a random vector $\boldsymbol{X}\in\mathbb{R}^d \setminus\boldsymbol{0}_d$ with standard Laplace margins and gauge function $g(\cdot)$. Let $\boldsymbol{X}_E$ denote the same vector with unit exponential margins, with gauge function $g_E(\cdot)$. We have equality of the gauge functions for positive angles, that is, $g(\boldsymbol{w}) = g_E(\boldsymbol{w})$ for all $\boldsymbol{w} \in \mathcal{S}^{d-1}_+$.
\end{proposition}
Proof of Proposition~\ref{prop:gauge_exp_laplace} is provided in Appendix~\ref{proof:prop:gauge_exp_laplace}. 
\begin{remark}
    From Proposition \ref{prop:gauge_exp_laplace}, we immediately see that the sets $\left\{ \boldsymbol{w}/g_E(\boldsymbol{w}) : \boldsymbol{w} \in \mathcal{S}_+^{d-1}  \right\}$ and $\left\{ \boldsymbol{w}/g(\boldsymbol{w}) : \boldsymbol{w} \in \mathcal{S}_+^{d-1} \right\}$ are identical, implying equality of the unit-level sets on the positive orthant. This has the strong implication that the results proposed by \citet{Nolde2022}, which link the frameworks described in equations \eqref{eqn:led_tawn} and \eqref{eqn:wads_tawn}, as well as the model proposed by \citet{Simpson2020}, to the unit-level set $\partial \mathcal{G}$, are also valid for random vectors with Laplace margins. 
\end{remark}

\subsection{Constructing valid unit-level sets}
\label{sec:valid}

We now consider estimation of the unit-level set $\partial \mathcal{G}$ or, equivalently, the gauge function $g(\cdot)$ on $\dsphere$. This requires estimates of the corresponding limit set, $\mathcal{G}$, to have certain properties; see Section~\ref{sec:the:overview}. In what follows, we show that such properties are easily satisfied via an appropriate estimator for $g(\cdot)$. 

We begin by noting that $g(\cdot)$ must satisfy the constraint described in the following Proposition \ref{prop:g_lowerbound}. 
\begin{proposition}\label{prop:g_lowerbound}
    For all $\boldsymbol{w} \in \mathcal{S}^{d-1}$, the gauge function $g(\cdot)$ satisfies the constraint that 
    \begin{equation*}
        g(\boldsymbol{w}) \geq ||\boldsymbol{w}||_{\infty},
    \end{equation*} 
    where $||\boldsymbol{x}||_{\infty} := \max\{|x_1|,\hdots,|x_d|\}$ denotes the infinity norm.
\end{proposition}
\begin{proof}
    Given $\boldsymbol{w} \in \mathcal{S}^{d-1}$, let $\boldsymbol{w}/g(\boldsymbol{w}) \in \partial \mathcal{G}$ denote the corresponding point on the unit-level set. Since $\partial \mathcal{G} \subset [-1,1]^d$, we must have $\max_{i=1,\dots,d} \{{w}_i\}/g(\boldsymbol{w}) = \max_{i=1,\dots,d} \{{w}_i/g(\boldsymbol{w})\} \leq 1,$ which implies that $\max_{i=1,\dots,d} \{{w}_i\} \leq g(\boldsymbol{w})$. Similarly for the vector minima, we obtain $\max_{i=1,\dots,d} \{-{w}_i\} = -\min_{i=1,\dots,d} \{{w}_i\} \leq g(\boldsymbol{w})$. Considering each component of $\boldsymbol{w}$ gives $g(\boldsymbol{w}) \geq \max(w_i,-w_i)$ for all $i = 1,\hdots, d$; hence,  $g(\boldsymbol{w}) \geq ||\boldsymbol{w}||_{\infty}$.
\end{proof}
Ignoring this constraint may lead to estimates of $\mathcal{G}$ not contained within the hypercube $[-1,1]^d$. Such estimates lack any theoretical interpretation and must undergo a user-specified rescaling to be of use; see, e.g., \cite{Papastathopoulos2024} or \citet{Murphy-Barltrop2024}. We exploit Proposition~\ref{prop:g_lowerbound} to ensure that our model does not suffer from this problem and can be designed to always provide valid estimates of limit sets. Now consider any continuous \textit{radial} function $h:\dsphere \mapsto \RR_+$. The following proposition holds.
\begin{proposition} \label{prop:limit_set_properties}
    Suppose $h(\cdot):\dsphere \mapsto \RR_+$ satisfies $1/h(\boldsymbol{w}) \geq ||\boldsymbol{w}||_{\infty}$ for all $\boldsymbol{w} \in \dsphere$, and define the set 
    \begin{equation*}
        \mathcal{H} := \left\{ \boldsymbol{x} \in \RR^d\setminus\{ \boldsymbol{0}_d\} \; \bigg\vert \; ||\boldsymbol{x}|| \leq h(\boldsymbol{x}/||\boldsymbol{x}||) \right\} \bigcup \bigg\{ \boldsymbol{0}_d\bigg\}.
    \end{equation*}
    Then $\mathcal{H}$ is star-shaped and satisfies $\mathcal{H} \subset [-1,1]^d$. Moreover, $\mathcal{H}$ is compact. 
\end{proposition} 
Proof of Proposition~\ref{prop:limit_set_properties} is provided in Appendix~\ref{proof:prop:limit_set_properties}. Proposition \ref{prop:limit_set_properties} implies that we can design models for $\mathcal{G}$ that satisfy (some of) the validity properties of limit sets by starting with any (potentially arbitrary) continuous radial function $h(\cdot)$ that satisfies $1/h(\boldsymbol{w}) \geq ||\boldsymbol{w}||_{\infty}$; the corresponding set $\mathcal{H}$ will be a subset of $[-1,1]^d$, star-shaped, and compact. Whilst $\mathcal{H}$ must be a subset of the unit hyper-cube, it is not guaranteed to intersect with its boundary in each component, or even at all; the same is true of the boundary set $\partial \mathcal{H}$. Since this is a requirement of valid limit and unit-level sets, we propose a rescaling procedure to ensure it is always satisfied for our model. 

 Observe that the boundary of $\mathcal{H}$ is given by $\partial \mathcal{H} = \left\{ \boldsymbol{w} h(\boldsymbol{w}) : \boldsymbol{w} \in \dsphere \right\}$. For each $i = 1,\hdots, d$, we define $b_i(w_i):= \mathbbm{1}(w_i \geq 0)b_i^U - \mathbbm{1}(w_i < 0)b_i^L > 0$, where
\begin{align}
\label{Eq:scaling_factors}
       b_i^U:= \max \left\{ w_ih(\boldsymbol{w})  \mid \boldsymbol{w} \in \dsphere  \right\} > 0, \quad \text{and}\quad 
        b_i^L := \min \left\{ w_ih(\boldsymbol{w}) \mid \boldsymbol{w} \in \dsphere \right\} < 0,
\end{align}
for all $i = 1,\hdots,d$. Using these scaling functions, we define the rescaled set 
\begin{equation*}
    \widetilde{\partial \mathcal{H}} := \left\{  h(\boldsymbol{w}) \left(\frac{w_1}{b_1(w_1)}, \hdots, \frac{w_d}{b_d(w_d)} \right) \bigg\vert \boldsymbol{w} \in \dsphere \right\},
\end{equation*}
which satisfies the following Proposition \ref{prop:rescaling}. 
\begin{proposition} \label{prop:rescaling}
    The rescaled set $\widetilde{\partial \mathcal{H}}$ is in one-to-one correspondence with $\partial \mathcal{H}$, satisfies $\widetilde{\partial \mathcal{H}} \subset [-1,1]^d$, and has componentwise maxima and minima of $ \boldsymbol{1}_d$ and $-\boldsymbol{1}_d$, respectively. 
\end{proposition}
To map $\partial \mathcal{H}$ to $\widetilde{\partial \mathcal{H}}$, we require the transformation described in the following lemma.

\begin{lemma} \label{lem:kappa_bijective}
    Let $\kappa:\dsphere \mapsto \dsphere$ denote the following mapping: 
        \begin{equation*}
            \kappa(\boldsymbol{w}) =   \left(\frac{w_1}{b_1(w_1)}, \hdots, \frac{w_d}{b_d(w_d)} \right) \bigg/  \bigg\lVert  \left(\frac{w_1}{b_1(w_1)}, \hdots, \frac{w_d}{b_d(w_d)} \right) \bigg \rVert .
        \end{equation*}
    Then $\kappa$ is a bijective mapping.  
\end{lemma}
Lemma~\ref{lem:kappa_bijective} is used to prove Proposition~\ref{prop:rescaling} in Appendix~\ref{proof:prop:rescaling}. Proposition \ref{prop:rescaling} permits a new construction of a valid gauge function, denoted by $\Tilde{g}:\dsphere \mapsto \RR_+$, with 
\begin{equation} \label{eqn:tilde_g}
    \Tilde{g}(\boldsymbol{w}) := 1 \bigg/ \left\Vert h(\kappa^{-1}(\boldsymbol{w})) \left(\frac{\kappa^{-1}(\boldsymbol{w})_1}{b_1(\kappa^{-1}(\boldsymbol{w})_1)}, \hdots, \frac{\kappa^{-1}(\boldsymbol{w})_d}{b_d(\kappa^{-1}(\boldsymbol{w})_d)} \right) \right\Vert .
\end{equation}
Note that we term $\Tilde{g}(\cdot)$ the \textit{rescaled} gauge function. Since $h(\cdot)$ is continuous, we have that $\Tilde{g}(\cdot)$ is also continuous. 
Moreover, as we show in the following corollaries, $\Tilde{g}(\cdot)$ satisfies the theoretical properties required to produce valid limit sets. 
\begin{corollary}\label{corol:rescaled_gauge_estimate}
    The rescaled gauge function $\Tilde{g}(\cdot)$ satisfies $\Tilde{g}(\boldsymbol{w})\geq || \boldsymbol{w} ||_{\infty}$ for all $\boldsymbol{w} \in \dsphere$. Furthermore, letting $\boldsymbol{w}^{u,i} = \argmax_{\boldsymbol{w} \in \dsphere}\{w_ih(\boldsymbol{w}) \}$ and $\boldsymbol{w}^{l,i} = \argmin_{\boldsymbol{w} \in \dsphere}\{w_ih(\boldsymbol{w}) \}$ for $i=1,\dots,d$, we have $\Tilde{g}(\kappa(\boldsymbol{w}^{u,i})) = || \kappa(\boldsymbol{w}^{u,i}) ||_{\infty}$ and $\Tilde{g}(\kappa(\boldsymbol{w}^{l,i})) = || \kappa(\boldsymbol{w}^{l,i}) ||_{\infty}$ for each $i = 1,\hdots, d$.  
\end{corollary}

\begin{corollary} \label{corol:all_properties}
    Let $h(\cdot)$ be any continuous radial function and define the corresponding rescaled gauge function by $\tilde{g}(\cdot)$ as in \eqref{eqn:tilde_g}. The set 
     \begin{equation*}
            \widetilde{\mathcal{H}} := \left\{ \boldsymbol{x} \in \RR^d\setminus\{ \boldsymbol{0}_d\} \bigg\vert ||\boldsymbol{x}|| \leq \frac{1}{ \tilde{g}(\boldsymbol{x}/||\boldsymbol{x}||)} \right\} \bigcup \bigg\{ \boldsymbol{0}_d\bigg\},
        \end{equation*}
    is star-shaped, compact, and satisfies $\widetilde{\mathcal{H}} \subset [-1,1]^d$. Furthermore, $\widetilde{\mathcal{H}}$ has componentwise maxima and minima $ \boldsymbol{1}_d$ and $-\boldsymbol{1}_d$, respectively.
\end{corollary}

   Proof of Corollary~\ref{corol:rescaled_gauge_estimate} is provided in Appendix~\ref{proof:corol:rescaled_gauge_estimate}. Proof of Corollary~\ref{corol:all_properties} follows directly from Propositions \ref{prop:limit_set_properties} and \ref{prop:rescaling}, and Corollary \ref{corol:rescaled_gauge_estimate}.

\begin{remark}
   Corollary \ref{corol:all_properties} also implies that the boundary $\widetilde{\partial \mathcal{H}}$ associated with $\widetilde{\mathcal{H}}$ satisfies all of the required theoretical properties of valid unit-level sets. 
\end{remark} 

For inference of the limit set, we exploit Corollary~\ref{corol:all_properties} to ensure our estimates of the limit (or unit-level) set have the required properties for validity; see Section~\ref{sec:estimation_gauge_adf}. We note that an alternative rescaling was proposed by \citet{Papastathopoulos2024}. However, this was applied post-hoc on an initial estimate of the gauge function via a two-step procedure, whereas our rescaling is performed during inference without the need for an additional step. Furthermore, incorporating the rescaling described in equation~\eqref{eqn:tilde_g} within many standard optimisation algorithms would likely result in slow model fitting procedures due to additional labor-heavy computations. However, the use of stochastic gradient descent (see Section~\ref{sec:methods:estimation:training}) within deep learning procedures avoids the bulk of this computational burden, allowing rescaling to be incorporated within the proposed DeepGauge framework.    

\subsection{Extended Angular Dependence Function} \label{sec:extended_ADF}

Consider now any gauge function $g(\cdot)$ with limit and unit-level sets $\mathcal{G}$ and $\partial \mathcal{G}$ respectively. Proposition \ref{prop:gauge_exp_laplace} implies that valid estimates of $\partial \mathcal{G}$ can be used to immediately obtain parameter estimates for several existing modelling frameworks, providing information about the extremal dependence structure in the positive orthant of $\RR^d$. However, as noted in Section~\ref{sec:intro}, our interest lies more generally in understanding the extremal dependence in all orthants of $\RR^d$. The next proposition shows that, with suitable scaling coefficients, the limit set in any orthant can be obtained by considering a rescaled random vector on $\dsphere_+$.  

\begin{proposition} \label{prop:trans_gauge}
    Given $\boldsymbol{w}  \in \mathcal{S}^{d-1}$, let $\boldsymbol{c}:=(\varepsilon(w_1),\hdots,\varepsilon(w_d))^T$ where $\varepsilon(x) = 1$ for $x\geq 0$ and $\varepsilon(x)=-1$, otherwise. For $\boldsymbol{c}\boldsymbol{X}:=(c_1X_1,\hdots,c_dX_d)^T$ and $\boldsymbol{c}\boldsymbol{w}:=(c_1w_1,\hdots,c_dw_d)^T \in\dsphere_+$, we have $g(\boldsymbol{w}) = g_{\boldsymbol{c}\boldsymbol{X}}(\boldsymbol{c}\boldsymbol{w})$, where $g_{\boldsymbol{c}\boldsymbol{X}}(\cdot)$ denotes the gauge function of $\boldsymbol{c}\boldsymbol{X}$. 
\end{proposition}
\begin{proof}
    Letting $f_{\boldsymbol{c}\boldsymbol{X}}(\cdot)$ denote the joint density function of $\boldsymbol{c}\boldsymbol{X}$, we have 
    \begin{equation*}
        g(\boldsymbol{w}) = \lim_{t \to \infty}[-\log f_{\boldsymbol{X}}(t\boldsymbol{w})]/t = \lim_{t \to \infty}[-\log f_{\boldsymbol{c}\boldsymbol{X}}(t \boldsymbol{c}\boldsymbol{w})]/t = g_{\boldsymbol{c}\boldsymbol{X}}(\boldsymbol{c}\boldsymbol{w}),
    \end{equation*}
    as the Jacobian of the transformation $\boldsymbol{c}\boldsymbol{X} \mapsto \boldsymbol{X}$ is equal to $1$. 
\end{proof}

\begin{remark}
    Proposition \ref{prop:trans_gauge} implies that, for any $\boldsymbol{c} \in \{-1,1 \}^d$ that leads to the partitioning $\dsphere_{\boldsymbol{c}} := \{ \boldsymbol{w} \in \mathcal{S}^{d-1} : \boldsymbol{c}:=(\varepsilon(w_1),\hdots,\varepsilon(w_d))^T \}$ of the unit $(d-1)$-sphere, the sets ${\{ \boldsymbol{w}/g(\boldsymbol{w}) : \boldsymbol{w} \in \dsphere_{\boldsymbol{c}} \}}$ and $\{ \boldsymbol{w}/g_{\boldsymbol{c}\boldsymbol{X}}(\boldsymbol{c}\boldsymbol{w}) : \boldsymbol{w} \in \dsphere_{\boldsymbol{c}} \}$ are equal. Therefore,  the unit-level set $\partial \mathcal{G}$ can be obtained by evaluating gauge functions for rescaled vectors on the set $\{ \boldsymbol{w} \in \mathcal{S}^{d-1} : \min_{i=1,\dots,d}\{w_{i}\} \geq 0\}$.  
\end{remark}
Propositions \ref{prop:gauge_exp_laplace} and \ref{prop:trans_gauge} have implications when using the geometric representation to infer the ADF from \eqref{eqn:wads_tawn}. As this framework is given for vectors on standard exponential margins, careful treatment is required to define an analogous model for Laplace margins.
\begin{proposition} \label{prop:wads_tawn_exp_laplace}
    Given any angle $\boldsymbol{w} \in \dsphere_+$, assume that equation \eqref{eqn:wads_tawn} holds for the random vector $\boldsymbol{X}_E$. Then 
    \begin{equation*} 
        \Pr\left(\min_{i \in \mathcal{D}}\{X_{i}/w_{i}\}>u\right) = L(e^u;\boldsymbol{w})e^{-\lambda(\boldsymbol{w})u}, \; \; u \to \infty. 
    \end{equation*}
\end{proposition}
Proof of Proposition~\ref{prop:wads_tawn_exp_laplace} is provided in Appendix~\ref{proof:prop:wads_tawn_exp_laplace}. Combining Proposition~\ref{prop:wads_tawn_exp_laplace} and model \eqref{eqn:wads_tawn} allows us to assess joint tail behaviour in the positive orthant $\RR^d_+$ through the ADF, $\lambda(\cdot)$. We further extend this model to $\mathbb{R}^d$ by considering an \textit{extended} ADF, denoted by $\Lambda(\cdot)$. Specifically, given any $\boldsymbol{w}  \in  \mathcal{S}^{d-1}\setminus \mathcal{A}$, where $ \mathcal{A} := \bigcup_{i=1}^d \{ \boldsymbol{w} \in \dsphere : w_i = 0 \}$ is the intersection of $\dsphere$ with each axis, we assume that 
\begin{equation} \label{eqn:wads_tawn_laplace} 
    \Pr\left(\min_{i \in \mathcal{D}}\{X_{i}/w_{i}\}>u\right) = \mathcal{L}(e^u;\boldsymbol{w})e^{-\Lambda(\boldsymbol{w})u}, \; \; u \to \infty,
\end{equation}
where $\mathcal{L}(\cdot;\boldsymbol{w})$ is a slowly varying function. A formal definition of $\Lambda(\cdot)$ is given in Appendix~\ref{proof:prop:extended_ADF}. Note that the extended ADF corresponds with the copula exponent function introduced by \citet{Mackay2023} for data on uniform margins. The next proposition illustrates that, under mild conditions, the convergence in \eqref{eqn:wads_tawn_laplace} is always achieved. 

\begin{proposition} \label{prop:extended_ADF}
    Assume that the conditions of Proposition \ref{prop:wads_tawn_exp_laplace} are satisfied for any random vector $\boldsymbol{c}\boldsymbol{X}$, where $\boldsymbol{c} \in \{-1,1\}^d$. Then equation \eqref{eqn:wads_tawn_laplace} holds for any $\boldsymbol{w}  \in \mathcal{S}^{d-1} \setminus \mathcal{A} $.  
\end{proposition}
Proof of Proposition~\ref{prop:extended_ADF} is provided in Appendix~\ref{proof:prop:extended_ADF}. Note that while the assumptions of Proposition \ref{prop:extended_ADF} may seem restrictive, \citet{Wadsworth2013} demonstrate through rigorous theoretical treatment that model \eqref{eqn:wads_tawn} captures the joint tail structure of $\boldsymbol{X}_E$ for a wide variety of theoretical examples. Therefore, it is reasonable to assume this framework can capture the tail structure in every orthant. 

\begin{remark} \label{remark:discont}
   In the definition of our model \eqref{eqn:wads_tawn_laplace}, we purposely exclude any angles in $\dsphere$ that intersect the axes, as the model is not well defined there. Take, e.g., $\boldsymbol{w}:=(1,0,\hdots,0)^T$; it is not clear whether one should consider the probability $\Pr(X_1 > u, X_i > 0, i = 2,\hdots,d)$ or $\Pr(X_1 > u, X_i < 0, i = 2,\hdots,d)$ as $u \to \infty$ under the modelling framework. This results in discontinuities for the extended ADF at the axes; see Figure \ref{fig:true_gauges}. 
\end{remark}

On the original exponential margins, \citet{Nolde2022} show that the unit-level set $\partial \mathcal{G}$ is linked to the ADF; see Section \ref{sec:intro}. As demonstrated by Proposition~\ref{prop:extended_ADF_boundary_link}, an analogous relationship holds for the extended ADF.
\begin{proposition} \label{prop:extended_ADF_boundary_link}
    Suppose equation \eqref{eqn:wads_tawn_laplace} holds for any angle $\boldsymbol{w}  \in \mathcal{S}^{d-1}  \setminus \mathcal{A}$. Then $\Lambda(\boldsymbol{w}) = ||\boldsymbol{w}||_{\infty} \times \Tilde{\mathfrak{r}}^{-1}_{\boldsymbol{w}}$, where $\Tilde{\mathfrak{r}}_{\boldsymbol{w}} = \max \{ \mathfrak{r} \in [0,1] : \mathfrak{r}\Tilde{\mathcal{R}}_{\boldsymbol{w}}\cap \partial \mathcal{G} \neq \emptyset \}$
    and $\Tilde{\mathcal{R}}_{\boldsymbol{w}} := \bigotimes_{i = 1,\hdots,d} \mathcal{U}_{w_i}$, with $\mathcal{U}_{w_i} := [w_i/||\boldsymbol{w}||_{\infty},\infty]$ for $w_i > 0$ and $[-\infty,w_i/||\boldsymbol{w}||_{\infty}]$ for $w_i < 0$. 
\end{proposition}
Proof of Proposition~\ref{prop:extended_ADF_boundary_link} is provided in Appendix~\ref{proof:prop:extended_ADF_boundary_link}, alongside the illustrative Figure~\ref{fig:adf_proof}. Proposition~\ref{prop:extended_ADF_boundary_link} illustrates that the extended ADF can be obtained directly from the unit-level set $\partial \mathcal{G}$. Consequently, the extended ADF is linked to the gauge function, as demonstrated by the following corollary. 

\begin{corollary} \label{corol:gauge_adf_bound}
    Suppose model \eqref{eqn:wads_tawn_laplace} holds for all $\boldsymbol{w} \in \mathcal{S}^{d-1}\setminus \mathcal{A}$. Then, we have that 
    \begin{equation*}
        g(\boldsymbol{w}) \geq \Lambda(\boldsymbol{w}) \geq ||\boldsymbol{w}||_{\infty}. 
    \end{equation*}
\end{corollary}

Proof of Corollary~\ref{corol:gauge_adf_bound} is provided in Appendix~\ref{proof:corol:gauge_adf_bound}, and visualised in Figure~\ref{fig:true_gauges}. The blue lines in each panel of Figure \ref{fig:true_gauges} denote the sets $\{\boldsymbol{w}/\Lambda(\boldsymbol{w}) : \boldsymbol{w} \in \mathcal{S}^{d-1}\setminus \mathcal{A} \}$, corresponding to the intersection points $\{\Tilde{\mathfrak{r}}_{\boldsymbol{w}}\Tilde{\mathcal{R}}_{\boldsymbol{w}} \cap \partial \mathcal{G}, \boldsymbol{w} \in \mathcal{S}^{d-1}\setminus \mathcal{A}$\}. Furthermore, from the first and third panels of Figure \ref{fig:true_gauges}, one can clearly observe the discontinuities of the extended ADF at the axes, as discussed in Remark \ref{remark:discont}.

\section{Inference}\label{sec:methods}
\subsection{Overview}
Here we describe our DeepGauge framework for modelling and estimating limit sets. Section~\ref{sec:estimation_gauge_adf} describes model assumptions for the conditional radii $R \mid (\boldsymbol{W}=\boldsymbol{w})$, $\boldsymbol{w}\in\dsphere$, through which we obtain estimates of the unit-level set $\partial G$. Section~\ref{sec:method:neuralnets} describes our DeepGauge representation of gauge functions using neural networks. Section \ref{sec:method:est_ADF} covers estimation of the extended ADF and its use in probability estimation. Section \ref{sec:inference:gof} concludes with diagnostic tools for assessing goodness-of-fit. 

\subsection{Modelling assumptions for the conditional radii} \label{sec:estimation_gauge_adf}

\citet{Wadsworth2024} demonstrate that, for large radial values $r$, we have that $f_{R\mid \boldsymbol{W}}(r\mid\boldsymbol{w}) \propto r^{d-1}\exp\{-rg(\boldsymbol{w})\}$, where $f_{R\mid \boldsymbol{W}}(\cdot\mid\cdot)$ denotes the density function of $R\mid( \boldsymbol{W} = \boldsymbol{w})$. This implies that the upper tail of $R\mid( \boldsymbol{W} = \boldsymbol{w})$ follows a gamma kernel form.
Note that this holds for $\boldsymbol{X}$ on both exponential and Laplace margins. To accommodate this form, \citet{Wadsworth2024} propose the modelling assumption: 
\begin{equation} \label{eqn:trunc_gamma_assum}
    R\mid (\boldsymbol{W} =\boldsymbol{w}, R>r_\tau(\boldsymbol{w}) ) \sim \text{truncGamma}(\alpha,g(\boldsymbol{w})),
\end{equation}
where `truncGamma' denotes a \textit{truncated Gamma} distribution with shape and rate parameters $\alpha>0$ and $g(\boldsymbol{w})>0$, respectively, and $r_\tau(\boldsymbol{w})$ denotes the $\tau$-quantile of $R\mid( \boldsymbol{W} = \boldsymbol{w})$ for some $\tau \in (0,1)$ close to 1, that is, $\Pr \{R \leq r_\tau(\boldsymbol{w}) \mid \boldsymbol{W} = \boldsymbol{w}\} = \tau$ for all $\boldsymbol{w} \in \mathcal{S}^{d-1}$. Under the distributional assumption of \eqref{eqn:trunc_gamma_assum}, we can use maximum likelihood techniques to estimate the gauge function $g(\cdot)$, which can be interpreted as a non-stationary rate parameter of the truncated gamma distribution. 

Through rigorous theoretical treatment, \citet{Wadsworth2024} show that equation \eqref{eqn:trunc_gamma_assum} is a valid modelling assumption for a wide range of parametric copulas, including the three illustrated in Figure~\ref{fig:true_gauges}. For most examples, the true shape parameter is $\alpha = d$. However, to increase the flexibility of their model, \citet{Wadsworth2024} permit estimation of this parameter; the DeepGauge framework follows suit.

\subsection{Estimation of the gauge function using neural networks}\label{sec:method:neuralnets}
To overcome the limited flexibility of existing geometric modelling approaches, we model the rescaled gauge function ($\tilde{g}$; see Equation~\eqref{eqn:tilde_g}) using neural networks. Full inference requires the construction of two models: one for the radial threshold, $r_\tau(\boldsymbol{w})$ in \eqref{eqn:trunc_gamma_assum}, and one for an (unscaled) gauge function ${g}(\boldsymbol{w})$ satisfying ${g}(\boldsymbol{w}) \geq ||\boldsymbol{w} ||_{\infty}$ for all $\boldsymbol{w} \in \dsphere$ (see Proposition~\ref{prop:g_lowerbound}). With the latter, we can employ the transformation described in Equation~\eqref{eqn:tilde_g}, resulting in valid unit-level sets; see Corollary \ref{corol:all_properties}. 

We model the radial threshold $r_\tau:\mathcal{S}^{d-1}\mapsto \RR_+$ using a multi-layer perceptron (MLP) with rectified linear unit (ReLU) activation functions, denoted by $m_{\boldsymbol{\psi}}:\mathcal{S}^{d-1}\mapsto \RR_+$ and parameterised by the set $\boldsymbol{\psi}$. For brevity, we provide details of the construction in Appendix~\ref{appendix:NN1}; see, also, \cite{Richards2024} for a review of these models and their inference. We similarly represent the (unscaled) gauge function ${g}(\boldsymbol{w})$ using an MLP. To ensure that ${g}(\cdot)$ satisfies ${g}(\boldsymbol{w}) \geq \|\boldsymbol{w}\|_\infty$ for all $\boldsymbol{w}\in\mathcal{S}^{d-1}$, we let ${g}(\boldsymbol{w})=\text{ReLU}\{m_{\boldsymbol{\psi}}(\boldsymbol{w})\}+\|\boldsymbol{w}\|_\infty$, where ReLU($\boldsymbol{x})=(\max\{x_1,0\},\max\{x_2,0\},\dots)^T$ for a vector $\boldsymbol{x}=(x_1,x_2,\dots)^T$ of finite (unspecified) length. Note that we share neither parameters nor architectures between the two neural networks determining $r_\tau(\boldsymbol{w})$ and ${g}(\boldsymbol{w})$.

To transform $g(\cdot)$ to $\tilde{g}(\cdot)$ requires evaluation of all scaling factors $b_i^U$ and $b_i^L, i=1,\dots,d,$ in \eqref{Eq:scaling_factors}. We do so numerically using a sample of angles, denoted by $\mathcal{W}$, that provides a dense coverage of the $(d-1)$-sphere. In practice, we simulate $|\mathcal{W}|=10^6$ points using the rejection sampling algorithm of \citet{neumann1951various}.
 
We fit or \textit{train} the neural network $r_\tau(\boldsymbol{w})$ via minimisation of some loss function, denoted by $\ell(r, r_\tau)$. Consider the set $\{(r_j,\boldsymbol{w}_j)\}_{j=1}^n$, where $r_j$ and $\boldsymbol{w}_j$ are observations of $R$ and $\boldsymbol{W}$, respectively. Optimal estimates of the neural network parameters, denoted $\widehat{\boldsymbol{\psi}}$, can be found by solving the minimisation problem
\begin{equation}
\label{eq:emp_loss}
\widehat{\boldsymbol{\psi}}\in \argmin\limits_{\boldsymbol{\psi}}\frac{1}{n}\sum^n_{j=1}\ell(r_j,r_{\tau}(\boldsymbol{w}_j)),
\end{equation}
where we have suppressed the dependency of $r_\tau(\cdot)$ on $\boldsymbol{\psi}$ in the notation. As $r_\tau(\boldsymbol{w})$ denotes the $\tau$-quantile of $R\mid (\boldsymbol{W}=\boldsymbol{w})$, the most appropriate choice of loss function is the tilted loss, given by 
$l(r,r_\tau)=\rho_\tau(r-r_\tau)$ for $\rho_\tau(z):=z(\tau-\mathbbm{1}\{z < 0\})$ \citep{Koenker2017}. We can also define a suitable loss function and minimisation problem to train the deep representation of the rescaled gauge function, $\tilde{g}(\boldsymbol{w}).$ This model can be described as a conditional density network estimation \citep[see, e.g.,][]{rothfuss2019conditional}, where the loss function is the negative log-likelihood associated with the truncated gamma model defined in \eqref{eqn:trunc_gamma_assum}; note that this is dependent on both the exceedance threshold $r_\tau(\boldsymbol{w}_j)$ and a scaling parameter, $\alpha > 0$. To fit the gauge function model, we replace $\ell(r_j,r_{\tau}(\boldsymbol{w}_j))$ in \eqref{eq:emp_loss} with 
\begin{align}
\label{eq:nll}
\ell\{r_j, \tilde{g}(\boldsymbol{w}_j), \alpha; r_\tau(\boldsymbol{w}_j)\}=-\mathbbm{1}\{r_j > r_\tau(\boldsymbol{w}_j)\}&\big[\alpha\log\{\tilde{g}(\boldsymbol{w}_j)\} + (\alpha-1)\log(r_j) - r_j\tilde{g}(\boldsymbol{w}_j) \nonumber\\
&- \log\{\Gamma(\alpha)\}-\log\{Q(\alpha,\tilde{g}(\boldsymbol{w}_j)r_\tau(\boldsymbol{w}_j))\}\big],
\end{align}
where $Q(\alpha,z)=\Gamma(\alpha,z)/\Gamma(z)$ for $\Gamma(\alpha,z)=\int^\infty_zt^{\alpha-1}\exp(-t)\, dt$. Note that, through an abuse of notation, we have suppressed dependency of $\boldsymbol{\psi}$ on $\alpha$, but this parameter is estimated concurrently with the parameters that comprise the neural network. 

Full inference for our framework is performed in a two-stage fashion. We first train a model for $r_\tau(\boldsymbol{w})$ and derive its estimate, $\hat{r}_\tau(\boldsymbol{w})$. This estimate is then used in a subsequent training step for the rescaled gauge function $\tilde{g}(\cdot)$ and $\alpha$, by replacing $
r_\tau(\boldsymbol{w})$ in \eqref{eq:nll} with its estimated counterpart. For details of the algorithms and regularisation techniques used for training of both models, as well as practical advice for pre-training \citep[e.g.,][]{goodfellow2016deep} of the models using initial estimates for $r_\tau(\cdot)$, see Appendix~\ref{sec:methods:estimation:training}. 

\subsection{Estimating the extended ADF}\label{sec:method:est_ADF}

We now adapt the approach of \citet{Simpson2022}, for estimation of the classical ADF, to permit estimation of our extended ADF, $\Lambda(\cdot)$.  First, given a sample of angles $\{\boldsymbol{w}_j\}_{j=1}^n$, we define the corresponding point estimates for the unit-level set by $\{\tilde{\boldsymbol{x}}_j\}^n_{i=1}$ where ${\tilde{\boldsymbol{x}}_j := (\tilde{x}_{j,1},\hdots,\tilde{x}_{j,d})^T=\boldsymbol{w}_j/\tilde{g}(\boldsymbol{w}_j) \in \widetilde{\partial \mathcal{G}}}$ for $j = 1,\hdots, n$. Then, for any angle $\boldsymbol{w} \in \dsphere \setminus \mathcal{A} $ at which we wish to estimate $\Lambda(\boldsymbol{w})$, we consider the sample $\{\tilde{\boldsymbol{x}}_j\}^n_{i=1}$ to be candidates for the intersection of $\widetilde{\partial \mathcal{G}}$ and the scaled-back set $\Tilde{\mathcal{R}}_{\boldsymbol{w}}$ (see Proposition~\ref{prop:extended_ADF_boundary_link}). The corresponding scaling coefficient $\Tilde{\mathfrak{r}}_{\boldsymbol{w}}$ must satisfy $\Tilde{\mathfrak{r}}_{\boldsymbol{w}} \geq \tilde{\mathfrak{r}}_j$ for all $j=1,\dots,n$, where $\tilde{\mathfrak{r}}_j := ||\boldsymbol{w}||_{\infty}\min_i\{\tilde{x}_{j,i}/w_i \}$; hence, we approximate $\Tilde{\mathfrak{r}}_{\boldsymbol{w}}$ as $\Tilde{\mathfrak{r}}_{\boldsymbol{w}} \approx\max_{j=1,\dots,n}\{\tilde{\mathfrak{r}}_j\}$. Recalling Proposition \ref{prop:extended_ADF_boundary_link}, it follows that an estimate $\Hat{\Lambda}(\boldsymbol{w})$ of the extended ADF  at $\boldsymbol{w}$ is  $\Hat{\Lambda}(\boldsymbol{w}) = ||\boldsymbol{w}||_{\infty}/\max_j\{\tilde{\mathfrak{r}}_j \}$; see \citet{Simpson2022} for further details. Furthermore, since the estimated unit-level set $\widetilde{\partial \mathcal{G}}$ satisfies all of the required theoretical properties for $\partial \mathcal{G}$, the resulting estimate $\Hat{\Lambda}(\cdot)$ must also satisfy the properties of the extended ADF. 

 Estimates $\Hat{\Lambda}(\boldsymbol{w}), \boldsymbol{w} \in \dsphere \setminus \mathcal{A},$ can be used to obtain probability estimates for a wide variety of joint tail regions of $\boldsymbol{X}$. The estimation scheme outlined below avoids the need to sample from, or model, the distribution of the angles $\boldsymbol{W}$, as has been considered in other probability estimation schemes using the geometric representation \citep{Papastathopoulos2024,Wadsworth2024}. 
To begin, let $\boldsymbol{x} \in \RR^d$ be such that $r := ||\boldsymbol{x}||$ is large and define the corresponding angle $\boldsymbol{w} := \boldsymbol{x}/r$. Further define the structure variable $T_{\boldsymbol{w}} := \min_{i=1,\dots,d}\{X_i/w_i\}$, and let $u$ denote a quantile of $T_{\boldsymbol{w}}$ satisfying $q := \Pr(T_{\boldsymbol{w}} \leq u) < \Pr(T_{\boldsymbol{w}} \leq r)$, with $q$ close to 1 and $u>0$ large. Equation \eqref{eqn:wads_tawn_laplace} implies that 
\begin{align}
    &\Pr(\text{sgn}(x_i)X_i > \text{sgn}(x_i)x_i, i = 1,\hdots,d) = \Pr(T_{\boldsymbol{w}} > r) \label{eq:ADF_prob_form}\\
     &= \Pr(T_{\boldsymbol{w}} > r \mid T_{\boldsymbol{w}} > u)\Pr(T_{\boldsymbol{w}} > u) 
     \approx \exp\{ -\Hat{\Lambda}(\boldsymbol{w})(r-u) \} (1-q) \label{eq:ADF_emp},
\end{align}
where $\text{sgn}(\cdot)$ denotes the regular signum function. Figure \ref{fig:wt_model_probs} illustrates examples of joint tail regions in the two dimensional setting that can be estimated using the framework described in equation \eqref{eqn:wads_tawn_laplace}. Observe that these regions are more general than many inference procedures for multivariate extremes, which just focus on the joint survivor or cumulative distribution function \citep[e.g.,][]{Ledford1996,Ramos2009}. 

\begin{figure}[t!]
    \centering
    \includegraphics[width=0.4\textwidth]{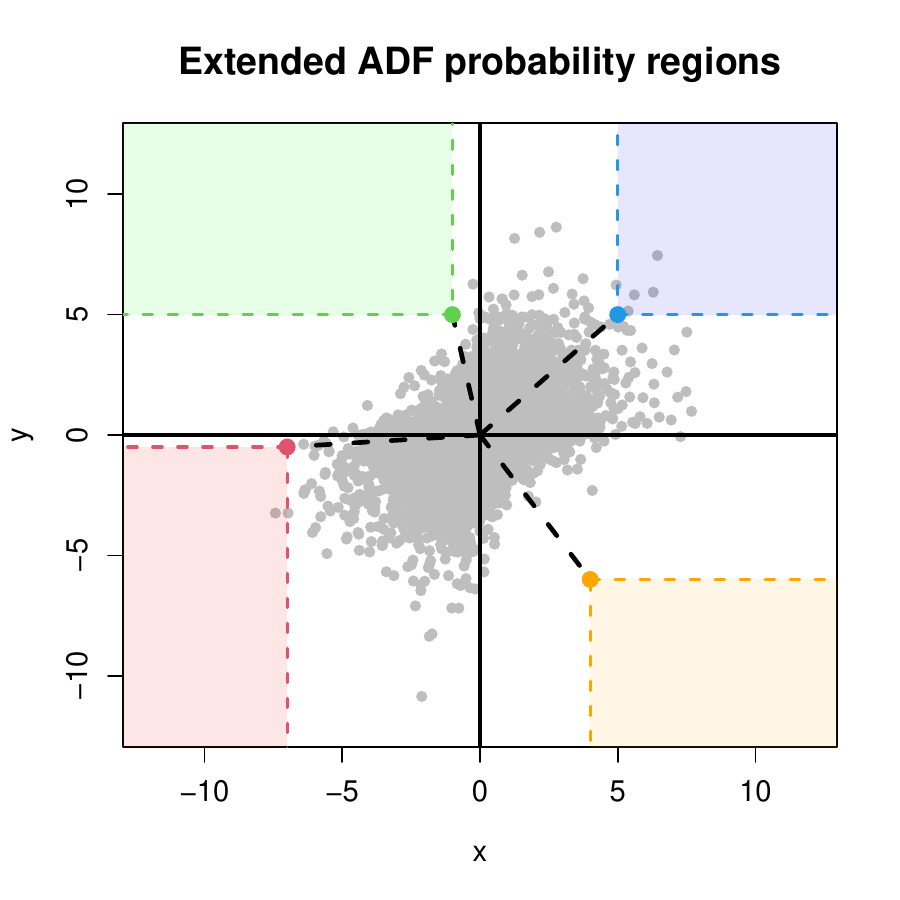}
    \caption{Example probability regions that can be evaluated using the model described in equation \eqref{eqn:wads_tawn_laplace}, with data simulated from a Gaussian copula on Laplace margins. The blue, green, red, and orange regions denote the sets $[5,\infty) \times [5,\infty)$, $(-\infty,-1] \times [5,\infty)$, $(-\infty,-7] \times (-\infty,-0.5]$, and $[4,\infty) \times (-\infty,-6]$, respectively.} 
    \label{fig:wt_model_probs}
\end{figure}

\subsection{Assessing goodness-of-fit for DeepGauge model fits}\label{sec:inference:gof}
We now discuss several diagnostic tools for assessing goodness-of-fit under the DeepGauge modelling framework, which we utilise for the case study in Section \ref{sec:app}. The first three tools correspond to quantile-quantile (QQ) plots on a unit exponential scale. Comparing quantiles on this scale is common within extreme value analyses, as it allows one to assess how well a given model captures the most extreme observations \citep[see, e.g.,][]{Coles2001,heffernan2001extreme}. The latter two tools described below are for visualisation of unit-level sets and the extended ADF in low ($d\leq3$) and high ($d>3$) dimensional settings.

\paragraph{Truncated gamma QQ plot:} 
We consider the diagnostic proposed by \citet{Wadsworth2024} to assess the validity of the modelling assumption described in equation \eqref{eqn:trunc_gamma_assum}. Here, the fitted truncated gamma model is used to transform all of the threshold-exceeding observations, i.e., $R\mid (\boldsymbol{W} =\boldsymbol{w}, R>r_\tau(\boldsymbol{w}) )$, to a unit exponential scale, and we compare the observed and theoretical quantiles via a QQ plot. This provides a global diagnostic for the model fit over the entire angular domain. 

\paragraph{Extended ADF diagnostic:} 
We adapt the diagnostic proposed by \citet{Murphy-Barltrop2024c} to assess goodness-of-fit with respect to the estimated extended ADF. Note that equation \eqref{eqn:wads_tawn_laplace} can be reformulated as follows: given any angle $\boldsymbol{w} \in \dsphere \setminus \mathcal{A}$, we have that ${T^u_{\boldsymbol{w}}  \sim \text{Exp}(\Lambda(\boldsymbol{w}))}$ as $u \to \infty$ for exceedances $T^u_{\boldsymbol{w}} := (T_{\boldsymbol{w}} - u )\mid (T_{\boldsymbol{w}} > u)$, with $T_{\boldsymbol{w}}$ defined in Section \ref{sec:estimation_gauge_adf}. Therefore, for sufficiently large $u$, the random variable $E := -\log[ 1 - (1 - \exp\{-\Lambda(\boldsymbol{w})T^u_{\boldsymbol{w}} \}) ] = \Lambda(\boldsymbol{w})T^u_{\boldsymbol{w}}$ approximately follows a unit exponential distribution, independent of the choice of $\boldsymbol{w}$. \citet{Murphy-Barltrop2024c} exploit this fact and compute min-projection exceedances for a fixed angle, before transforming these exceedances to the standard exponential scale. One can then evaluate the alignment of the exponential quantiles on a QQ plot. This results in a local diagnostic and allows us to test the validity of the modelling assumption described in equation \eqref{eqn:wads_tawn_laplace}. The corresponding algorithm for computing the diagnostic is given in Appendix~\ref{appen:extend_ADF}. 

\paragraph{Return level sets and probabilities:} 
Using our fitted model, we can obtain estimates of return level sets, which provide a summary of the joint extremal behaviour of random vectors and are used extensively in practice for design sensitivity analysis \citep[see, e.g.,][]{Mackay2021,Papastathopoulos2024,Simpson2024}. For a probability $p$, a return level set is defined as any region $\mathcal{B}\subset \RR^d$ satisfying $\Pr(\boldsymbol{X} \in \mathcal{B}) = p$. When $\mathcal{B}$ is centred at the origin $\boldsymbol{0}_d$, it can be obtained directly from the modelling framework outlined in Section \ref{sec:estimation_gauge_adf}; for each angle $\boldsymbol{w} \in \mathcal{S}^{d-1}$, set 
\begin{equation*}
    r_p(\boldsymbol{w}) = F^{-1}_\Gamma\left[\left( \frac{p - \tau}{1-\tau}\right) \times  \bar{F}_\Gamma\left(r_{\tau}\left(\boldsymbol{w}\right) ; \alpha, g\left(\boldsymbol{w}\right)\right) + F_\Gamma\left(r_{\tau}\left(\boldsymbol{w}\right) ; \alpha, g\left(\boldsymbol{w}\right)\right) ;\alpha, g\left(\boldsymbol{w}\right)\right],
\end{equation*}
where $F_\Gamma(\cdot)$ denotes the gamma distribution function. This corresponds to the estimated $p$-quantile of the conditional distribution $R \mid (\boldsymbol{W} = \boldsymbol{w})$ under model \eqref{eqn:trunc_gamma_assum}. Using these radial quantiles, define the Cartesian set  
\begin{equation*}
    \mathcal{B}_p := \left\{ \boldsymbol{x} \in \RR^d \setminus \{ \boldsymbol{0}_d\} : ||\boldsymbol{x}|| \leq r_p\left( \boldsymbol{x}/||\boldsymbol{x}|| \right) \right\} \bigcup \big\{ \boldsymbol{0}_d\big\}.
\end{equation*}
Assuming unbiased estimation, the total law of probability implies that $\Pr(\boldsymbol{X} \in \mathcal{B}_p) = p$ \citep{Papastathopoulos2024}. To assess the accuracy of return level set estimates, we propose the following diagnostic. First, observe that for any observation $\boldsymbol{x}_j \neq \boldsymbol{0}_d$, we have $\boldsymbol{x}_j \in \mathcal{B}_p \iff ||\boldsymbol{x}_j|| \leq r_p(\boldsymbol{x}_j/||\boldsymbol{x}_j||)$. Consequently, for a sample of non-zero observations $\{\boldsymbol{x}_j\}^n_{j=1}$, an empirical estimate $\hat{p}$ of $\Pr(\boldsymbol{X} \in \mathcal{B}_p)$ is $\hat{p} := (1/n)\sum_{j=1}^n \mathbbm{1}(||\boldsymbol{x}_j|| \leq r_p(\boldsymbol{x}_j/||\boldsymbol{x}_j||))$. Our diagnostic then plots the pairs $( -\log(1-p), -\log(1-\hat{p}))$, alongside tolerance bounds, for a subset of increasing probabilities close to 1. This approach for assessing goodness-of-fit provides a multivariate extension of `return level plots'; see, e.g., \citet{Coles2001}.   

\paragraph{Three-dimensional unit-level and extended ADF sets:} For $d\leq 3$, we can plot the scaled sample clouds $\{\boldsymbol{x}_j/\log(n/2) : j = 1,\hdots,n \}$ against the estimated unit-level and extended ADF sets; see, e.g., Figure \ref{fig:true_gauges}. For large enough $n$, we would expect the scaled sample clouds to lie approximately within the interiors of the estimated unit-level and extended ADF sets. Plotting scaled observations against the estimated unit-level sets has also been used for validation by, e.g., \citet{Majumder2024, Papastathopoulos2024}. In addition to being a visual indicator of how well the estimated shapes capture the complex extremal dependence features of data, one can also use these plots to verify that the estimated unit-level and extended ADF sets satisfy all of the theoretical properties discussed in Sections~\ref{sec:intro} and \ref{sec:theory}.

\paragraph{Bivariate unit-level set slices:} 
For $d \geq 4$, we cannot visualise unit-level sets. Furthermore, interpretation is challenging for $d = 3$ unless one can freely alter the perspective angle, e.g., using computational software. To account for this shortcoming, we propose considering bivariate slices of the estimated unit-level sets.  Specifically, given indices $ (\mathfrak{i},\mathfrak{j}),$ with $1\leq \mathfrak{i} < \mathfrak{j} \leq d$, consider the set of points 
\begin{equation*}
    \widetilde{\partial \mathcal{G}}_{\mathfrak{i},\mathfrak{j}} := \left\{ (w_\mathfrak{i},w_\mathfrak{j})/\Tilde{g}(\boldsymbol{w}) : \boldsymbol{w} \in \dsphere, w_k = 0, k \in \mathcal{D}\setminus \{\mathfrak{i},\mathfrak{j}\} \right\} \subset [-1,1]^2.
\end{equation*}
 It is important to note that $\widetilde{\partial \mathcal{G}}_{\mathfrak{i},\mathfrak{j}}$ is not the bivariate unit-level set for the vector $(X_{\mathfrak{i}},X_{\mathfrak{j}})$; rather, $\widetilde{\partial \mathcal{G}}_{\mathfrak{i},\mathfrak{j}}$ is a bivariate projection from the subset of $\widetilde{\partial \mathcal{G}}$ for which all angles indexed by $\mathcal{D}\setminus \{\mathfrak{i},\mathfrak{j}\}$ are equal to 0. This projection is illustrated in Figure \ref{fig:bivariate_slice} for a three-dimensional unit-level set with $(\mathfrak{i},\mathfrak{j})=(1,2)$. 

\begin{figure}
    \centering
    \begin{subfigure}{0.45\textwidth}
        \centering
        \includegraphics[width=\textwidth]{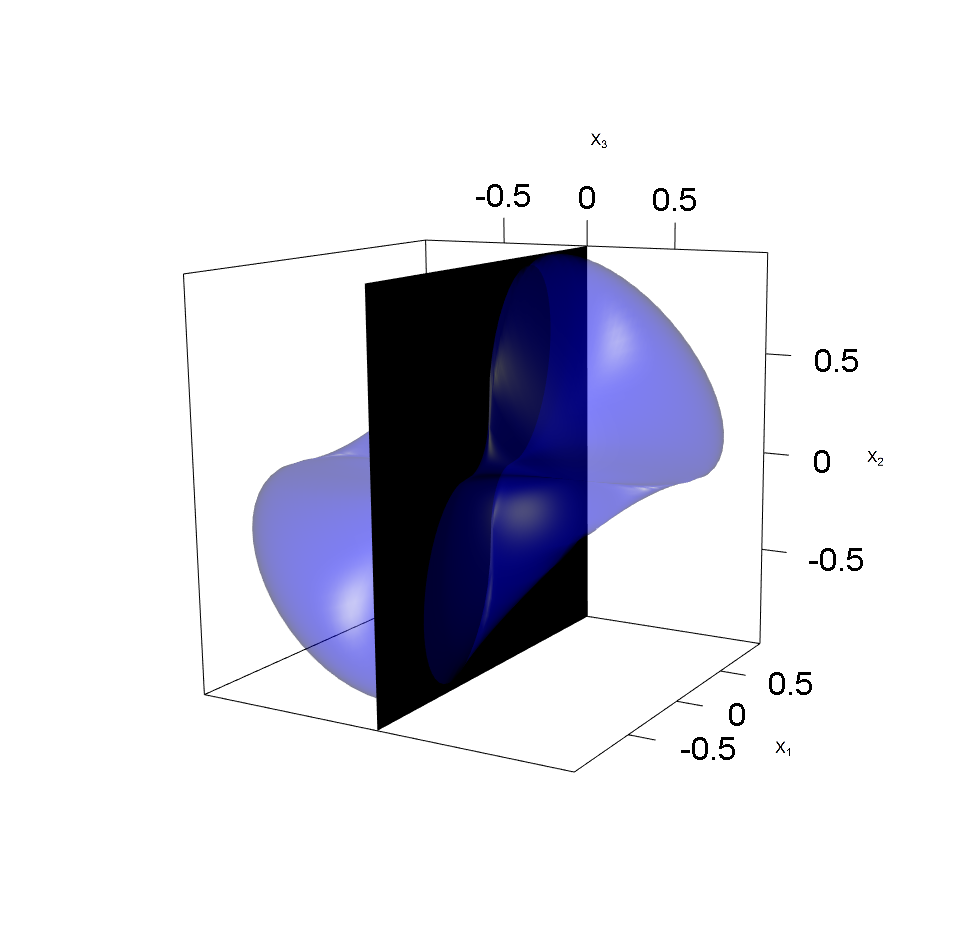}
    \end{subfigure}
    \quad
    \quad
    \begin{subfigure}{0.4\textwidth}
        \quad
        \quad
        \centering 
        \includegraphics[width=\textwidth]{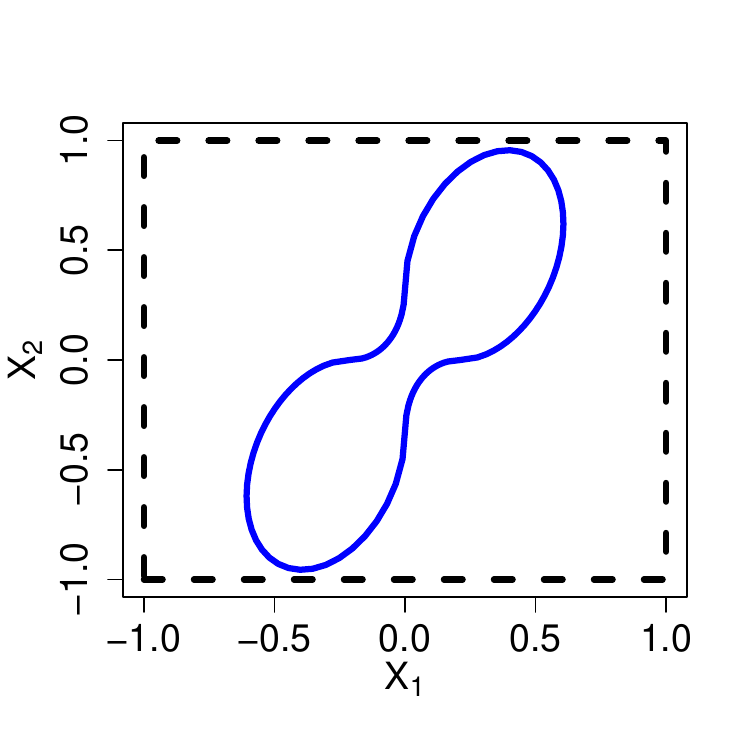}
    \end{subfigure}
    \caption{Pedagogical example of a bivariate slice from a three-dimensional unit-level set.}
    \label{fig:bivariate_slice}
\end{figure}

In practice, we observe very few angles in the region $\{ \boldsymbol{w} \in \dsphere : w_k = 0, k \in \mathcal{D}\setminus \{\mathfrak{i},\mathfrak{j}\}\}$; however, we will typically observe a significant number of angular observations that lie close to this region. Moreover, note that for any observation $\boldsymbol{x}$ with corresponding angular observation $\boldsymbol{w}$, we have $w_k = 0 \iff x_k/\log(n/2) = 0 $ for any $k \in \mathcal{D}$. Therefore, we plot, alongside  $\widetilde{\partial \mathcal{G}}_{\mathfrak{i},\mathfrak{j}}$, all observations of the the scaled bivariate sample clouds for which $||\boldsymbol{x}_{-\{\mathfrak{i},\mathfrak{j}\}}||/\log(n/2)  \leq \epsilon $, where $\boldsymbol{x}_{-\{\mathfrak{i},\mathfrak{j}\}}$ denotes the observation with its $\mathfrak{i}$-th and $\mathfrak{j}$-th components removed and $\epsilon >0$ denotes some small value. For sufficiently small $\epsilon$ and large enough $n$, one would expect the scaled bivariate sample cloud to lie approximately within the interior of the estimated bivariate slice. Selection of $\epsilon$ is considered in Section~\ref{sec:app}.

\section{Simulation study}\label{sec:sim_study}
\subsection{Overview}
We conduct two simulation studies to investigate the efficacy of the DeepGauge framework for estimating the gauge functions on $\dsphere$ (equivalently, the unit-level sets) of three known copulas: Gaussian, Student-t, and logistic. Details of the copulas and their theoretical gauge functions are provided in Section~\ref{sec:sim_study:models}.

Two studies are performed and their results presented in Section~\ref{sec:sim_study:results}. Efficacy in both studies is quantified using the validation diagnostics described in Section~\ref{sec:sim_study:models}. The first study considers, for a fixed quantile level $\tau$ and architecture, the effect of varying $d$ and $n$ on the accuracy of DeepGauge estimates of the gauge functions on $\dsphere$ and exceedance probabilities. The second study considers, for a fixed dimension $d$ and sample size $n$, the effect of hyper-parameter choice. In both studies, we perform $100$ experiments for every model specification. That is, for a single choice of $n$, $d$, copula, and hyper-parameter configuration, we simulate 100 data sets and apply the methodology, separately, to each. Performance metrics are then presented as the median and $95\%$ confidence intervals over all experiments. All models for both $r_\tau(\cdot)$ and $\tilde{g}(\cdot)$ are trained over 500 iterations with a mini-batch size of 1024, and the ADF function is estimated with $u$ in \eqref{eq:ADF_emp} taken to be the $q=0.9995$ empirical quantile of the structure variable $T_{\boldsymbol{w}}$. Early-stopping is used with a patience of $\Delta=5$; see Appendix~\ref{sec:methods:estimation:training} for details. 
\subsection{Models and performance metrics}\label{sec:sim_study:models}
We consider three copulas with known gauge functions and limit sets: Gaussian, Student-t, and logistic; the theoretical unit-level sets and gauge functions associated with them are detailed in \cite{Papastathopoulos2024}. If the random vector $\boldsymbol{X}$ also has a $d$-variate Gaussian copula with positive-definite precision matrix $Q\in \mathbb{R}^{d \times d}$, its gauge function is $g(\boldsymbol{x})=\left\{\text{sgn}(\boldsymbol{x})|\boldsymbol{x}|^{1/2}\right\}^TQ\left\{\text{sgn}(\boldsymbol{x})|\boldsymbol{x}|^{-1/2}\right\}$, where $\text{sgn}(\cdot)$ is applied component-wise. If $\boldsymbol{X}$ follows a $d$-variate Student-t copula with $Q$ as before and degrees of freedom $\nu>0$, its gauge function is $g(\boldsymbol{x})=-\nu^{-1}\sum^d_{i=1}|x_i|+(1+d\nu^{-1})\max_{i=1,\dots,d}|x_i|$. Finally, let $\boldsymbol{X}$ follow a logistic copula, i.e., a Gumbel copula with dependence parameter $1/\theta$ for $\theta \in (0,1]$. In this case, the form of the theoretical gauge function is cumbersome; see Appendix A.4 of \cite{Papastathopoulos2024}. For the Gaussian and Student-t copulas, we randomly generate, for each $d$, a valid precision matrix. However, we impose that the corresponding correlation matrices are ordered with respect to $d$. That is, if we consider two copulas with dimension $d_1$ and $d_2$, satisfying $d_1<d_2$, then the $d_1 \times d_1$ correlation matrix is a submatrix of the $d_2\times d_2$ correlation matrix. Note that the generated correlation matrices are kept fixed throughout all studies, and we further set $\nu=1$ and $\theta=0.3$ 

We simulate data from these copulas and apply the methodology described in Section~\ref{sec:methods} to estimate the gauge function on $\dsphere$. To quantify the accuracy of the estimated unit-level set, we measure the integrated squared error (ISE) of the form 
\begin{equation}\label{eq:ISE}
\mbox{ISE}:=\int_{\dsphere}\{1/g_0(\boldsymbol{w})-1/\tilde{g}(\boldsymbol{w})\}^2\mathrm{d}\boldsymbol{w},
\end{equation}
where $g_0$ and $\tilde{g}$ are, respectively, the theoretical and estimated gauge functions. We evaluate the integral in \eqref{eq:ISE} using Monte-Carlo methods; specifically, we use the estimator $\widehat{\mbox{ISE}}:=(A_d/|\mathcal{W}|)\sum_{\boldsymbol{w} \in \mathcal{W}}\{1/g_0(\boldsymbol{w})-1/\tilde{g}(\boldsymbol{w})\}^2$, where $A_d:=2\pi^{d/2}/\Gamma(d/2)$ is the surface area of $\mathcal{S}^{d-1}$, $\mathcal{W}$ is the set of points described in Section~\ref{sec:method:neuralnets}, and $\Gamma(\cdot)$ is the gamma function.

We validate estimates of the extended angular dependence function (ADF) by evaluating, for each considered copula, four joint probabilities. The first two are exceedance probabilities of the form $\Pr\{X_i > u, i=1,\dots,d\},$ with $u$ taken to be both the $0.99$ and $0.999$ standard Laplace quantiles. For the joint lower tail, we also evaluate $\Pr\{X_i < u, i=1,\dots,d\},$ with $u$ taken as the $0.01$ and $0.001$ standard Laplace quantiles. One of the strengths of the DeepGauge framework is that it is not limited to estimation of probabilities on hyper-cubes. To show this, we also estimate
\[
\Pr\{ u_{1,i}< X_i < u_{2,i}, i=1,\dots,d\},
\]
where $u_{1,i}=u_1, u_{2,i}=\infty $ if $i$ is odd and, otherwise, $u_{1,i}=- \infty, u_{2,i}=u_2$, where $u_1$ corresponds to the $0.999$ standard Laplace quantile. For $u_2$, we consider two cases: $u_2$ as either the $0.2$ or $0.4$ standard Laplace quantile. We consider estimation of the two latter probabilities for only the Gaussian and Student-t copulas, as evaluation is permissible using standard computational software. This is not the case for the logistic copula, and we found that Monte-Carlo methods perform poorly. To quantify efficacy, we provide the mean absolute log error (MALE), taking the average over all probabilities. The order-of-magnitude for the probabilities ranges from approximately $10^{-3}$ to $10^{-37}$. 
\subsection{Results}\label{sec:sim_study:results}
 We first consider a setting with the quantile level fixed to $\tau=0.75$ and the architecture for $g(\cdot)$  (equivalently, $\tilde{g}$) taken to be to a neural network with $N=3$ hidden layers of width 64 (see Appendix \ref{appendix:NN1} for details), but allow the sample size $n$ and dimension $d$ to vary: we consider $d=3$, $d=5$, and $d=8$, as well as $n=10,000$, $n=50,000$, $n=100,000$, and $n=250,000$. For $r_\tau(\cdot)$, we use a neural network with $N=3$ hidden layers of width $32$ and keep this fixed throughout. Results for this setting are provided in Figure~\ref{fig:sim_study3} of Appendix~\ref{appendx:supfigs}. As expected, we generally observe decreasing ISE and MALE with increasing sample size $n$ across all dimensions and copulas. 
 
 We then consider a setting with the sample size $n$ and dimension $d$ fixed to investigate how the choice of hyper-parameters impact the model fits. We fix $n=100,000$ and $d=5$ or $d=8$, but, across two scenarios, we allow the quantile level and architecture of $g(\cdot)$ to vary. For the first scenario, we consider a sequence of quantile levels $\tau\in\{0.1,0.3,\dots,0.9\}$ and take $g(\cdot)$ to be a neural network with $N=3$ hidden layers and with $32$ nodes per layer.  For the second scenario, we fix $\tau=0.75$ and vary the architecture for $g(\cdot)$. We consider eight neural networks with $N=1$ to $N=4$ hidden layers, and with consistent width across layers; we take this to be either 16 or 64. Results for the first and second scenarios are presented in Figure~\ref{fig:sim_study1} and Table~\ref{tab:simstudy2}, respectively.

Figure~\ref{fig:sim_study1} illustrates the simulation results for varying quantile level, $\tau$. We find that the optimal value of $\tau$ (in terms of minimising the performance metrics) differs across the copulas and dimension. For the Gaussian and Student-t copulas, we find that higher values of $\tau$ are preferable when estimating the unit-level set, as the ISE is reduced; the converse holds for the logistic copula. When considering the MALE, we find that the Gaussian and Student-t copulas tend to favour a lower value of $\tau$, particularly for the high dimensional case where $d=8$. Conversely, for the logistic copula, the MALE is minimised with the largest value of $\tau$. We note that the logistic copula has the quickest rate of convergence to the truncated gamma model described in \eqref{eqn:trunc_gamma_assum} \citep{Wadsworth2024}, which may explain why this copula benefits from a smaller $\tau$.

Table~\ref{tab:simstudy2} in Appendix~\ref{appendx:supfigs} presents the results for a fixed quantile level ${\tau=0.75}$, but with varying architecture for $g(\cdot)$. We find that, for the Gaussian and Student-t copulas, the ISE is minimised when using the more complicated architecture, e.g., 3--4 layers of width 64; the converse holds for the logistic copula. Note that the logistic copula is determined by a single parameter, and so its corresponding gauge function has a much simpler shape that the other copulas (see Figure \ref{fig:true_gauges} for the case when $d=2$). Hence, it is unsurprising that the logistic copulas favours a simpler neural network when considering the ISE. For the MALE, the results are less clear: all three copulas generally favour an architecture with fewer parameters (and, hence, a larger effective sample size), except the Gaussian copula with $d=5$. In applications, the shape of the gauge function that arises from the data generating process is likely to be quite complex. Figure~\ref{fig:sim_study1} and Table~\ref{tab:simstudy2} suggest that, if the goal is to accurately estimate the unit-level set $\partial \mathcal{G}$, then one should consider using a more complicated neural network architecture for $g(\cdot)$. However, a more conservative approach should be taken if we wish to accurately estimate tail probabilities; reducing the quantile level $\tau$ and using a simpler neural network architecture will increase the effective sample size used in inference for model \eqref{eqn:trunc_gamma_assum}, and this improve estimation of tail probabilities. 

 \begin{figure}[t!]
\centering
    \includegraphics[width=\linewidth]{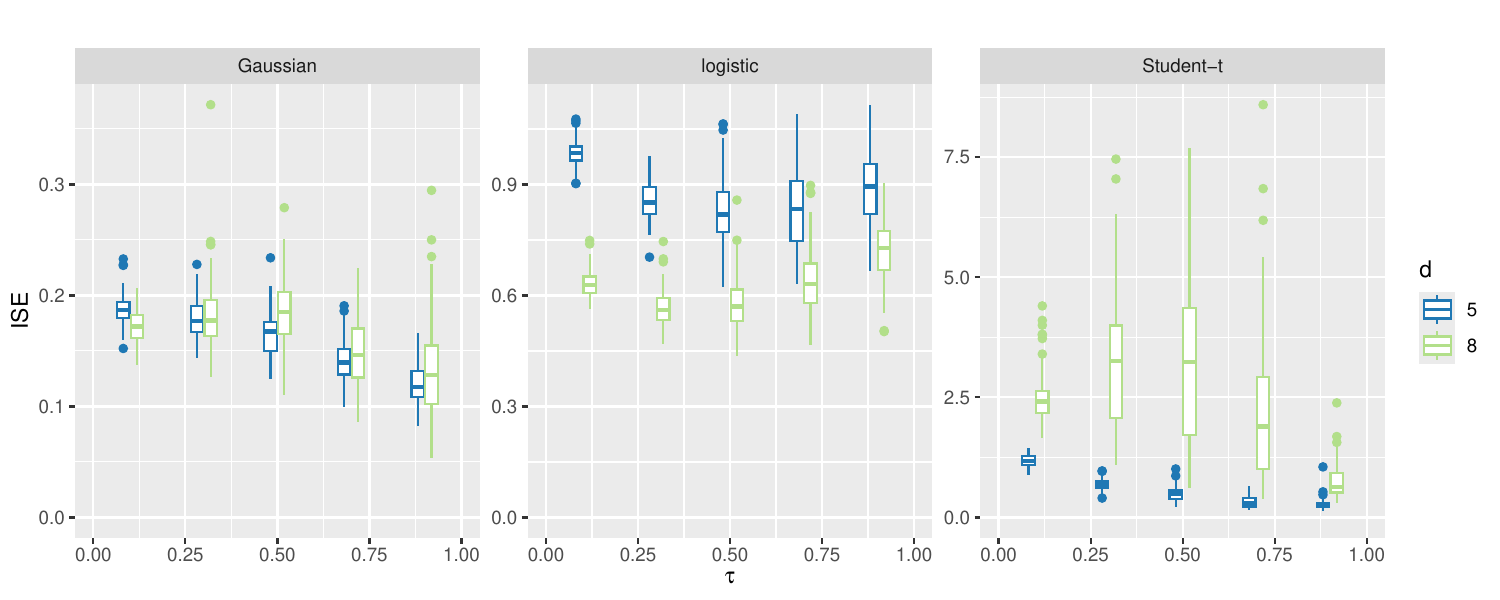}
    \includegraphics[width=\linewidth]{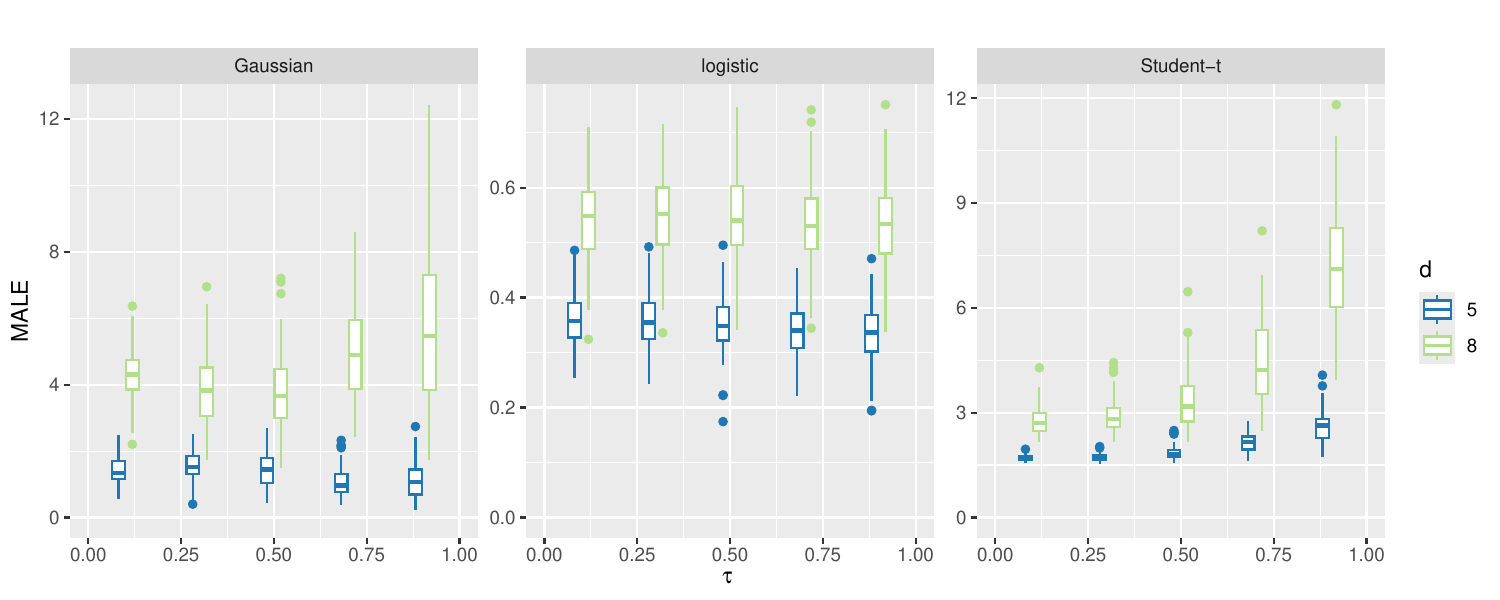}
    \caption{Boxplots of the estimated ISE and MALE for the quantile-level, $\tau$, simulation study described in Section~\ref{sec:sim_study:models}. Here  the sample size is $n=100,000$, the architecture for the rescaled gauge function has $N=3$ hidden layers, each of width 32, and the considered dimensions are $d=5$ and $d=8$.}
    \label{fig:sim_study1}
\end{figure}

Generally, we find some variation in model fits across different choices of architecture and quantile level $\tau$, but the results are broadly similar. From this study, we can conclude that there is no one-size-fits-all approach to fully optimising hyper-parameters of the the DeepGauge framework. In practice, where the true gauge function and underlying data generating process are unknown, we advocate that practitioners fit a collection of DeepGauge models with various quantile levels and architectures, and choose the best model fit using the goodness-of-fit metrics detailed in Section~\ref{sec:inference:gof}.


\section{Analysis of the NORA10 Metocean Data}\label{sec:app}
\subsection{Overview} \label{subsec:case:overview}
\begin{figure}[t!]
\centering
    \includegraphics[scale=0.6]{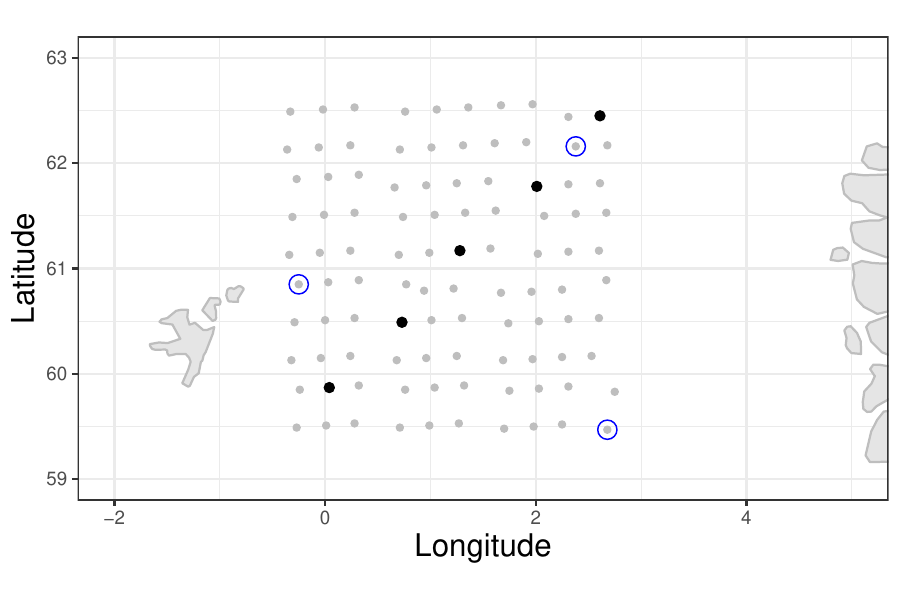}
        \caption{Locations used to study severe ocean events. Locations where we study the joint extremal dependence of \texttt{ws}, \texttt{hs}, and \texttt{mslp} are circled in blue. Locations for studying the joint extremal dependence in \texttt{hs} are highlighted in black.}
    \label{fig:sites_all}
\end{figure}
We demonstrate the efficacy of the DeepGauge framework by modelling extremal dependence among metocean variables associated with severe ocean events. Such events typically occur when multiple metocean variables are simultaneously extreme, and pose a risk to offshore and coastal structures, such as wind farms and oil platforms \citep{Shooter2021,Shooter2022}. Robust analyses of joint extremes is therefore crucial for informed decision making.  

Our study uses the NORA10 hindcast gridded data set \citep[NOrwegian ReAnalysis 10km,][]{Reistad2011}, which provides three-hourly wave fields over the Norwegian Sea, the North Sea, and the Barents Sea, at a spatial resolution of 10km. We focus on an area between the British Isles and Iceland for the period between September 1957 and December 2009 (inclusive), which amounts to $n=152,917$ observations. These data have also previously been analysed by \citet{Shooter2022}. We run two analyses to illustrate the efficacy of the DeepGauge framework for dimensions $d=3$ and $d=5$. For both analyses, hyper-parameters (e.g., model architecture and quantile level $\tau$) were optimised with respect to the goodness-of-fit diagnostics described in Section \ref{sec:inference:gof}. As our interest is in modelling extremal dependence and not marginal extremes, the data are apriori transformed to standard Laplace margins using a standard rank transform, applied independently at each spatial location and for each variable. All study locations are plotted in Figure~\ref{fig:sites_all}.

First, we consider the joint behaviour of wind speed (\texttt{ws}; m/s), significant wave height (\texttt{hs}; m), and mean sea level pressure (\texttt{mslp}; hPa) separately at three locations (01, 46, and 85, from south to north), outlined in blue in Figure \ref{fig:sites_all}. These three variables are associated with severe ocean events \citep{ewans2014recent}, and we aim to model their joint tail behaviour separately at each of the three locations. Of the $d=3$ variables, \texttt{ws} and \texttt{hs} exhibit positive dependence, while the other two pairs of variables exhibit negative dependence; see Figure~\ref{fig:subgauges_1}. Given the complex dependence structure of these data, it is likely that standard parametric extremal dependence models will perform poorly here; hence, these data provide the perfect candidate for illustrating the flexibility of our DeepGauge model.

Our second analysis models the joint tail behaviour of significant wave height (\texttt{hs}) across five locations (16, 33, 58, 72, and 92; highlighted in black, Figure~\ref{fig:sites_all}) along a transect that runs approximately south-west to north-east, representing storm fronts that move along that direction. Simultaneous extremes at multiple locations pose higher risks than extremes at individual locations \citep{Simpson2021a}, and so quantification of joint extreme risk across these $d=5$ locations represents an important area for investigation. 

For the $d=3$ and $d=5$ cases, we use $\tau = 0.9$ and $\tau = 0.4$, respectively. The neural networks for $r_{\tau}(\cdot)$ and $\tilde{g}(\cdot)$ in the $d=3$ setting have identical architectures: three hidden layers, each with 64 neurons. For $d=5$, both neural networks have two hidden layers; the neural networks for $r_{\tau}(\cdot)$ and $\tilde{g}(\cdot)$ have all widths 16 and 64, respectively. Since the data exhibit strong temporal dependence, uncertainty is quantified via block bootstrapping: this sampling scheme retains temporal dependence in the resampled data set, ensuring the additional uncertainty that arises due to lower effective sample sizes is accounted for \citep{Politis1994}. Moreover, the block size set to $4$ days, since this length is appropriate for capturing the temporal dependence typically observed in metocean data sets \citep{Murphy-Barltrop2024}. 

\subsection{Results}
\subsubsection{Joint distribution of \texttt{hs}, \texttt{ws}, and \texttt{mslp}: $d=3$}
\begin{figure}
\centering
    \includegraphics[scale=.4]{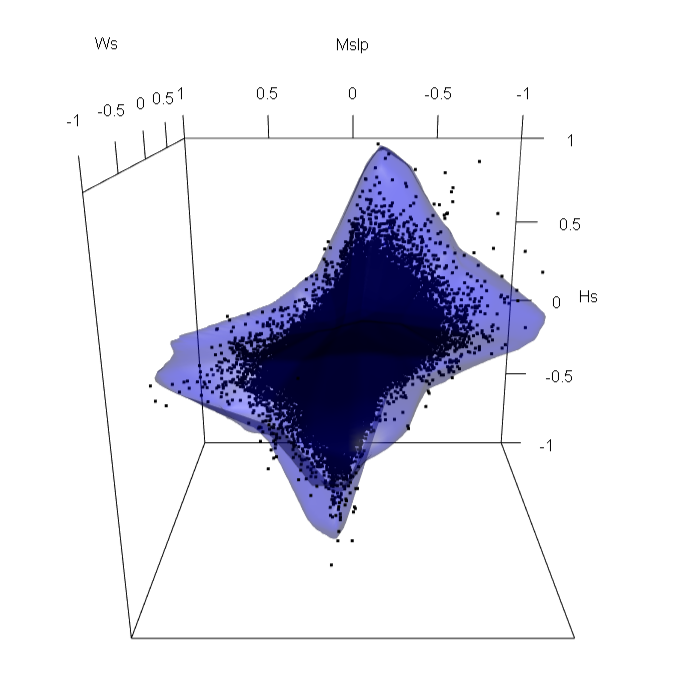}
    \includegraphics[scale=.4]{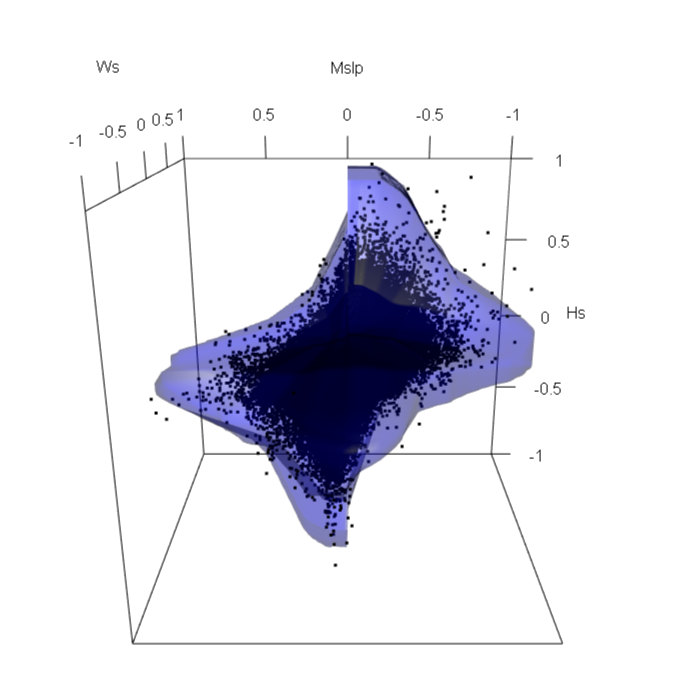}
         \caption{The estimated unit-level (left) and extended ADF (right) sets for \texttt{hs}, \texttt{ws}, and \texttt{mslp} at location 01 (on standard Laplace margins).}
    \label{fig:estimates_1}
\end{figure}

Here we present results for the three metocean variables observed at location 01; results for the remaining two locations are presented in Appendix~\ref{appendx:supfigs}. Figure~\ref{fig:estimates_1} plots the estimates of the three-dimensional unit-level set and the extended ADF for the three variables, with the sample clouds transformed to standard Laplace margins. We observe that the estimated unit-level set captures well the shape of the scaled sample cloud. We also plot bivariate unit-level set slices for each of the three pairs of variables in Figure~\ref{fig:subgauges_1}.
For these, we set $\epsilon = 0.01$ (see Section \ref{sec:inference:gof}), resulting in a reasonable amount of validation data for each unit-level set slice. We observe that the estimated slices closely match the shapes of the corresponding bivariate sample clouds, again indicating accurate model fits. Corresponding plots for the remaining two locations are provided in Figures~\ref{fig:estimates_46}--\ref{fig:subgauges_85} (Appendix~\ref{appendx:supfigs}). 

Goodness-of-fit is further verified using the truncated gamma QQ plot, the extended ADF diagnostic evaluated at $\boldsymbol{w} = \boldsymbol{1}_d/||\boldsymbol{1}_d||$, and the estimated return level set probabilities; see Section \ref{sec:inference:gof}. The plots, provided in Figure~\ref{fig:qqplot_1} of Appendix~\ref{appendx:supfigs}, illustrate strong agreement between model and empirical quantiles across all three diagnostics. The goodness-of-fit analyses for the remaining two locations also indicate good model fits over all diagnostics (Figures~\ref{fig:qqplot_46}--\ref{fig:qqplot_85} of Appendix~\ref{appendx:supfigs}). For all locations, we also provide animated three dimensional plots for the unit-level sets and extended ADFs in the Supplementary Material. 

\begin{figure}
\centering
    \includegraphics[width=\linewidth]{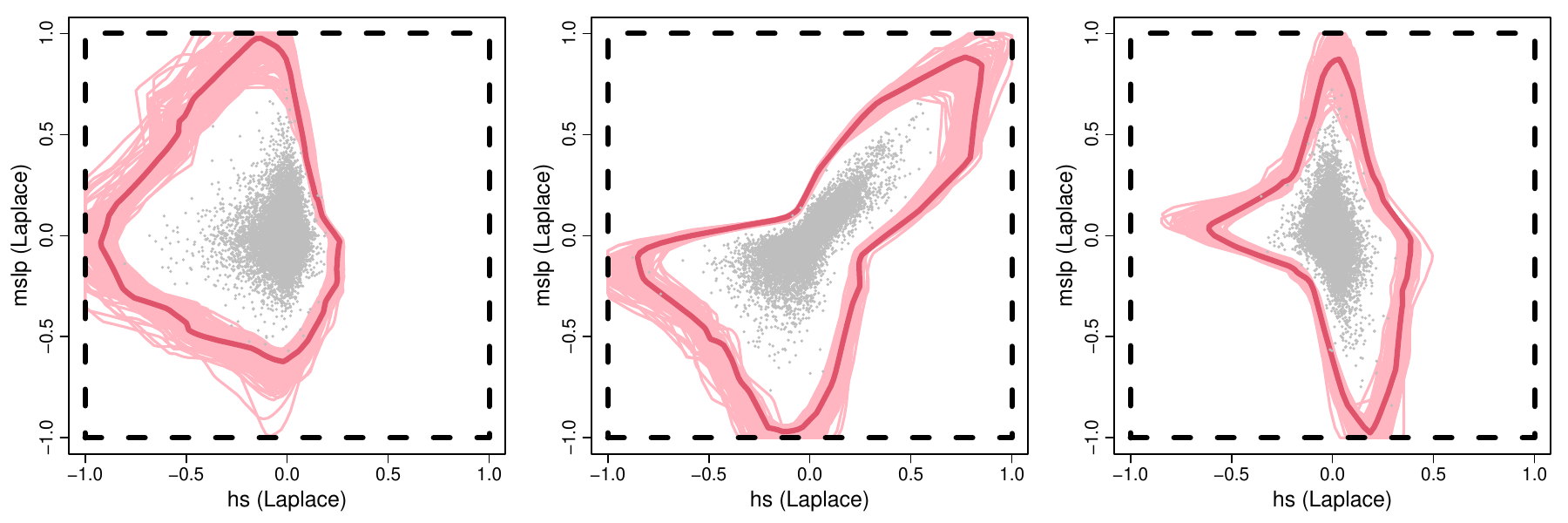}
         \caption{Scaled bivariate sample clouds with $\epsilon = 0.01$ for pairwise combinations of \texttt{ws}, \texttt{hs}, and \texttt{mslp} at location 01 on standard Laplace margins; the red lines describe the estimated bivariate unit-level set slices.}
    \label{fig:subgauges_1}
\end{figure}

Finally, to demonstrate a practical use case of the three dimensional model fits, we use the estimated extended ADFs, obtained from the estimated unit-level sets, to estimate probabilities for joint tail regions. These tail regions, or hypercubes, are illustrated in Figure \ref{fig:hypercubes} for metocean data of location $01$ (on Laplace margins), noting that the same hypercubes are considered across all three locations. The orange, blue and purple hypercubes correspond to the tail regions $\{\texttt{hs}_L > u_{0.99}, \texttt{ws}_L > u_{0.99}, \texttt{mslp}_L < -u_{0.99} \}$, $\{\texttt{hs}_L > u_{0.99}, \texttt{ws}_L < -u_{0.99}, \texttt{mslp}_L < -u_{0.99} \}$, and $\{\texttt{hs}_L < -u_{0.99}, \texttt{ws}_L > u_{0.99}, \texttt{mslp}_L < -u_{0.99} \}$ respectively, where $u_{0.99}$ denotes the 0.99 quantile of the standard Laplace distribution and $\cdot_L$ denote each metocean variable on Laplace margins. As noted previously, high wind speeds and wave heights combined with low sea level pressure are typical conditions for the most extreme storm events \citep{Ewans2014,Lobeto2024}; this corresponds with the orange hypercube in Figure \ref{fig:hypercubes}. Conversely, owing to the underlying physics of metocean phenomena, one would expect to observe very little data in the remaining hypercubes. Bootstrapped probability estimates for each hypercube and location are given in Appendix~\ref{appendx:supfigs}; the results agree with what one would expect, in the sense that the estimates for the orange hypercubes are many orders of magnitude greater than the blue and purple regions. For example, at location $01$, the bootstrapped probability estimates for the orange region exist within the interval $[2.9\times 10^{-4},8.2\times 10^{-4}]$ (i.e., rare but feasible within a reasonable time window), whereas the corresponding estimates for the remaining regions are all smaller than $2.73\times 10^{-8}$. This indicates the estimated probabilities are accurately representing the observed tail structure. 
\begin{figure}
\centering
    \includegraphics[width=.5\linewidth]{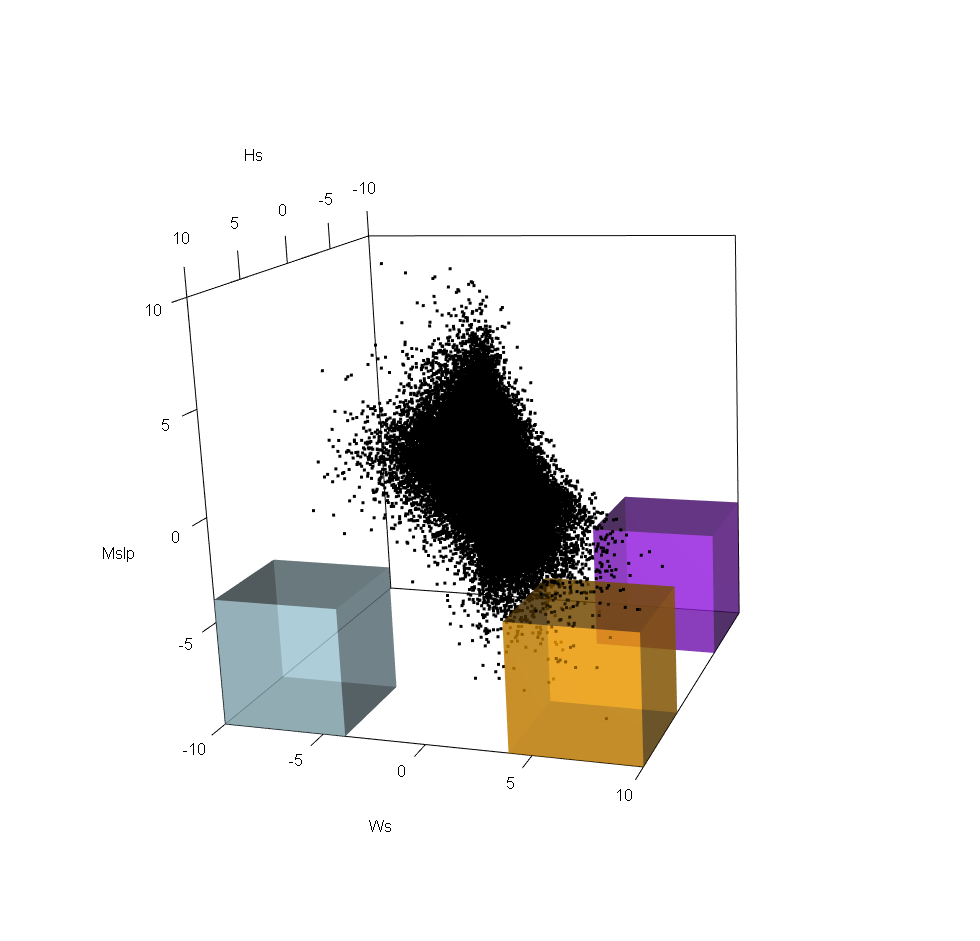}
         \caption{Hypercube regions considered for the three metocean variables. }
    \label{fig:hypercubes}
\end{figure}

\subsubsection{Joint distribution of \texttt{hs} along transect: $d=5$}
For the $d=5$ joint distribution of significant wave height \texttt{hs}, we cannot visualise the unit-level set estimate. However, there are 10 possible bivariate unit-level slices corresponding to different pairs of locations; three of these are illustrated in Figure~\ref{fig:subgauges_combined}. We again set $\epsilon = 0.01$ as this results in a reasonable amount of validation data for each slice. As previously, we observe that the estimated slices closely match the shapes of the corresponding bivariate sample cloud, indicating that the DeepGauge model can capture extremal dependence of \texttt{hs} across all five locations. Goodness-of-fit diagnostics, as illustrated in Figure~\ref{fig:qqplot_d5} of Appendix~\ref{appendx:supfigs}, show generally good agreement between model and empirical quantiles, albeit with slightly more divergence at extremely high quantiles relatively to the $d=3$ cases. This is especially true for the ADF diagnostic; however, this is not surprising considering the fact observations in each orthant are more sparse in higher dimensions (relative to $d=3$), and more data is generally required in such cases to ensure convergence to asymptotic forms. 

We also test the DeepGauge framework on a transect with $d=8$ locations; these results are detailed in the Supplementary Material. For this case, diagnostics again indicate generally good performance, albeit with a slightly reduction in quality compared to the lower dimensional cases. 

\begin{figure}[h]
\centering
    \includegraphics[width = \linewidth]{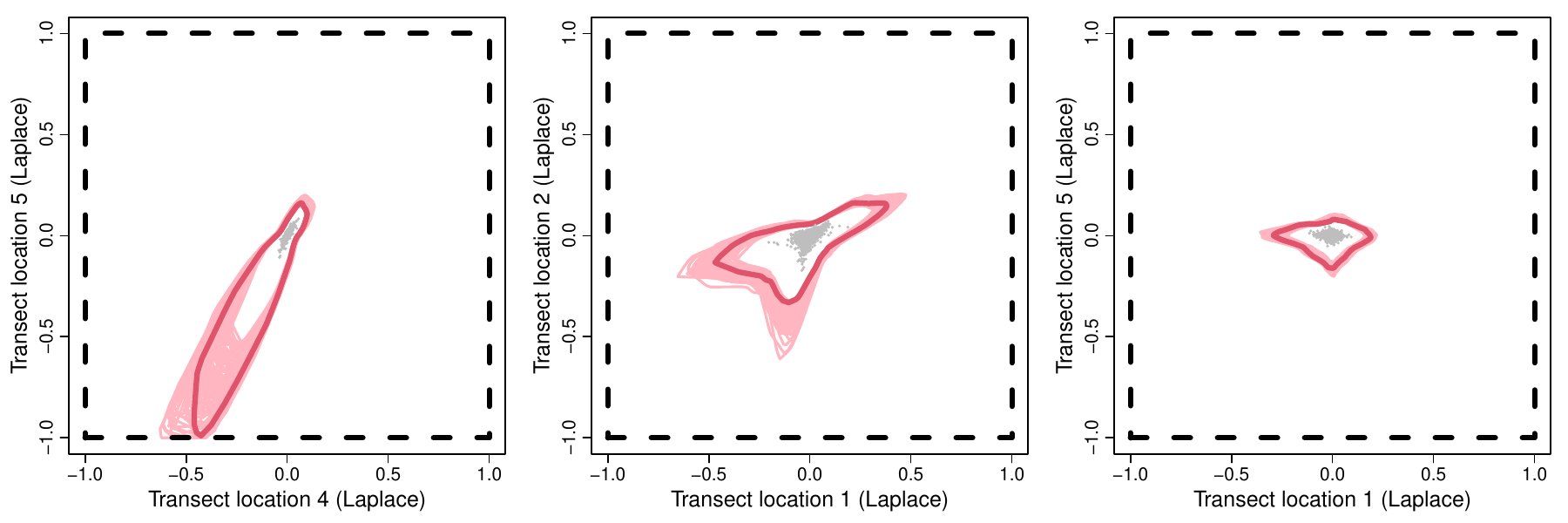}
         \caption{Scaled bivariate sample clouds with $\epsilon = 0.01$ for three pairwise combinations of the \texttt{hs} variable over the five locations on standard Laplace margins; the red lines describe the estimated bivariate unit-level set slices.}
    \label{fig:subgauges_combined}
\end{figure}

Overall, the proposed diagnostics indicate that the DeepGauge framework provides accurate model fits for metocean data sets of $d=3$, $d=5$, and $d=8$ dimensions. The estimated unit-level sets facilitate inference of the entire extremal dependence structure of the data, not only dependence in a single orthant of $\RR^d$, without the need for restrictive parametric models or low-order summary measures (see Section~\ref{sec:intro}). We also remark that our results for the cases of $d=5$ and $d=8$ are particularly encouraging: flexible models for extremal dependence are typically limited to $d\leq 3$, but we have shown that this limitation does not hold for the DeepGauge framework.

\section{Discussion}\label{sec:discussion}

We have introduced novel theoretical contributions for the geometric representation of multivariate extremes, which include extensions of the \citet{Wadsworth2013} model and a demonstration that many of the results introduced by \citet{Nolde2022} can be extended to data on Laplace margins. We further proposed the DeepGauge modelling framework and, through rigorous theoretical treatment and in contrast to existing models, we proved that our framework gives estimates that satisfy all of the required theoretical properties of unit-level sets. In practice, this results in consistent inference of extremal dependence and provides a means to estimate joint tail properties; we demonstrate this by evaluating the performance of the DeepGauge framework on simulated and observed data in Sections~\ref{sec:sim_study} and \ref{sec:app}, respectively. Unlike the majority of existing models for multivariate extremes, the DeepGauge framework is not limited to asymptotically dependent data nor low dimensions ($d \leq 4$). Thus, our approach represents a significant step towards flexible, non-parametric, and scalable models for multivariate extremes. 

We acknowledge that there is no theoretical guarantee that our proposed estimator for $\partial \mathcal{G}$ will converge to the true unit-level set. One could, for instance, try to prove consistency of the rescaled gauge function estimator, $\Tilde{g}(\cdot)$, for all angles in the set  $\dsphere$. However, theoretical results of this nature usually necessitate strict assumptions, which themselves can be difficult to verify. In practice, however, one can only ever look at diagnostics obtained from the data; we have therefore opted for a more practical treatment of our proposed estimator, and leave proofs of theoretical convergence for future work.

One noticeable omission from the DeepGauge framework is a model for the angular density $f_{\boldsymbol{W}}(\boldsymbol{w}), \boldsymbol{w} \in \dsphere$; see Section \ref{sec:method:est_ADF} for further discussion. Several existing approaches propose non-parametric techniques for this estimation, albeit in low dimensional ($d\leq 3$) settings \citep[e.g.,][]{Papastathopoulos2024,Murphy-Barltrop2024,Simpson2024}. Estimates of the angular density allow one to simulate from the model described in equation \eqref{eqn:trunc_gamma_assum}, from which one can apply Monte Carlo techniques to obtain probability estimates for joint tail regions of any form, i.e., more general than those illustrated in Figure \ref{fig:wt_model_probs}; see \citet{Wadsworth2024} for instance. The estimation of, and simulation from, the angular density function $f_{\boldsymbol{W}}(\boldsymbol{w})$ therefore represents an important area for future work in the context of multivariate extremes. 

Finally, given an estimate of $\partial \mathcal{G}$ for some $d \geq 3$, we remark that our modelling framework provides no means of obtaining lower dimensional unit-level sets, e.g., the unit-level set of $(X_{\mathfrak{i}},X_{\mathfrak{j}})$ for indices $ (\mathfrak{i},\mathfrak{j}),$ with $1\leq \mathfrak{i} < \mathfrak{j} \leq d$. This is because we only estimate the gauge function $g(\cdot)$ over $\dsphere$, and to obtain a lower-dimensional gauge function, such as $g_{(X_{\mathfrak{i}},X_{\mathfrak{j}})}(\cdot)$, would require us to minimise $g(\cdot)$ over all components in $\mathcal{D} \setminus \{\mathfrak{i},\mathfrak{j} \} $ \citep{Nolde2022}; this is only possible if we can evaluate $g(\cdot)$ on the whole of $\RR^d$. Future work could investigate techniques for extracting lower dimensional unit-level sets from higher dimensional estimates, which would avoid the need for refitting. 

\section*{Declarations}
\subsection*{Ethical Approval}
Not Applicable
\subsection*{Availability of supporting data}
The data sets analysed in Section \ref{sec:app} are available from the corresponding author upon reasonable request. 
\subsection*{Competing interests}
The authors have no relevant financial or non-financial interests to disclose.

\subsection*{Funding details} 
This work was supported by grants from the United States Geological Survey’s National Climate Adaptation Science Center (G24AC00197), and the National Science Foundation (DMS2152887).
\subsection*{Acknowledgments}
The authors would like to thank Phil Jonathan for providing the data, Ryan Campbell for access to code, and Lambert de Monte for helpful discussions.

\renewcommand{\theequation}{A.\arabic{equation}}
\renewcommand{\thefigure}{A\arabic{figure}}
\renewcommand{\thetable}{A\arabic{table}}
\renewcommand{\thesection}{A\arabic{section}}
\renewcommand\theHtable{Appendix.\thetable}
\renewcommand\theHfigure{Figure.\thefigure}

\setcounter{figure}{0}
\setcounter{table}{0}
\setcounter{equation}{0}
\setcounter{theorem}{0}
\begin{appendix}
\section*{Appendix}
\section{Neural network construction and inference}
\subsection{
Neural network representation
}
\label{appendix:NN1}
We here describe the construction of $r_\tau(\cdot)$ and $g(\cdot)$ used multi-layer perceptrons. Let $r_\tau(\boldsymbol{w}):=
\exp\{m_{\boldsymbol{\psi}}(\boldsymbol{w})\},$ where $m_{\boldsymbol{\psi}}(\boldsymbol{w})$ is a neural network, parameterised by $\boldsymbol{\psi}$, with $N\in\mathbb{N}$ hidden layers. Then $m_{\boldsymbol{\psi}}$ is the composition ${m_{\boldsymbol{\psi}}(\cdot)=m_{\boldsymbol{\psi}}^{(N+1)}\circ\dots\circ \mathbf{m}^{(2)}_{\boldsymbol{\psi}}\circ \mathbf{m}^{(1)}_{\boldsymbol{\psi}}(\cdot)}$, where, for $j=1,\dots,N$, the output from the $j$-th hidden layer, $\mathbf{m}^{(j)}_{\boldsymbol{\psi}}:\mathbb{R}^{n_{j-1}}\mapsto \mathbb{R}^{n_{j}}$, can be written as 
\begin{equation}
\label{Eq:MLP}
\boldsymbol{w}^{(j)}:=\mathbf{m}^{(j)}(\boldsymbol{w}^{(j-1)})=\text{ReLU}\left( \mathbf{M}^{(j)}\boldsymbol{w}^{(j-1)}+\mathbf{b}^{(j)}\right),
\end{equation}
where $\boldsymbol{w}^{(0)}:=\boldsymbol{w}$ is the input angles with dimension $n_{0}=d$ and the final layer is $m^{(N+1)}(\boldsymbol{w}^{(N)})=\mathbf{M}^{(N+1)}\boldsymbol{w}^{(N)}+\mathbf{b}^{(N)} \in \RR$, i.e., $n_{N+1} = 1$. The ReLU function evaluates the component-wise maxima of its input and a vector of zeroes (of suitable length), i.e., ReLU($\boldsymbol{x})=(\max\{x_1,0\},\max\{x_2,0\},\dots)^T$, where $\boldsymbol{x}=(x_1,x_2,\dots)^T$ is a vector of finite (unspecified) length. Each layer $j$, for $j=1,\dots,N+1,$ is parameterised by estimable weights matrices and vectors ${\mathbf{M}^{(j)}\in\mathbb{R}^{n_j\times n_{j-1}}}$ and $\mathbf{b}^{(j)}\in\mathbb{R}^{n_j}$, respectively. Thus, the neural network is fully-parameterised by the collection of parameters $\boldsymbol{\psi}=\{\{\mathbf{M}^{(j)},\mathbf{b}^{(j)}\}_{j=1}^{N+1}\}$. The dimension $n_j$ of the output of layer $\mathbf{m}^{(j)}$, as well as the ``depth'' $N$ of $m_{\boldsymbol{\psi}}(\cdot)$, are referred to as the neural network's architecture and are tunable hyper-parameters.

\subsection{Training and pre-training}\label{sec:methods:estimation:training}
The minimisation in \eqref{eq:emp_loss} is performed using the Adaptive Moment Estimation (Adam) algorithm \citep{kingma2014adam}, which is a variant of mini-batch stochastic gradient descent. The Adam algorithm is applied with its hyper-parameters at their default settings, and the models are trained using the \texttt{R} interface to \texttt{keras} \citep{kerasforR}. Prior to training, we uniformly-at-random partition the data into $80\%$ training and $20\%$ validation, with this partition consistent across the training stage for both $r_\tau(\boldsymbol{w})$ and $\tilde{g}(\boldsymbol{w})$. To mitigate overfitting of the neural network model, we train with checkpoints and early-stopping \citep{prechelt2002early}. Neural network models are optimised for a finite number of pre-specified iterations, say $M$, by minimising the loss function evaluated on the training data. The validation data are not used to optimise the neural network parameters. Instead, at every iteration of the training algorithm, we evaluate the loss function on the validation data. Then, the ``best fitting'' model is determined to be that which which minimises the loss function, evaluated on the validation data, across all epochs. Early-stopping ensures that the training scheme does not necessarily run for all $M$ iterations; training stops early if the validation loss has not decreased in the last $\Delta\in\mathbb{N}$ iterations. We also add $L_1$ and $L_2$ penalties to the estimable neural network parameters, $\psi$, to provide further model regularisation; throughout we adopt these penalties with equal shrinkage weight of $1\times10^{-4}$.  For further details on fitting and regularising deep neural network models, see \cite{goodfellow2016deep}.

Typically, without prior knowledge of the function that one seeks to approximate, the estimable parameter set, $\boldsymbol{\psi},$ of a neural network model is randomly initialised prior to training; thus when training our neural network model for the exceedance threshold $r_\tau(\cdot)$ we always choose to take this approach, as we have no prior knowledge of its structure. However, for training of the neural network that comprises the gauge function $g(\cdot)$, we can exploit some of the theoretical structure of gauge functions to perform pre-training, which is a popular technique used in applications of neural networks to reduce computation times. Pre-training uses parameter estimates of a neural network trained for one task as initial parameter estimates for a neural network model designed for a different task (with equivalent architecture; see \cite{goodfellow2016deep}, Ch. 8). Following Corollary~\ref{corol:all_properties}, we can use an estimate of the radial exceedance threshold, say $\hat{r}_\tau(\boldsymbol{w}),$ to obtain an initial estimate for the gauge function, say $\hat{g}_{\tau}(\boldsymbol{w})$; this is achieved by replacing the radial function $h(\cdot)$ in Corollary~\ref{corol:all_properties} with $\hat{r}_\tau(\cdot)$ and deriving the corresponding rescaled gauge function.
In the two dimensional setting, \citet{Wadsworth2024} use a similar idea to get an initial estimate for the unit-level set $\partial \mathcal{G}$. The authors also show that, for $\tau$ close to 1, the rescaled quantile set is a valid approximation of $\partial \mathcal{G}$. We use this idea to pre-train the neural network of the gauge function model $\tilde{g}(\cdot)$; we perform an initial optimisation of its weights subject to minimisation of the squared-error loss between the initial gauge function $\hat{g}_{\tau}(\boldsymbol{w})$ and the model output, $\tilde{g}(\boldsymbol{w})$. For $\alpha$, we always set its initial estimate to $d$, as this is the theoretical value attained by gauge functions for many popular copula models \citep{Wadsworth2024}. In unreported experiments, we found that this pre-training helps to mitigate numerical instability during training and improves accuracy of gauge function estimation. 

With the initial pre-trained weights, we estimate the rescaled gauge function, from which we obtain estimates of the unit-level set $\partial \mathcal{G}$. Selection of tuning parameters is discussed in Section~\ref{sec:sim_study} of the main text.

\section{Additional proofs}
\subsection{Proof of Proposition~\ref{prop:gauge_exp_laplace}}
\label{proof:prop:gauge_exp_laplace}
\begin{proof}
    First, observe that for any $X_E \sim \text{Exp}(1)$, we have that the variable given by
    \begin{equation*}
        X :=
            \begin{cases}
            \log(1 - e^{-X_E}) + \log(2), \hspace{1em} X_E \leq \log(2),\\
            X_E - \log(2), \hspace{5.57em} X_E > \log(2),
            \end{cases}
    \end{equation*}
    is standard Laplace distributed, where the $\log(2)$ here denotes the median of $X_E$. Given $\boldsymbol{w} \in \mathcal{S}^{d-1}_+$ we have, by definition, $g_E(\boldsymbol{w})=\lim_{t \to \infty} (-\log f_{\boldsymbol{X}_E}(t\boldsymbol{w})/t) $, where $f_{\boldsymbol{X}_E}(\cdot)$ denotes the continuous joint density function for $\boldsymbol{X}_E$ \citep{Nolde2022}. For any large $t$ satisfying $t > \max_{i=1,\dots,d} \{\log(2)/w_i\} > 0$, we have that 
    \begin{align*}
        \frac{-\log f_{\boldsymbol{X}_E}(t\boldsymbol{w})}{t} = \frac{-\log f_{\boldsymbol{X}}(t\boldsymbol{w} - {\log 2})}{t},
    \end{align*}
    where $f_{\boldsymbol{X}}(\cdot)$ is the continuous joint density function for $\boldsymbol{X}$ and the Jacobian of the transformation $\boldsymbol{X} \mapsto \boldsymbol{X}_E$ equals $1$ for the considered coordinates. For a fixed $\boldsymbol{w}$, set $t^* := || t\boldsymbol{w} - {\log 2} || = t|| \boldsymbol{w} - {\log 2}/t ||$ and {$\boldsymbol{w}^* := (t\boldsymbol{w} - {\log 2})/t^* = (\boldsymbol{w} - {\log 2}/t)/|| \boldsymbol{w} - {\log 2}/t ||$}. We then observe that $t^* \sim t$ and $\boldsymbol{w}^* \sim \boldsymbol{w}$ as $t \to \infty$. Therefore, it follows that
    \begin{equation*}
        \frac{-\log f_{\boldsymbol{X}}(t\boldsymbol{w} - {\log 2})}{t} = \frac{-\log f_{\boldsymbol{X}}(t^*\boldsymbol{w}^*)}{t} \sim \frac{-\log f_{\boldsymbol{X}}(t\boldsymbol{w})}{t}.
    \end{equation*}
    Taking the limit as $t \to \infty$, the result follows. 
\end{proof}

\subsection{Proof of Proposition~\ref{prop:limit_set_properties}}
\label{proof:prop:limit_set_properties}
\begin{proof}
    To prove the star-shaped property, it suffices to show that for any $\boldsymbol{x} \in \mathcal{H}$, we have $t\boldsymbol{x} \in \mathcal{H}$ for all $t \in [0,1]$. This is trivial for $\boldsymbol{x} = \boldsymbol{0}_d$. Taking $\boldsymbol{x} \in \mathcal{H}\setminus \{ \boldsymbol{0}_d\}$, we see that 
    \begin{equation*}
       ||t\boldsymbol{x}|| = t||\boldsymbol{x}|| \leq th(\boldsymbol{x}/||\boldsymbol{x}||) \leq h(\boldsymbol{x}/||\boldsymbol{x}||) = h(t\boldsymbol{x}/||t\boldsymbol{x}||),
    \end{equation*}
    and thus $t\boldsymbol{x} \in {G}$. Moreover, given $\boldsymbol{x} \in \mathcal{H} \setminus \{ \boldsymbol{0}_d\}$, we have 
    \begin{equation*}
        ||\boldsymbol{x}|| \leq h(\boldsymbol{x}/||\boldsymbol{x}||) \leq \frac{1}{\big\| \boldsymbol{x}/||\boldsymbol{x}|| \big\|_{\infty}} = \frac{||\boldsymbol{x}||}{|| \boldsymbol{x}||_{\infty}},
    \end{equation*}
    so $|| \boldsymbol{x}||_{\infty} \leq 1$, implying $-1 \leq x_i \leq 1$ for all $i = 1,\hdots,d$; thus, $\mathcal{H} \subset [-1,1]^d$. This also implies that $\mathcal{H}$ is bounded; thus, to prove $\mathcal{H}$ is compact, we just need to show that $\mathcal{H}$ is closed, owing to the Heine–Borel theorem \citep[see, e.g.,][]{Hayes1956}. Let $(\boldsymbol{x}_n)_{n \in \NN}$ denote a sequence in $\mathcal{H}$ such that $\boldsymbol{x}_n \to \boldsymbol{x}$ as $n \to \infty$, i.e., $||\boldsymbol{x}_n - \boldsymbol{x}|| \to 0$. We must show that $\boldsymbol{x} \in \mathcal{H}$. If $\boldsymbol{x} = \boldsymbol{0}_d$, the proof is trivial. Thus, we assume $\boldsymbol{x} \neq \boldsymbol{0}_d$. Given $\epsilon > 0$, there exists $N_1 \in \NN$ such that $||\boldsymbol{x}_n - \boldsymbol{x}|| < \epsilon$ for all $n \geq N_1$. Furthermore, since ${h}(\cdot)$ is continuous, there must exist $N_2 \in \NN$ such that $|{h}(\boldsymbol{x}_n/||\boldsymbol{x}_n||) - {h}(\boldsymbol{x}/||\boldsymbol{x}||) | < \epsilon$ for all $n \geq N_2$. Setting $N:= \max\{N_1,N_2 \}$, we have that 
    \begin{align*}
        ||\boldsymbol{x}|| &= ||\boldsymbol{x} - \boldsymbol{x}_n + \boldsymbol{x}_n ||\leq ||\boldsymbol{x} - \boldsymbol{x}_n|| + ||\boldsymbol{x}_n ||< \epsilon + {h}(\boldsymbol{x}_n/||\boldsymbol{x}_n||) < 2\epsilon + {h}(\boldsymbol{x}/||\boldsymbol{x}||).
    \end{align*}
    Taking the limit as $\epsilon \to 0$, we have $\boldsymbol{x} \in \mathcal{H}$; thus, $\mathcal{H}$ is compact.
\end{proof}
\subsection{Proof of Proposition~\ref{prop:rescaling}}
\label{proof:prop:rescaling}
To prove Proposition~\ref{prop:rescaling}, we require Lemma~\ref{lem:kappa_bijective}, which we first prove below.
\begin{proof}[Proof of Lemma \ref{lem:kappa_bijective}]
    To prove bijectivity, it suffices to show that $\kappa$ is injective and surjective. 

    \noindent Injectivity: suppose $\kappa(\boldsymbol{w}) = \kappa(\boldsymbol{w}^*)$ for $\boldsymbol{w},\boldsymbol{w}^* \in \dsphere$, and define the constant
    \begin{equation*}
        c:= \frac{ \Big\lVert  \left(\frac{w_1}{b_1(w_1)}, \hdots, \frac{w_d}{b_d(w_d)} \right) \Big \rVert }{ \Big\lVert  \left(\frac{w^*_1}{b_1(w^*_1)}, \hdots, \frac{w^*_d}{b_d(w^*_d)} \right) \Big \rVert } \in \RR_+.
    \end{equation*}
    If $\kappa(\boldsymbol{w}) = \kappa(\boldsymbol{w}^*)$, then for each $i = 1,\hdots,d$, we have $w_i = cw_i^*$. Furthermore, since $\boldsymbol{w},\boldsymbol{w}^* \in \dsphere$, we see that $c=c|| \boldsymbol{w}^*|| = || c\boldsymbol{w}^*|| = || \boldsymbol{w}|| = 1$; thus, $\boldsymbol{w} = \boldsymbol{w}^*$. 

    \noindent Surjectivity: let $\boldsymbol{w} \in \dsphere$ and without loss of generality, assume that $w_d \neq 0$. Setting
    \begin{equation*}
        \kappa^{-1}(\boldsymbol{w}) := \left( \frac{a(\boldsymbol{w})b_1(w_1)w_1}{b_d(w_d)|w_d|}, \hdots,\frac{a(\boldsymbol{w})b_{d-1}(w_{d-1})w_{d-1}}{b_d(w_d)|w_d|}, a(\boldsymbol{w}) \text{sgn}(w_d) \right), 
    \end{equation*}
    where $a(\boldsymbol{w}) := 1/\sqrt{1 + \sum_{j=1}^{d-1} ( b_j(w_j)w_j/b_d(w_d)w_d )^2 }$, it is straightforward to verify that $\kappa(\kappa^{-1}(\boldsymbol{w})) = \boldsymbol{w}$, completing the proof. 
\end{proof}

\begin{proof}[Proof of Proposition \ref{prop:rescaling}]
    First, observe that we can rewrite $\widetilde{\partial \mathcal{H}}$ as 
    \begin{equation*}
        \widetilde{\partial \mathcal{H}} = \left\{ \boldsymbol{w} h(\kappa^{-1}(\boldsymbol{w}))  \left\lVert  \left(\frac{\kappa^{-1}(\boldsymbol{w})_1}{b_1(\kappa^{-1}(\boldsymbol{w})_1)}, \hdots, \frac{\kappa^{-1}(\boldsymbol{w})_d}{b_d(\kappa^{-1}(\boldsymbol{w})_d)} \right) \right\rVert \bigg\vert \boldsymbol{w} \in \dsphere \right\},
    \end{equation*}
    where $\kappa^{-1}(\boldsymbol{w})_i$ denotes the $i$-th component of $\kappa^{-1}(\boldsymbol{w})$ for each $i = 1,\hdots,d$. Since $\kappa(\cdot)$ is bijective, there exists a unique, bijective mapping between $\partial \mathcal{H}$ and $\widetilde{\partial \mathcal{H}}$; thus, the sets are in one-to-one correspondence. Furthermore, by definition, we have $-1 \leq w_ih(\boldsymbol{w})/b_i(w_i) \leq 1$ for each $i = 1,\hdots, d$ and $\boldsymbol{w} \in \dsphere$; thus $\widetilde{\partial \mathcal{H}} \subset [-1,1]^d$. Finally, considering $i =1,\dots,d$ and setting $\boldsymbol{w}^{u,i} = \argmax_{\boldsymbol{w} \in \dsphere}\{w_ih(\boldsymbol{w}) \}$ and $\boldsymbol{w}^{l,i} = \argmin_{\boldsymbol{w} \in \dsphere}\{w_ih(\boldsymbol{w}) \}$, we have $w^{u,i}_ih(\boldsymbol{w}^{u,i})/b_i(w^{u,i}_i) = 1$ and  $w^{l,i}_ih(\boldsymbol{w}^{l,i})/b_i(w^{l,i}_i) = -1$, implying the componentwise maxima and minima equal $\boldsymbol{1}_d$ and $-\boldsymbol{1}_d$, respectively.
\end{proof}
\subsection{Proof of Corollary~\ref{corol:rescaled_gauge_estimate}}\label{proof:corol:rescaled_gauge_estimate}

\begin{proof}
    By definition, we have that $\boldsymbol{w}/\Tilde{g}(\boldsymbol{w}) \in \widetilde{\partial \mathcal{H}} \subset [-1,1]^d$ for any $\boldsymbol{w} \in \dsphere$. Applying similar reasoning to the proof of Proposition~\ref{prop:g_lowerbound}, we immediately see that $\Tilde{g}(\boldsymbol{w})\geq || \boldsymbol{w} ||_{\infty}$. Furthermore, taking any $i =1,\dots,d$, we have
    \begin{align*}
        ||\kappa(\boldsymbol{w}^{u,i})||_{\infty} &= \left\lVert \frac{  \left(\frac{w^{u,i}_1}{b_1(w^{u,i}_1)}, \hdots, \frac{w^{u,i}_d}{b_d(w^{u,i}_d)} \right) }{ \Big\lVert  \left(\frac{w^{u,i}_1}{b_1(w^{u,i}_1)}, \hdots, \frac{w^{u,i}_d}{b_d(w^{u,i}_d)} \right) \Big \rVert } \right\rVert_{\infty} = \frac{ \left\lVert h(\boldsymbol{w}^{u,i})\left(\frac{w^{u,i}_1}{b_1(w^{u,i}_1)}, \hdots, \frac{w^{u,i}_d}{b_d(w^{u,i}_d)} \right) \right\rVert_{\infty} }{ \left\lVert  h(\boldsymbol{w}^{u,i})\left(\frac{w^{u,i}_1}{b_1(w^{u,i}_1)}, \hdots, \frac{w^{u,i}_d}{b_d(w^{u,i}_d)} \right)\right \rVert \; \;} \\
        &= \frac{ 1 }{ \Big\lVert  h(\boldsymbol{w}^{u,i})\left(\frac{w^{u,i}_1}{b_1(w^{u,i}_1)}, \hdots, \frac{w^{u,i}_d}{b_d(w^{u,i}_d)} \right) \Big \rVert } = \Tilde{g}(\kappa(\boldsymbol{w}^{u,i})). 
    \end{align*}
    The same reasoning holds for $\boldsymbol{w}^{l,i}$ with $i = 1,\hdots,d$.
\end{proof}

\subsection{Proof of Proposition~\ref{prop:wads_tawn_exp_laplace}}
\begin{proof}
\label{proof:prop:wads_tawn_exp_laplace}
    Given any large $u$ satisfying $u > \max_{i=1,\dots,d} \{\log(2)/w_i\} > 0$, and applying similar reasoning as in the proof of Proposition~\ref{prop:gauge_exp_laplace}, we have that 
    \begin{align*}
        \Pr\left(\min_{i \in \mathcal{D}}\{X_{E,i}/w_{i}\}>u\right) &= \Pr\left(X_{E,i}> w_{i}u, i = 1,\hdots,d\right) \\
        &= \Pr\left(X_{i} > w_{i}u - \log 2, i = 1,\hdots,d\right) \\
        &\sim \Pr\left(X_{i} > w_{i}u, i = 1,\hdots,d\right) \\
        &= \Pr\left(\min_{i \in \mathcal{D}}\{X_{i}/w_{i}\}>u\right).
    \end{align*}
    Taking $u \to \infty$, we have 
    \begin{equation*}
        \Pr\left(\min_{i \in \mathcal{D}}\{X_{i}/w_{i}\}>u\right) = \Pr\left(\min_{i \in \mathcal{D}}\{X_{E,i}/w_{i}\}>u\right) = L(e^u;\boldsymbol{w})e^{-\lambda(\boldsymbol{w})u},
    \end{equation*}
    proving the statement. 
\end{proof}

\subsection{Proof of Proposition~\ref{prop:extended_ADF}}\label{proof:prop:extended_ADF}
\begin{proof}
    Given $\boldsymbol{w}  \in \mathcal{S}^{d-1}\setminus \mathcal{A} $, set $\boldsymbol{c}:=\text{sgn}(\boldsymbol{w}) = (\text{sgn}(w_1),\hdots,\text{sgn}(w_d))^T$. Then, we have that 
    \begin{align*}
        \Pr\left(\min_{i \in \mathcal{D}}\{X_{i}/w_{i}\}>u\right) &= \Pr\left(\min_{i \in \mathcal{D}}\{c_iX_{i}/c_iw_{i}\}>u\right) \to L_{\boldsymbol{c}\boldsymbol{X}}(e^u;\boldsymbol{c}\boldsymbol{w})e^{-\lambda_{\boldsymbol{c}\boldsymbol{X}}(\boldsymbol{c}\boldsymbol{w})u},
    \end{align*}
    as $u \to \infty$, where $L_{\boldsymbol{c}\boldsymbol{X}}(\cdot ;\boldsymbol{w})$ and $\lambda_{\boldsymbol{c}\boldsymbol{X}}(\cdot)$ denote the slowly varying function and ADF for $\boldsymbol{c}\boldsymbol{X}$ as in equation \eqref{eqn:wads_tawn}, respectively. The result follows with $\mathcal{L}(\cdot;\boldsymbol{w}) = L_{\boldsymbol{c}\boldsymbol{X}}(\cdot;\boldsymbol{c}\boldsymbol{w})$ and $\Lambda(\boldsymbol{w}) = \lambda_{\boldsymbol{c}\boldsymbol{X}}(\boldsymbol{c}\boldsymbol{w})$. 
\end{proof}

\subsection{Proof of Proposition~\ref{prop:extended_ADF_boundary_link}}
\label{proof:prop:extended_ADF_boundary_link}
\begin{proof}
    Let $\boldsymbol{c}=\text{sgn}(\boldsymbol{w}) \in \{-1,1\}^d$, and define $\boldsymbol{c}\boldsymbol{X}$ and $\boldsymbol{c}\boldsymbol{w} \in \dsphere_+$ as in Proposition \ref{prop:trans_gauge}, noting that $||\boldsymbol{c}\boldsymbol{w} ||_{\infty} = ||\boldsymbol{w} ||_{\infty}$. Further, let $\partial \mathcal{G}_{\boldsymbol{c}\boldsymbol{X}}$ and $g_{\boldsymbol{c}\boldsymbol{X}}(\cdot)$ denote the unit-level set and gauge function, respectively, for $\boldsymbol{c}\boldsymbol{X}$. For any $\mathfrak{r} \in \{ \mathfrak{r} \in [0,1] : \mathfrak{r}\mathcal{R}_{\boldsymbol{c}\boldsymbol{w}} \cap \partial \mathcal{G}_{\boldsymbol{c}\boldsymbol{X}} \neq \emptyset \}$, there must exist some $j \in \{1,\dots,d\}$ and $\boldsymbol{w}^* \in \dsphere_+$ such that 
    \begin{equation*}
        \mathfrak{r} \frac{c_jw_j}{||\boldsymbol{c}\boldsymbol{w}||_{\infty}} = \frac{w^*_j}{g_{\boldsymbol{c}\boldsymbol{X}}(\boldsymbol{w}^*)}.
    \end{equation*}
    By Proposition \ref{prop:trans_gauge}, $g_{\boldsymbol{c}\boldsymbol{X}}(\boldsymbol{w}^*) = g(\boldsymbol{w}^*/\boldsymbol{c})$, with $\boldsymbol{w}^*/\boldsymbol{c} := (w^*_1/c_1,\hdots,w^*_d/c_d)^T \in \mathcal{S}^{d-1}\setminus \mathcal{A}$ and $\text{sgn}(\boldsymbol{w}) = \text{sgn}(\boldsymbol{w}^*/\boldsymbol{c})$. Thus, we have 
    \begin{equation*}
        \mathfrak{r} \frac{w_j}{||\boldsymbol{w}||_{\infty}} = \frac{w^*_j/c_j}{g(\boldsymbol{w}^*/\boldsymbol{c})}, 
    \end{equation*}
    implying $\mathfrak{r} \in \{ \mathfrak{r} \in [0,1] : \mathfrak{r}\Tilde{\mathcal{R}}_{\boldsymbol{w}}\cap \partial \mathcal{G} \neq \emptyset \}$. Analogous reasoning shows that \[\mathfrak{r} \in \{ \mathfrak{r} \in [0,1] : \mathfrak{r}\Tilde{\mathcal{R}}_{\boldsymbol{w}}\cap \partial \mathcal{G} \neq \emptyset \} \Rightarrow \mathfrak{r} \in \{ \mathfrak{r} \in [0,1] : \mathfrak{r}\mathcal{R}_{\boldsymbol{c}\boldsymbol{w}} \cap \partial \mathcal{G}_{\boldsymbol{c}\boldsymbol{X}} \neq \emptyset \},\] giving $\{ \mathfrak{r} \in [0,1] : \mathfrak{r}\Tilde{\mathcal{R}}_{\boldsymbol{w}}\cap \partial \mathcal{G} \neq \emptyset \} = \{ \mathfrak{r} \in [0,1] : \mathfrak{r}\mathcal{R}_{\boldsymbol{c}\boldsymbol{w}} \cap \partial \mathcal{G}_{\boldsymbol{c}\boldsymbol{X}} \neq \emptyset \}$ and $\Tilde{\mathfrak{r}}_{\boldsymbol{w}} = \max \{ \mathfrak{r} \in [0,1] : \mathfrak{r}\Tilde{\mathcal{R}}_{\boldsymbol{w}}\cap \partial \mathcal{G} \neq \emptyset \} = \mathfrak{r}_{\boldsymbol{c}\boldsymbol{w}}$. Applying Proposition 3.3 of \citet{Nolde2022}, alongside Proposition \ref{prop:gauge_exp_laplace}, we see that 
    \begin{align*}
        \Lambda(\boldsymbol{w}) = \lambda_{\boldsymbol{c}\boldsymbol{X}}(\boldsymbol{c}\boldsymbol{w})= ||\boldsymbol{c}\boldsymbol{w}||_{\infty} \times \mathfrak{r}_{\boldsymbol{c}\boldsymbol{w}}^{-1} = ||\boldsymbol{w}||_{\infty} \times \Tilde{\mathfrak{r}}_{\boldsymbol{w}}^{-1},
    \end{align*}
    completing the proof. 
\end{proof}
An illustration of Proposition \ref{prop:extended_ADF_boundary_link} is given in Figure \ref{fig:adf_proof}. 

\subsection{Proof of Corollary~\ref{corol:gauge_adf_bound}}\label{proof:corol:gauge_adf_bound}
\begin{proof}
    Let $\boldsymbol{w}/g(\boldsymbol{w}) \in \partial \mathcal{G}$ denote coordinates on the unit-level set. Recalling the definitions of $\Tilde{\mathfrak{r}}_{\boldsymbol{w}}$ and $\Tilde{\mathcal{R}}_{\boldsymbol{w}}$, we must have that $\Tilde{\mathfrak{r}}_{\boldsymbol{w}}|w_i|/||\boldsymbol{w}||_{\infty} \geq |w_i|/g(\boldsymbol{w})$ for all $i = 1,\hdots,d$; to see this, consider the two possibilities $\boldsymbol{w}/g(\boldsymbol{w})\in \{\Tilde{\mathfrak{r}}_{\boldsymbol{w}}\Tilde{\mathcal{R}}_{\boldsymbol{w}} \cap \partial \mathcal{G}\}$ and $\boldsymbol{w}/g(\boldsymbol{w})\notin \{\Tilde{\mathfrak{r}}_{\boldsymbol{w}}\Tilde{\mathcal{R}}_{\boldsymbol{w}} \cap \partial \mathcal{G}\}$ in turn. 
    Consequently, 
    \begin{equation*}       
         \frac{\Tilde{\mathfrak{r}}_{\boldsymbol{w}}}{||\boldsymbol{w}||_{\infty}} \geq \min_{i=1,\dots,d}\left\{ \frac{|w_i|}{g(\boldsymbol{w}) |w_i|} \right\} = \frac{1}{g(\boldsymbol{w})}.
    \end{equation*}
    Recalling from Proposition \ref{prop:extended_ADF_boundary_link} that $||\boldsymbol{w}||_{\infty} \Lambda(\boldsymbol{w})^{-1} = \Tilde{\mathfrak{r}}_{\boldsymbol{w}}$, the first inequality follows. The second inequality follows directly from Proposition \ref{prop:extended_ADF_boundary_link}.
\end{proof}

\subsection{Algorithm for computing the extended ADF diagnostic} \label{appen:extend_ADF}

Consider a sample $\{\boldsymbol{x}_j\}_{j=1}^n$ of independent copies of  $\boldsymbol{X}$, and let $q$ denote quantile level close to 1. Given any angle $\boldsymbol{w} \in \dsphere \setminus \mathcal{A}$, the extended ADF diagnostic is computed using Algorithm \ref{alg:adf}. We then compare the resulting quantiles, $(e_1,\hdots,e_b)^T$, to the corresponding theoretical quantiles using a QQ plot.
\begin{algorithm}
    \caption{Computing the extended ADF diagnostic.}
    \label{alg:adf}
    \begin{algorithmic}
        \For{$k \gets 1$ to $n$}
            \State $t_{\boldsymbol{w}}^k \gets \min_{i = 1,\hdots,d} \{ x_{k,i}/w_{i} \}$
        \EndFor
        \State $\mathbf{t}_{\boldsymbol{w}} \gets (t_{\boldsymbol{w}}^1,\hdots,t_{\boldsymbol{w}}^n)$
        \State $\hat{u}_{\boldsymbol{w}} \gets \texttt{quantile}(\mathbf{t}_{\boldsymbol{w}},q)$

        \State $b \gets 1$ 
        \For{$k \gets 1$ to $n$}
            \If{$t_{\boldsymbol{w}}^k > \hat{u}_{\boldsymbol{w}}$}
                \State $e_b \gets \Hat{\Lambda}(\boldsymbol{w})(t_{\boldsymbol{w}}^k - \hat{u}_{\boldsymbol{w}})$
                \State $b \gets b+1$
            \EndIf
        \EndFor
    \end{algorithmic}
\end{algorithm}

\section{Supplementary figures and tables} \label{appendx:supfigs}

  \begin{figure}
    \centering
    \includegraphics[width=0.6\textwidth]{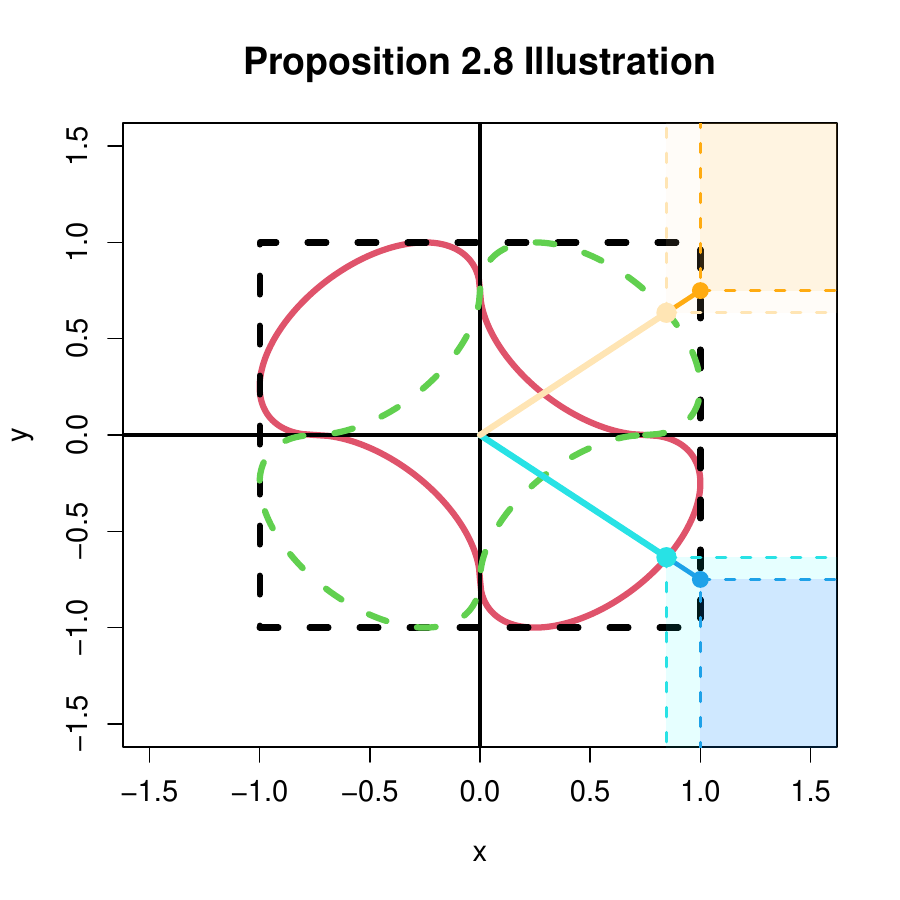}
    \caption{Plot illustrating the result described in Proposition \ref{prop:extended_ADF_boundary_link}. Here, the red line denotes the unit-level set $\partial \mathcal{G}$ of a Gaussian copula with $\rho = -0.5$. Moreover, the blue and cyan points, and regions, illustrate the re-scaling to obtain $\Tilde{\mathfrak{r}}_{\boldsymbol{w}}$ for the angle $\boldsymbol{w} = (0.8,-\sqrt{1-0.8^2})$. Setting $\boldsymbol{c} = \text{sgn}(\boldsymbol{w}) = (1,-1)$, the green dashed line denotes the unit-level set $\partial \mathcal{G}$ of $\boldsymbol{c}\boldsymbol{X}$, with the corresponding scaling procedure for the angle $\boldsymbol{c}\boldsymbol{w} = (0.8,\sqrt{1-0.8^2})$ illustrated by the orange points and regions. One can observe that the rescaling procedures are analogous in both orthants.} 
    \label{fig:adf_proof}
\end{figure}

  \begin{figure}
\centering
    \includegraphics[width=\linewidth]{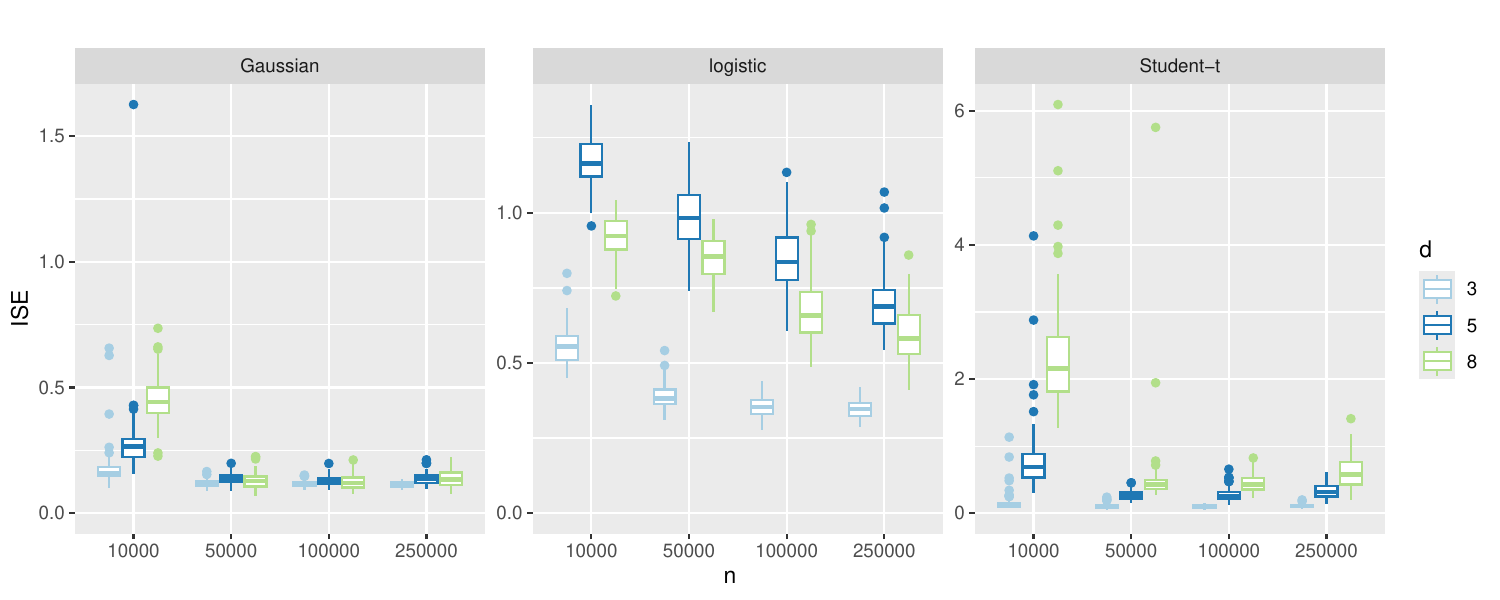}
    \includegraphics[width=\linewidth]{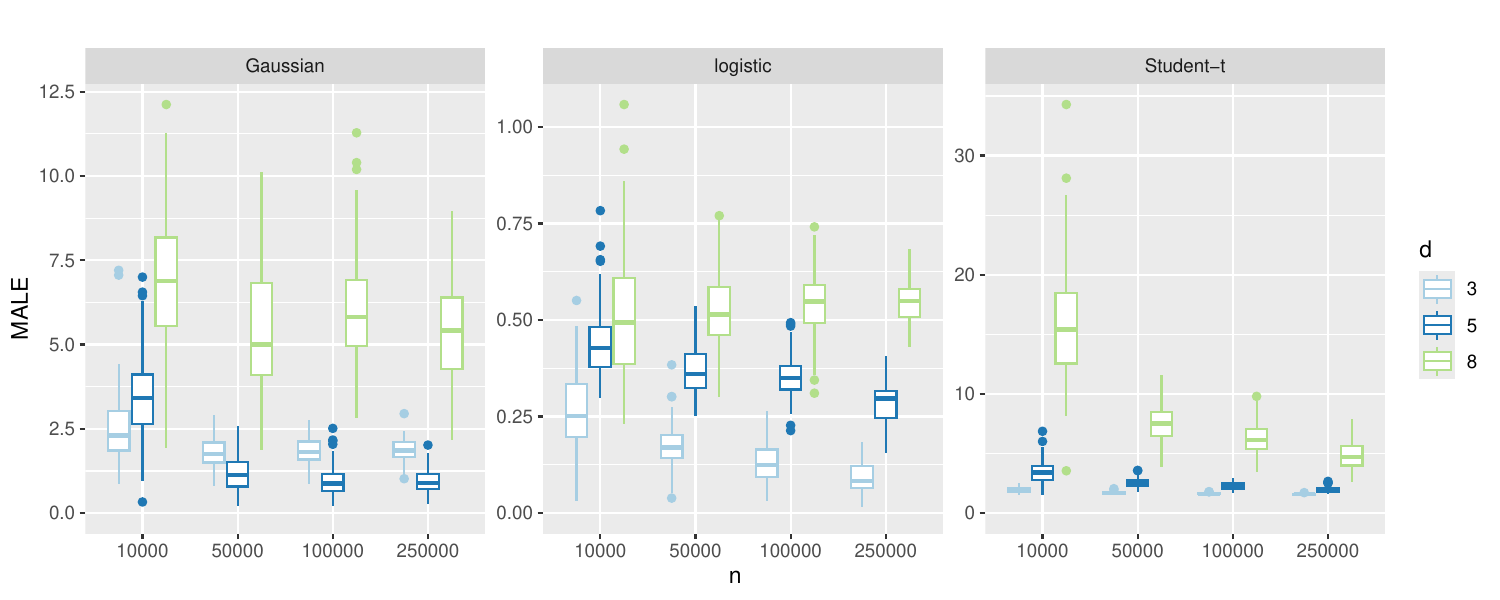}
    \caption{Boxplots of the estimated ISE and MALE for the varying sample size, $n$, and dimension, $d$, simulation study described in Section~\ref{sec:sim_study:models}. Here the radial quantile level is $\tau=0.75$ and the architecture for the rescaled gauge function has $N=3$ hidden layers (each with width 64).}
    \label{fig:sim_study3}
\end{figure}

\begin{landscape}
\begin{table}[t!]
\caption{Estimates of the median ($2.5\%$, $97.5\%$ percentile) of the ISE and MALE for the simulation study of the architecture for the gauge function $g(\cdot)$, described in Section~\ref{sec:sim_study:models}. Here the sample size is $n=100,000$, the radial quantile level is $\tau=0.75$, and the considered dimensions are $d=5$ and $d=8$.}
\label{tab:simstudy2}
 \centering
 \small
  \begin{tabular}{cccccc} 
 &&\multicolumn{2}{c}{$d=5$}&\multicolumn{2}{c}{$d=8$}\\
 \cline{3-6}
Copula&$g(\cdot)$&ISE&MALE&ISE&MALE\\
\hline
\multirow{8}{*}{Gaussian} & $(16)$ & 0.44 (0.27, 1.87)  & 4.58 (2.85, 9.03)& 2.52 (0.80. 5.67) & 14.4 (6.48, 20.8)\\
 & $(16,16)$ & 0.16 (0.13, 0.23) & 1.43 (0.50, 2.81)  & 0.28 (0.14, 0.93)& 3.12 (1.02, 8.28)\\
 & $(16,16,16)$ & 0.14 (0.11, 0.19)  & 1.11 (0.39, 2.27) & 0.19 (0.11, 0.62) & 3.87 (1.71, 6.83)\\
& $(16,16,16,16)$ & 0.14 (0.11, 0.22)  & 1.10 (0.35, 2.52) & 0.23 (0.11, 0.61)& 3.61 (1.19, 7.20)\\
 & $(64)$ & 0.19 (0.16, 0.25) & 1.88 (0.98, 3.02) & 0.33 (0.22, 0.49) & \textbf{2.44 }(1.27, 4.49)\\

 & $(64,64)$ & 0.14 (0.11, 0.20) & 0.91 (0.28, 1.87) & 0.16 (0.11, 0.24)& 4.52 (2.31, 7.13)\\
& $(64,64,64)$ & 0.13 (0.10, 0.17)  & \textbf{0.88} (0.36, 1.80) & 0.11 (0.08, 0.19) & 5.94 (3.58, 9.14)\\
 & $(64,64,64,64)$ & \textbf{0.12} (0.10, 0.18) & 0.97 (0.37, 2.04)  & \textbf{0.11} (0.08, 0.17) & 6.44 (3.13, 8.83)\\

\hline
\multirow{8}{*}{Logistic}& $(16)$ & 0.49 (0.40, 0.58)  & 0.35 (0.28, 0.45)& \textbf{0.36} (0.26, 0.53) & \textbf{0.53} (0.36, 0.66)\\
 & $(16,16)$ & 0.67 (0.53, 0.91) & 0.33 (0.24, 0.44)& 0.52 (0.38, 0.69)  & 0.54 (0.39, 0.69)\\
 & $(16,16,16)$ & 0.78 (0.62, 1.04) & \textbf{0.32} (0.16, 0.44) & 0.62 (0.50, 0.81) & 0.54 (0.40, 0.66)\\

 & $(16,16,16,16)$ & 0.85 (0.65, 1.03) & 0.32 (0.21, 0.45) & 0.66 (0.52, 0.90) & 0.54 (0.40, 0.69)\\
 & $(64)$ & \textbf{0.48} (0.39, 0.64)  & 0.34 (0.28, 0.45) & 0.37 (0.31, 0.54)  & 0.56 (0.43, 0.68)\\

 & (64,64) & 0.69 (0.56, 1.01)  & 0.35 (0.29, 0.47)  & 0.57 (0.45, 0.78) & 0.54 (0.42, 0.67)\\
 & (64,64,64) & 0.84 (0.66, 1.07)  & 0.35 (0.28, 0.46) & 0.66 (0.53, 0.87)  & 0.55 (0.36, 0.68)\\
 & (64,64,64,64) & 0.91 (0.73, 1.07) & 0.35 (0.23, 0.45) & 0.78 (0.62, 0.98) & 0.54 (0.39, 0.67)\\

\hline
\multirow{8}{*}{Student-t} & (16) & 12.0 (6.92, 15.0) & \textbf{1.73} (1.59, 1.97)  & 25.0 (17.7, 28.8) & 3.74 (2.60, 5.40)\\
 & (16,16) & 4.39 (0.97, 7.51)  & 1.75 (1.58, 2.38)  & 16.2 (11.3, 21.1)  & \textbf{2.97} (2.36, 3.97)\\

 & (16,16,16) & 0.62 (0.25, 3.21) & 2.09 (1.70, 2.72)  & 9.44 (4.05, 15.3) & 3.06 (2.40, 4.70)\\

 & (16,16,16,16) & 0.43 (0.20, 2.07) & 2.24 (1.77, 2.78) & 6.42 (2.90, 13.1) & 3.22 (2.43, 5.98)\\
 & (64) & 2.43 (1.33, 4.41)  & 1.76 (1.61, 2.12)  & 6.76 (4.36, 14.0) & 3.13 (2.30, 4.83)\\

 & (64,64) & 0.30 (0.15, 0.64)  & 2.07 (1.64, 2.86) & 1.20 (0.42, 3.83) & 4.71 (2.80, 6.28)\\
& (64,64,64) & \textbf{0.26} (0.15, 0.47)  & 2.26 (1.76, 2.80) & 0.43 (0.25, 0.68) & 6.28 (4.40, 9.00)\\

 & (64,64,64,64) & {0.26} (0.18, 0.45) & 2.30 (1.87, 2.98) & \textbf{0.40} (0.25, 0.70) & 6.42 (4.13, 9.47)\\
\hline
 \end{tabular}
\normalsize
  \end{table}
  \end{landscape}

\begin{figure}
\centering
    \includegraphics[scale=.52]{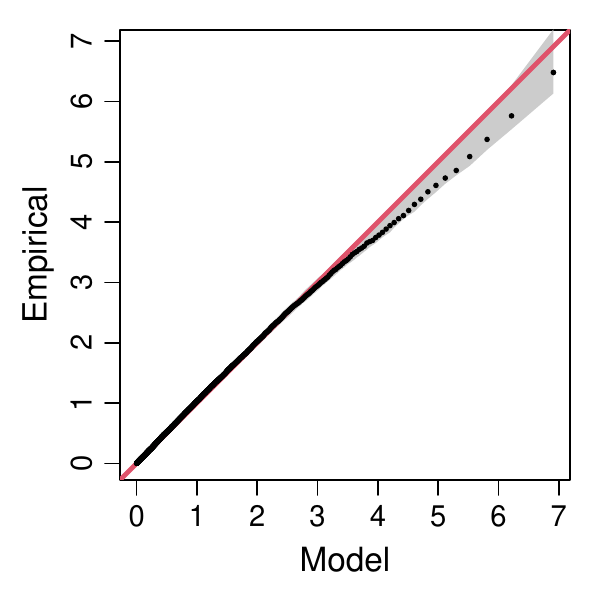}
    \includegraphics[scale=.52]{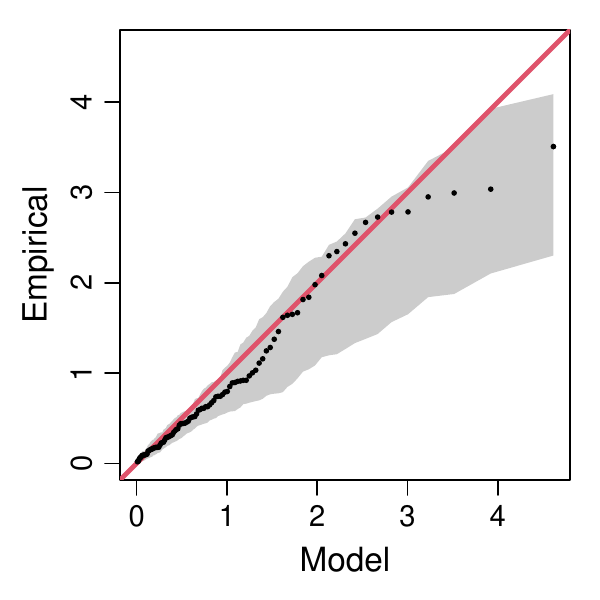}
    \includegraphics[scale=.52]{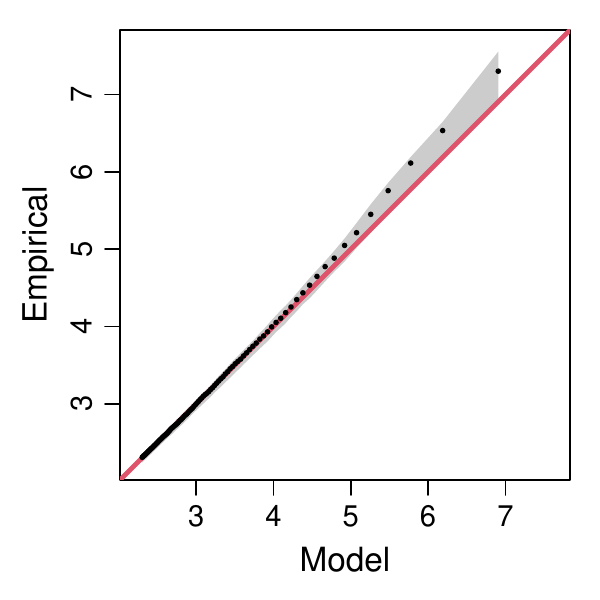}
         \caption{\textbf{Goodness-of-fit diagnostics for location 01}: Truncated gamma QQ plot (left), the extended ADF diagnostic at $\boldsymbol{w} = \boldsymbol{1}_d/||\boldsymbol{1}_d||$ (centre), and return level set probability estimates (right) for the DeepGauge model fitted to \texttt{hs}, \texttt{ws}, and \texttt{mslp}.}
    \label{fig:qqplot_1}
\end{figure}

\begin{figure}
\centering
    \includegraphics[scale=.52]{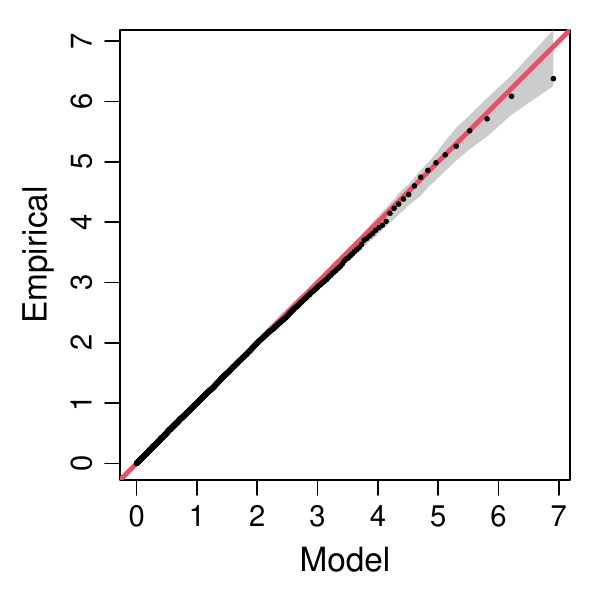}
    \includegraphics[scale=.52]{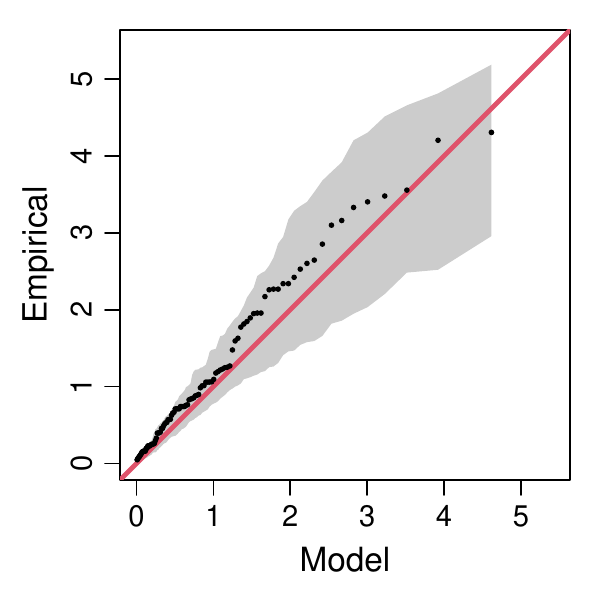}
    \includegraphics[scale=.52]{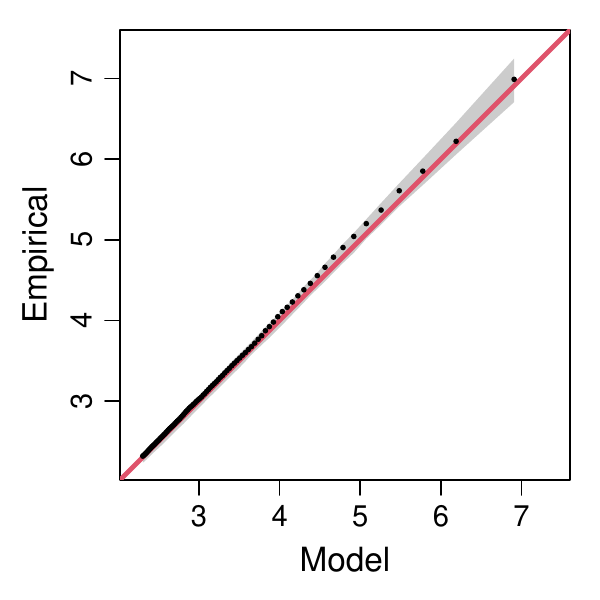}
         \caption{\textbf{Goodness-of-fit diagnostics for location 46}: Truncated gamma QQ plot (left), the extended ADF diagnostic at $\boldsymbol{w} = \boldsymbol{1}_d/||\boldsymbol{1}_d||$ (centre), and return level set probability estimates (right) for the DeepGauge model fitted to \texttt{hs}, \texttt{ws}, and \texttt{mslp}.}
    \label{fig:qqplot_46}
\end{figure}

\begin{figure}
\centering
    \includegraphics[scale=.52]{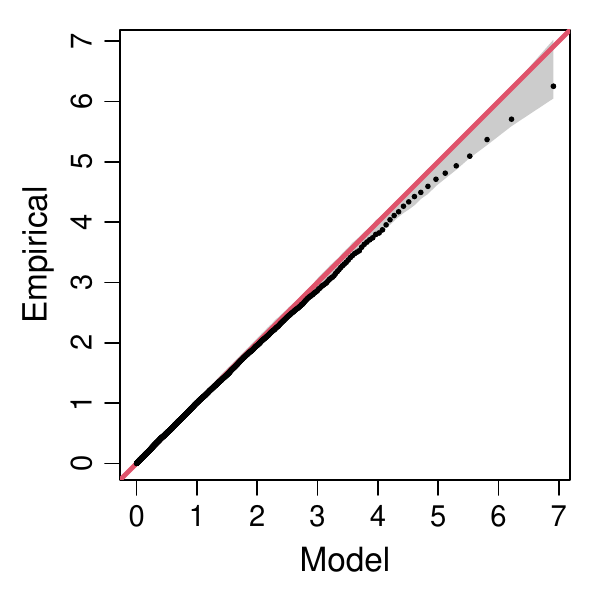}
    \includegraphics[scale=.52]{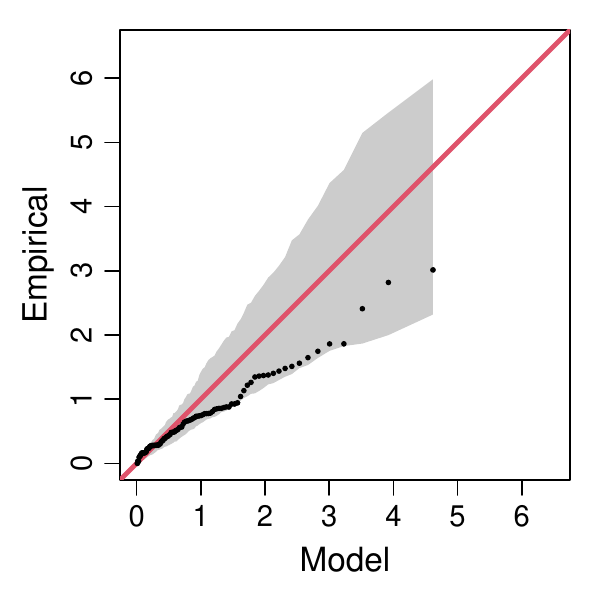}
    \includegraphics[scale=.52]{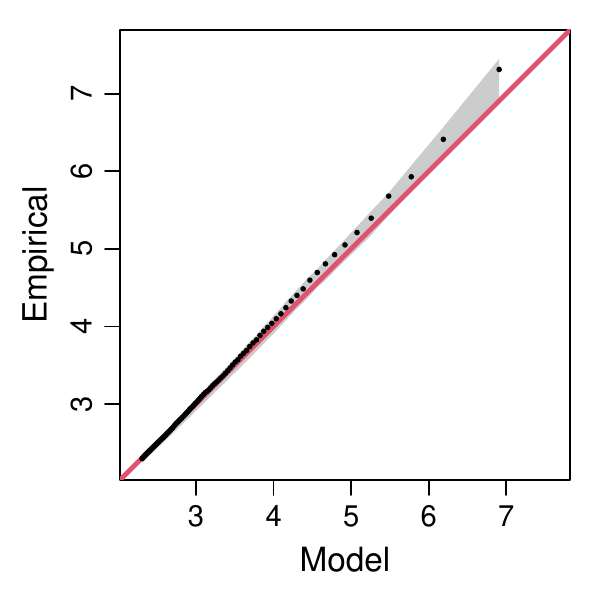}
         \caption{\textbf{Goodness-of-fit diagnostics for location 85}: Truncated gamma QQ plot (left), the extended ADF diagnostic at $\boldsymbol{w} = \boldsymbol{1}_d/||\boldsymbol{1}_d||$ (centre), and return level set probability estimates (right) for the DeepGauge model fitted to \texttt{hs}, \texttt{ws}, and \texttt{mslp}.}
    \label{fig:qqplot_85}
\end{figure}

\begin{figure}
\centering
    \includegraphics[scale=.4]{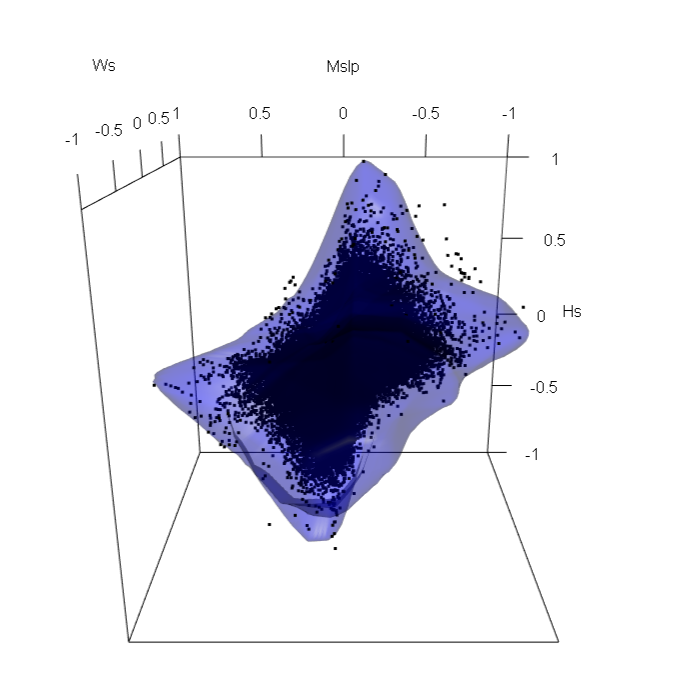}
    \includegraphics[scale=.4]{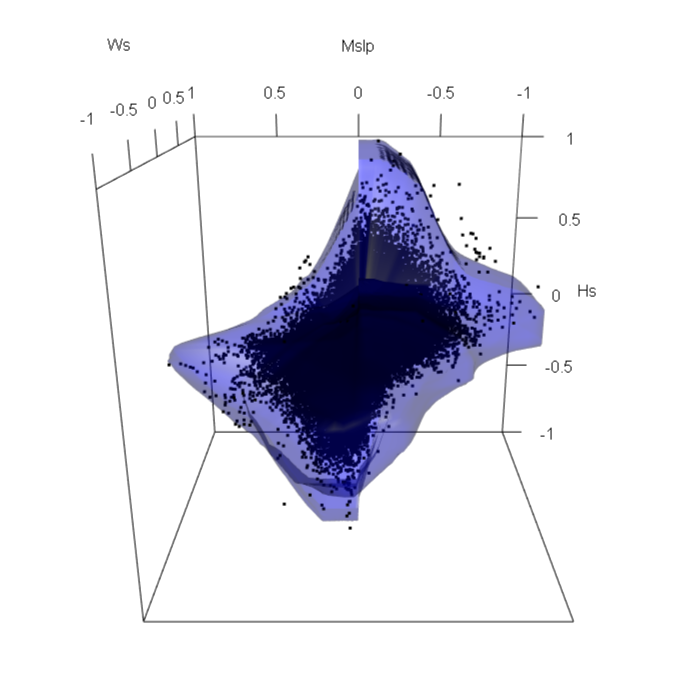}
         \caption{The estimated unit-level (left) and extended ADF (right) sets for \texttt{hs}, \texttt{ws}, and \texttt{mslp} at location 46 (on standard Laplace margins)}
    \label{fig:estimates_46}
\end{figure}
\begin{figure}
\centering
    \includegraphics[scale=.4]{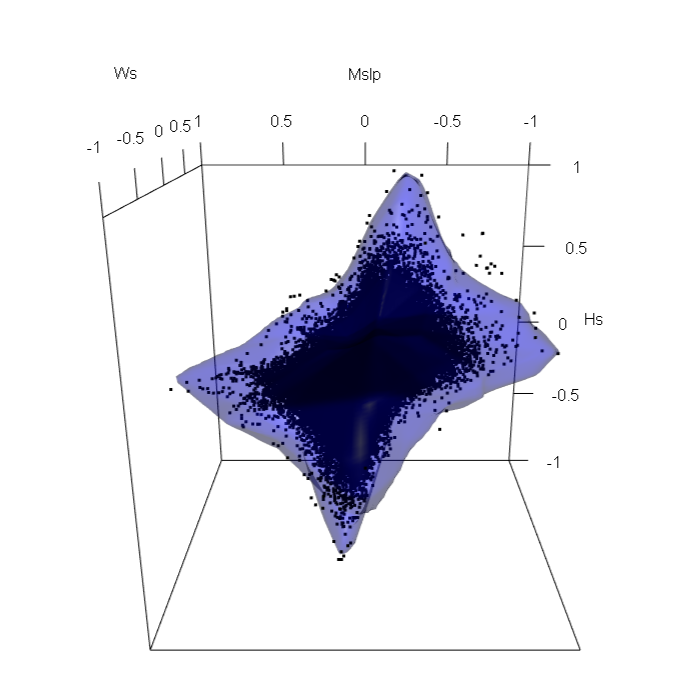}
    \includegraphics[scale=.4]{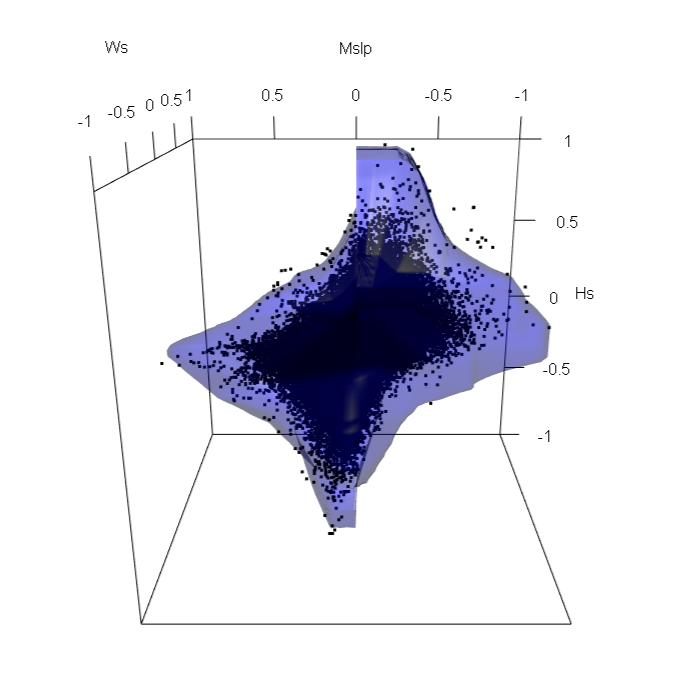}
         \caption{The estimated unit-level (left) and extended ADF (right) sets for \texttt{hs}, \texttt{ws}, and \texttt{mslp} at location 85 (on standard Laplace margins)}
    \label{fig:estimates_85}
\end{figure}

\begin{figure}
\centering
    \includegraphics[scale=.5]{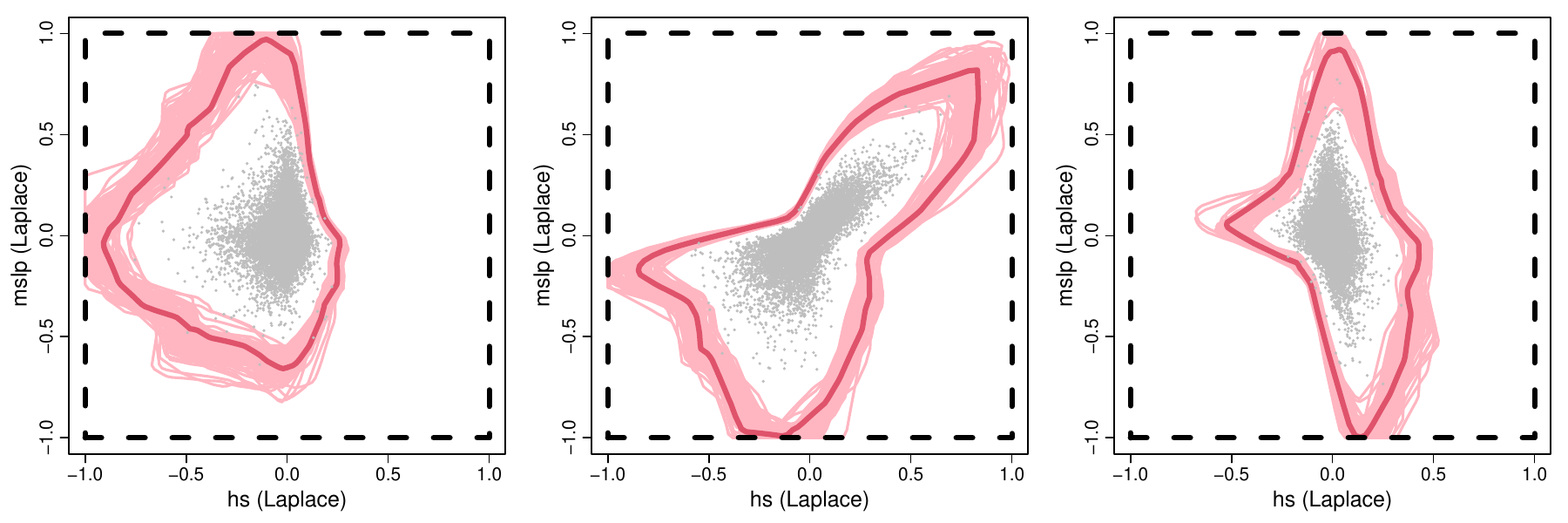}
         \caption{Scaled bivariate sample clouds with $\epsilon = 0.01$ for pairwise combinations of \texttt{ws}, \texttt{hs}, and \texttt{mslp} at location 46 on standard Laplace margins; the red lines describe the corresponding estimated bivariate unit-level set slices.}
    \label{fig:subgauges_46}
\end{figure}
\begin{figure}
\centering
    \includegraphics[scale=.5]{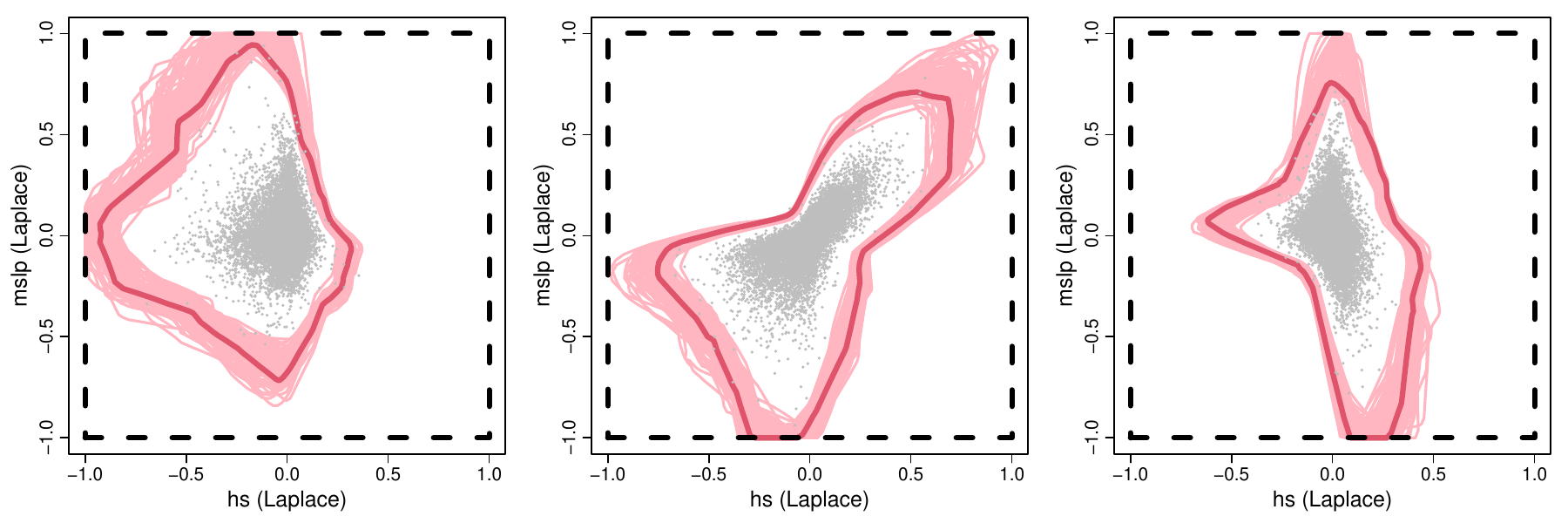}
         \caption{Scaled bivariate sample clouds with $\epsilon = 0.01$ for pairwise combinations of \texttt{ws}, \texttt{hs}, and \texttt{mslp} at location 85 on standard Laplace margins; the red lines describe the corresponding estimated bivariate unit-level set slices.}
    \label{fig:subgauges_85}
\end{figure}
\begin{figure}
    \centering
    \includegraphics[width=0.5\linewidth]{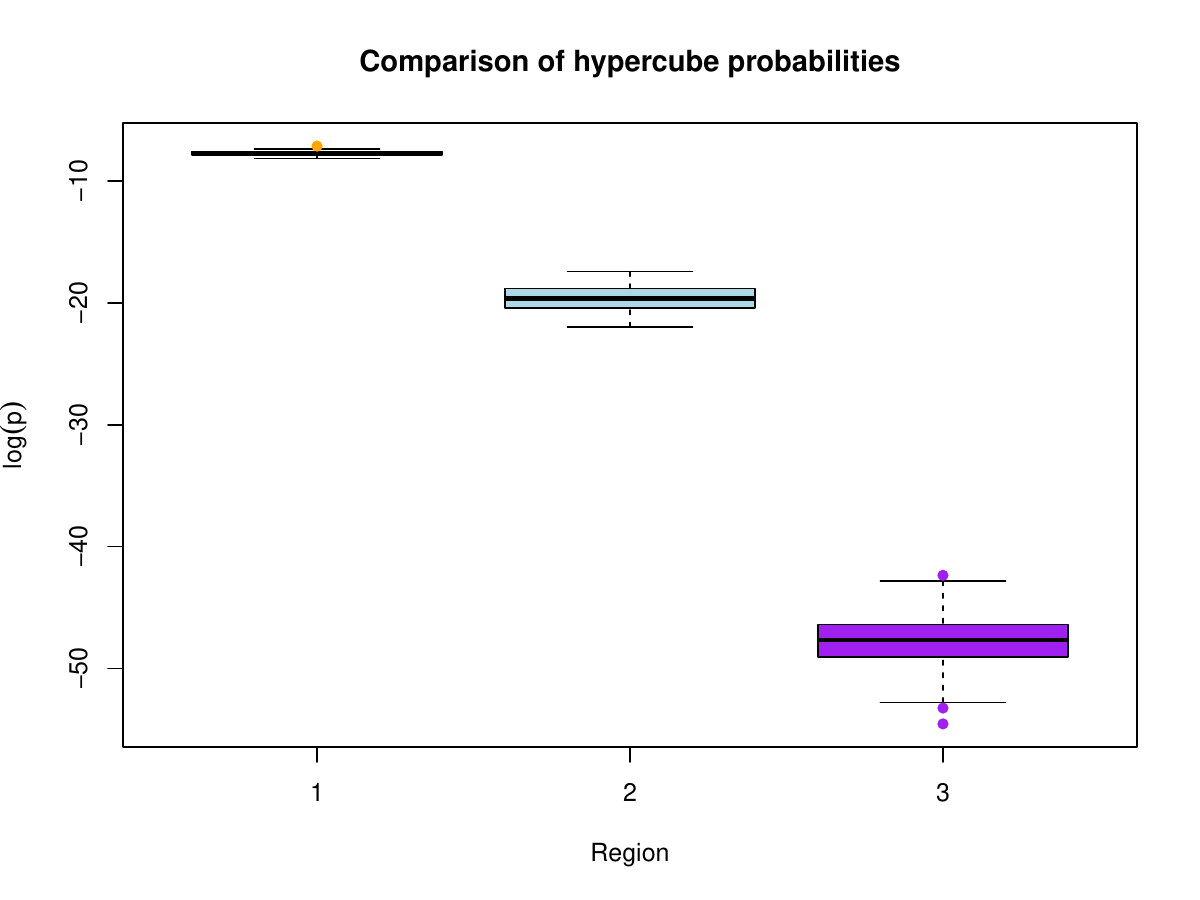}
    \caption{Bootstrapped log-probability estimates for each hypercube at location $01$. The colouring of each boxplot corresponds to the regions in Figure \ref{fig:hypercubes}.}
    \label{fig:hc_probs_1}
\end{figure}

\begin{figure}
    \centering
    \includegraphics[width=0.5\linewidth]{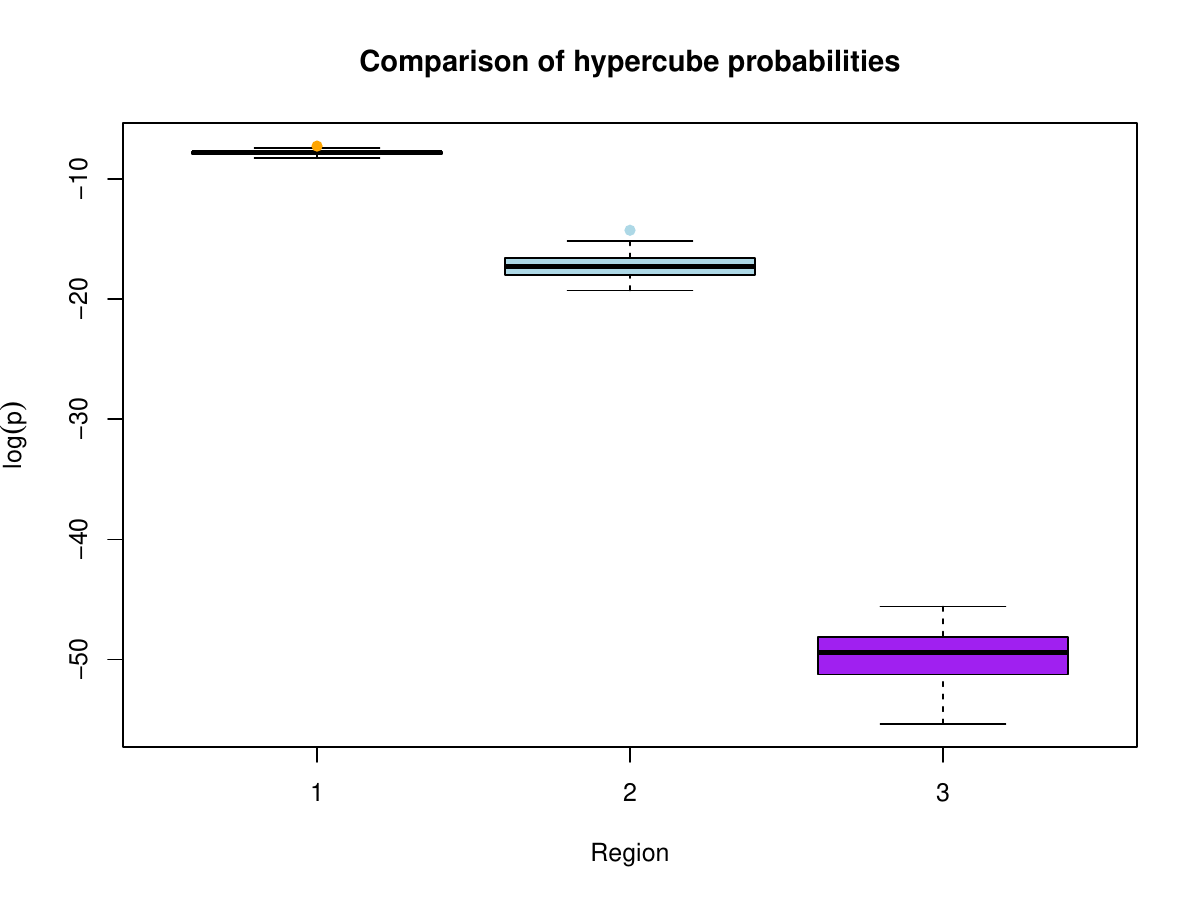}
    \caption{Bootstrapped log-probability estimates for each hypercube at location $46$. The colouring of each boxplot corresponds to the regions in Figure \ref{fig:hypercubes}.}
    \label{fig:hc_probs_46}
\end{figure}

\begin{figure}
    \centering
    \includegraphics[width=0.5\linewidth]{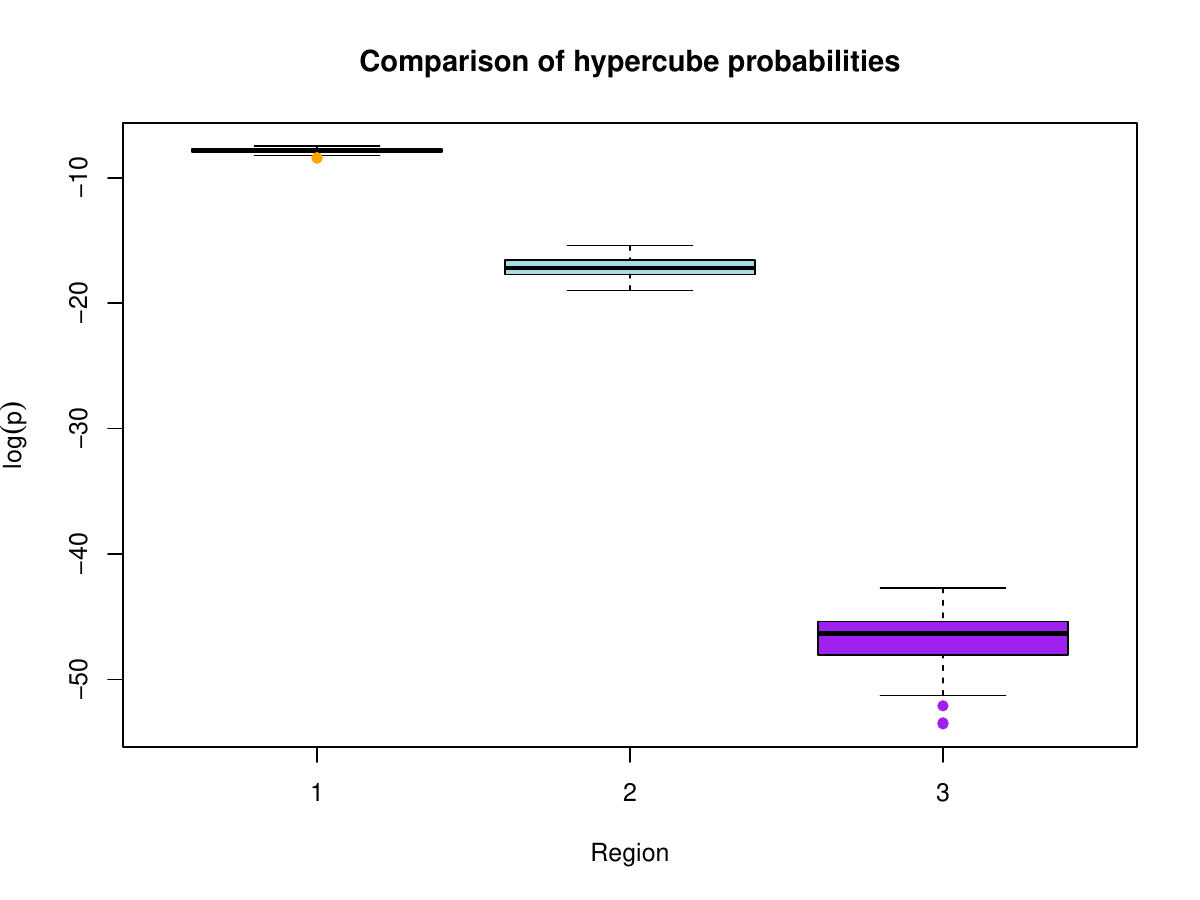}
    \caption{Bootstrapped log-probability estimates for each hypercube at location $85$. The colouring of each boxplot corresponds to the regions in Figure \ref{fig:hypercubes}.}
    \label{fig:hc_probs_85}
\end{figure}

\begin{figure}
\centering
    \includegraphics[scale=.52]{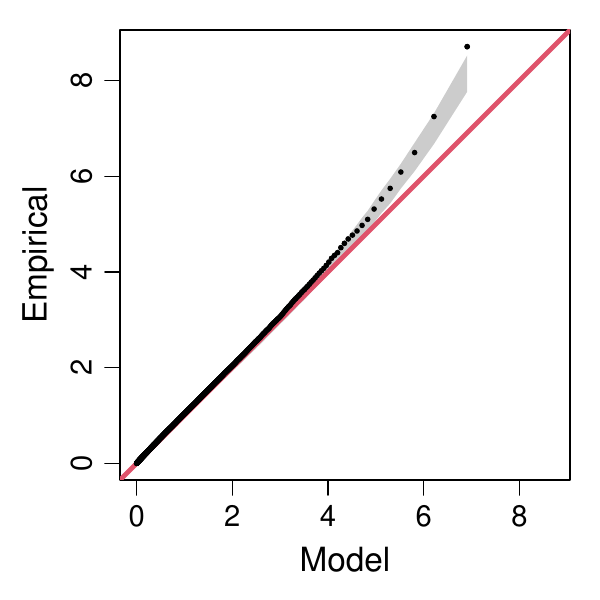}
    \includegraphics[scale=.52]{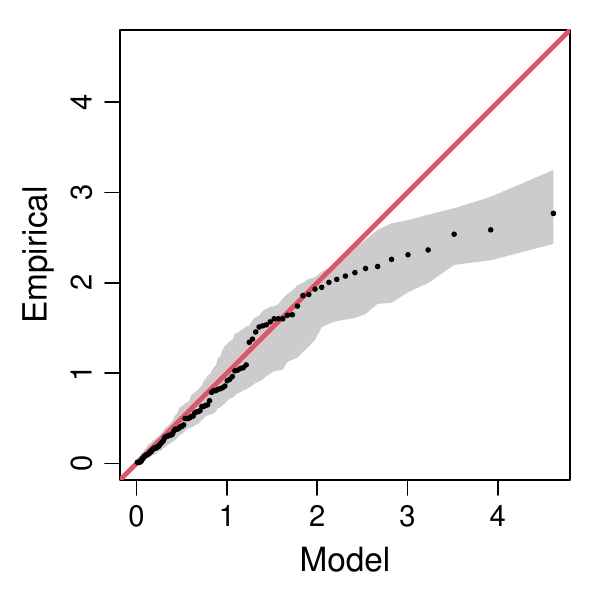}
    \includegraphics[scale=.52]{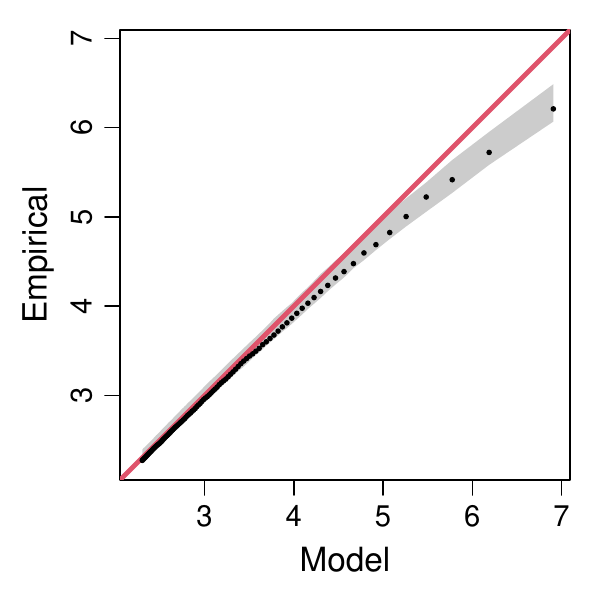}
         \caption{\textbf{Goodness-of-fit diagnostics for \texttt{hs}}: Truncated gamma QQ plot (left), the extended ADF diagnostic at $\boldsymbol{w} = \boldsymbol{1}_d/||\boldsymbol{1}_d||$ (centre), and return level set probability estimates (right) for the DeepGauge model fitted to wave height data at five locations.}
    \label{fig:qqplot_d5}
\end{figure}

\newpage

\section{Case study with eight dimensional transect} \label{appendx:d_8_case}

To further demonstrate the efficacy of the DeepGauge framework, we also apply the model to wave height data across $d=8$ locations in the transect; these locations are illustrates in Figure~\ref{fig:sites_d_8}. We remark that very few approaches for multivariate extremes are available to dimensions of this size \citep{Lederer2023}. 

\begin{figure}[t!]
\centering
    \includegraphics[scale=0.6]{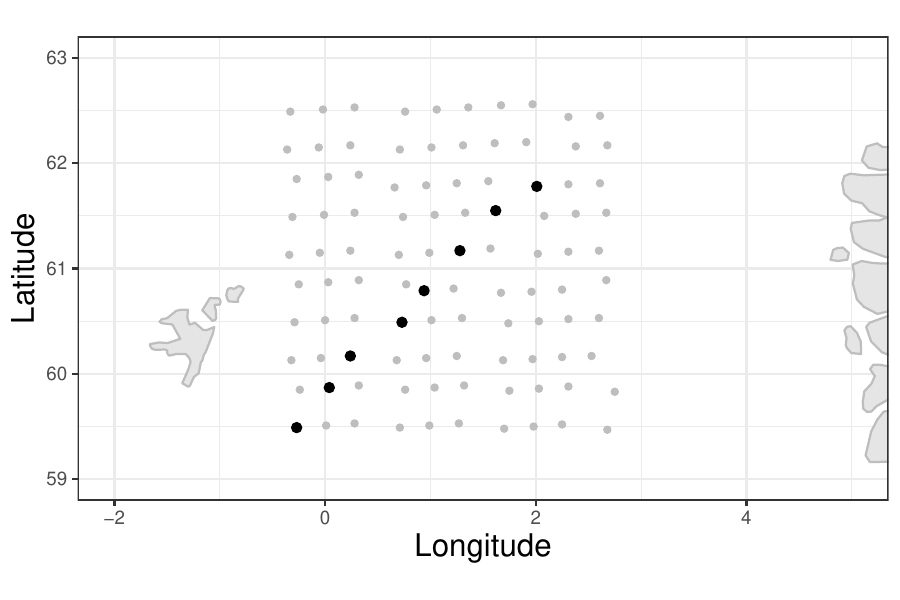}
        \caption{Locations used to study severe ocean events. Locations where we study the joint extremal dependence of \texttt{ws}, \texttt{hs}, and \texttt{mslp} are circled in blue. Locations for studying the joint extremal dependence in \texttt{hs} are highlighted in black.}
    \label{fig:sites_d_8}
\end{figure}
The diagnostics for the $d = 8$ case are illustrated in Figures~\ref{fig:subgauges_combined_d_8} and \ref{fig:qqplot_d_8}. One can observe generally good agreement across all metrics, indicating the model fits are reasonable. We observe that in a similar manner to the $d = 5$ case, there is slightly more divergence at extremely high levels of the quantile plots compared to the $d = 3$ examples. 

\begin{figure}[h]
\centering
    \includegraphics[width = \linewidth]{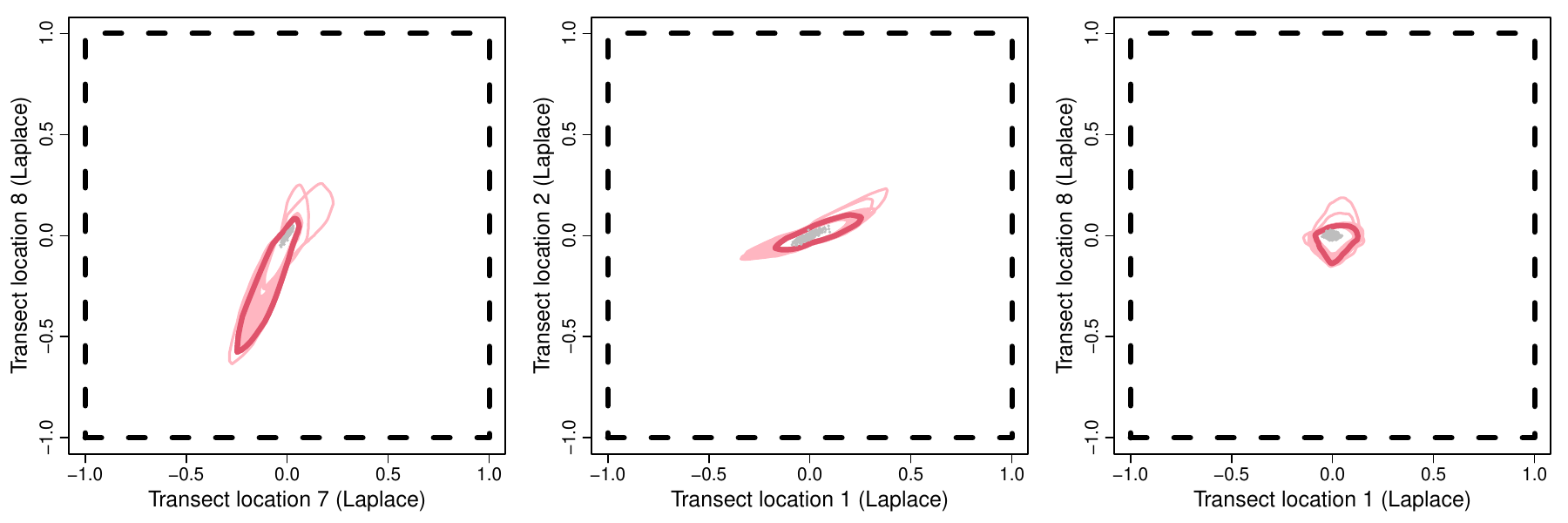}
         \caption{Scaled bivariate sample clouds with $\epsilon = 0.015$ for three pairwise combinations of the \texttt{hs} variable over the eight locations on standard Laplace margins; the red lines describe the estimated bivariate unit-level set slices.}
    \label{fig:subgauges_combined_d_8}
\end{figure}

\begin{figure}
\centering
    \includegraphics[scale=.52]{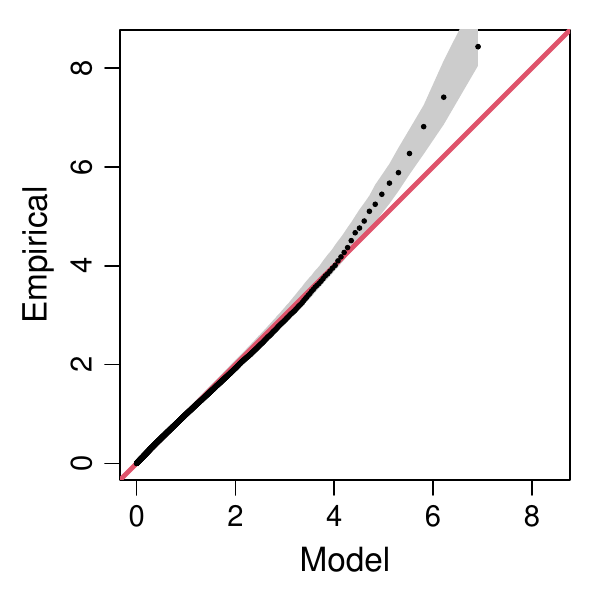}
    \includegraphics[scale=.52]{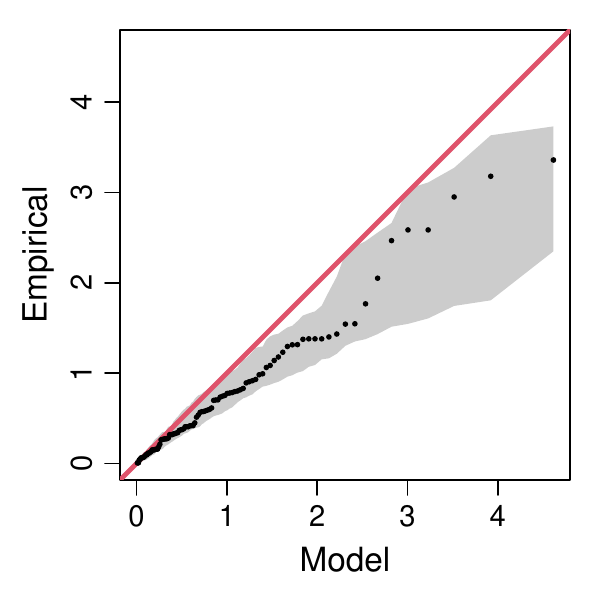}
    \includegraphics[scale=.52]{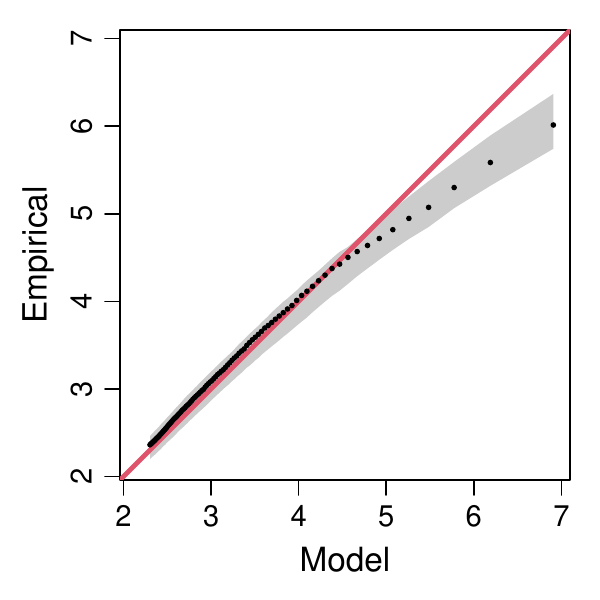}
         \caption{\textbf{Goodness-of-fit diagnostics for \texttt{hs}}: Truncated gamma QQ plot (left), the extended ADF diagnostic at $\boldsymbol{w} = \boldsymbol{1}_d/||\boldsymbol{1}_d||$ (centre), and return level set probability estimates (right) for the DeepGauge model fitted to wave height data at eight locations.}
    \label{fig:qqplot_d_8}
\end{figure}

\end{appendix}
\clearpage
\bibliographystyle{myapalike}
\bibliography{library, tooting}
\end{document}